\documentclass{article}
\usepackage{geometry}
\usepackage[margin=0.5cm,font=small]{caption}
\usepackage{wrapfig}
\usepackage{amssymb, amsmath, dsfont} 
\usepackage{float}
\usepackage{graphicx}
\usepackage{todonotes}

\usepackage{amsthm}
\usepackage{natbib}
\usepackage{algpseudocode,algorithm,algorithmicx}
\usepackage[colorlinks,citecolor=blue]{hyperref}
\usepackage{cleveref}
\usepackage[final]{showlabels}
\usepackage{multirow}
\usepackage{lscape}
\usepackage{soul}

\theoremstyle{plain}
\newtheorem{theorem}{Theorem}
\newtheorem{lemma}{Lemma}

\newtheorem{proposition}{Proposition}

\theoremstyle{definition}

\theoremstyle{remark}

\theoremstyle{condition}
\newtheorem{condition}{Condition}

\newcommand{\Ind}{\mathds{1}}
\newcommand{\ind}{\Ind}

\newcommand{\size}{0.26}

\newcommand{\cm}{\mathbf{c}}

\newcommand{\F}{\mathbf{F}}

\newcommand{\B}{\mathbf{B}}
\newcommand{\I}{\mathbf{I}}
\newcommand{\M}{\mathbf{M}}
\newcommand{\N}{\mathbf{N}}

\newcommand{\W}{\boldsymbol{\mathcal{W}}}
\newcommand{\D}{\mathbf{D}}
\newcommand{\E}{\mathds{E}}
\newcommand{\bGamma}{\mathbf{\Gamma}}

\newcommand{\A}{\mathbf{A}}
\newcommand{\V}{\mathbf{V}}
\newcommand{\X}{\widehat{\mathbf{X}}}
\newcommand{\Xlim}{{\mathbf{X}}}
\newcommand{\bLambda}{\mathbf{\Lambda}}
\newcommand{\Y}{\mathbf{Y}}
\newcommand{\K}{\mathbf{K}}
\newcommand{\CMa}{\mathbf{C}}
\newcommand{\indv}{\mathbf{e}}
\newcommand{\Hb}{\mathbf{H}}
\newcommand{\IndRank}{\widehat{\operatorname{I_{F}}}}
\newcommand{\IndRanklim}{\operatorname{I_{F}}}
\newcommand{\Q}{\widehat{\mathbf{Q}}}
\newcommand{\Qlim}{\mathbf{Q}}
\newcommand{\var}{\mathds{V}ar}
\newcommand{\Var}{\var}
\newcommand{\Prob}{\mathds{P}}
\newcommand{\R}{\mathds{R}}

\newcommand{\bone}{\mathbf{1}}
\newcommand{\bzero}{\mathbf{0}}

\newcommand\Diag[2]{\operatorname{Diag}(#1(#2))}
\newcommand{\bUpsilon}{\mathbf{\Upsilon}}
\newcommand{\qwei}[1]{\widetilde{q}^M_{n,i}(1-#1)}
\newcommand{\qwi}[1]{\widetilde{q}_{n,i}(1-#1)}

\renewcommand{\S}{\boldsymbol\eta}
\newcommand{\given}{\middle|}
\newcommand{\transpose}{\intercal}

\begin{document}

\title{ A Multiple kernel testing procedure for non-proportional hazards in factorial designs.}

\author{Marc Ditzhaus
\thanks{Department of Mathematics, Otto-von-Guericke University Magdeburg, Germany. \texttt{marc.ditzhaus@ovgu.de}}  \and Tamara Fern\'andez \thanks{Faculty of Engineering and Science, Universidad Adolfo Ib\'a\~nez, Chile. \texttt{t.a.fernandez.aguilar@gmail.com}} \and Nicol\'as Rivera\thanks{Instituto de Ingenier\'ia Matem\'atica, Universidad de Valpara\'iso, Chile. \texttt{n.a.rivera.aburto@gmail.com}}}
\maketitle

\abstract{
In this paper we propose a Multiple kernel testing procedure to infer survival data when several factors (e.g. different treatment groups, gender, medical history) and their interaction are of interest simultaneously. Our method is able to deal with complex data and can be seen as an alternative to the omnipresent Cox model when assumptions such as proportionality cannot be justified.  Our methodology combines well-known concepts from Survival Analysis, Machine Learning and Multiple Testing: differently weighted log-rank tests, kernel methods and multiple contrast tests. By that, complex hazard alternatives beyond the classical proportional hazard set-up can be detected. Moreover, multiple comparisons are performed by fully exploiting the dependence structure of the single testing procedures to avoid a loss of power. In all, this leads to a flexible and powerful procedure for factorial survival designs whose theoretical validity is proven by martingale arguments and the theory for $V$-statistics. We evaluate the performance of our method in an extensive simulation study and illustrate it by a real data analysis.
}

\section{Introduction}
Statistical methods for hypothesis testing are essential tools for the practice of Statistics, and since the conception of the discipline they have been one of the main sources of research questions and a topic of interesting debate between statisticians.  Due to the advances in data-collection technologies that the world has experience during the last 20 years we are now able to collect more data and of better quality, and practitioners expect to solve more complex problems using this data. These advances have created new challenges for researchers in Statistical Methodology, which now have to develop robust methods that are able to answer complex testing problems by analysing large volumes of  complex data. Acknowledging this challenge, many areas of Statistics, specially in applied settings have raised the goal of adapting and creating new tools which can be used in these new times, and Survival Analysis -one of the most applied areas of Statistics in practical problems - is no stranger to that.  Indeed, while the log-rank test and Cox regression used to be the gold standard in classic Survival Analysis for a long time due to their ability to deal with the simple proportional hazard assumption,  there is a current counter-trend and an actual need in biomedical application of strategies being more against the proportional hazard  assumption. For example,  \citet{kuitunen2021testing} has pointed out the problem in the context of total joint arthroplasty (TJA) research: \textit{``Reporting and testing of the PH assumption and dealing with non-proportionality in hip and knee TJA studies was limited. More awareness and education regarding the assumptions behind the used statistical models among researchers, reviewers and editors are needed to improve the quality of TJA research.''} Another prominent example is modern immunotherapy, where delayed treatment effects cause non-proportional hazard situations \citep[e.g.][]{mick2015statistical}.

From a methodological point of view, the problem of dealing with complex and non-proportional hazards has been mainly tackled in the two-sample setting \citep[e.g.][]{ditzhaus2018more,ditzhaus2019wild,liu2020resampling, fernandez2020reproducing, lin2020alternative}, we especially refer to \cite{li2015statistical} and \cite{dormuth2022test} for comparison based on simulations and reconstructed real data, respectively. However, the two-sample setting is too simple, and \textit{``evaluating more than 1 new intervention concurrently increases the chances of finding an effective intervention''} \citep{juszczak2019reporting}. In this line some extension to the $k$-sample problem had been proposed \citep{bathke2009combined,chen2016comparison,gorfine2020k,liu2017partitioned}, nevertheless this is not enough for practitioners and the goal should be the very general setting of factorial design where different discrete covariates (or factors) are provided and we are interested on understanding the relations between them (that has as particular case the $k$-sample problem). In this regard the work of \citet{akritas1997nonparametric} was the only strategy targeting general factorial designs for a long time. Unfortunately, their strategy requires a rather strong assumption on the censoring mechanism which typically cannot be justified in  practice. Just recently, their procedure was, finally, complemented by a flexible toolbox for factorial survival designs without restrictive assumptions in the form of the R package \textit{GFDsurv} \citep{GFDsurv}. This package currently covers two strategies on effect sizes, namely the concordance probability \citep{dobler2020factorial} and medians \citep{ditzhaus2021inferring}, and a nonparametric procedure \citep{ditzhaus2020permutation} combining differently weighted Nelson--Aalen type integrals. However, two important aspects are still pending: (1) \textbf{a strategy for complex hazard alternatives}, (2) respective \textbf{post-hoc tests} for a more in-depth analysis of the data and simultaneous comparisons for the different factor combinations. 

Unfortunately, tackling the problem of testing factorial designs requires the incorporation of new ideas outside of the standard toolbox of Survival Analysis. For the issue of dealing with \textbf{complex hazard alternatives} (1), in this paper we follow a kernel-based strategy to solve hypothesis testing problems. Kernel-based tests started their development around 15 years ago by the Machine Learning community \cite{gretton2007kernel}, but were quickly adopted by researchers in Statistical Methodology. The main idea behind kernel-based strategies is to embed the observed data points into a reproducing kernel Hilbert space of functions, and then, by using the structural properties of the space, we develop a test-statistic and a strategy to find rejection regions, but in many contexts the idea reduces to optimise over an infinite family of functions. This kernel principle is the foundation for various testing problems for complex data, including graphs, time-series, functional-data, words, images, etc. \citep{berlinet2011reproducing,chwialkowski2014kernel,chwialkowski2016kernel,gretton2007kernel,gretton2012kernel}, and it has proven to be very effective at providing new tests that are very robust, and computationally cheap.

While progress in the relatively new field of kernel-based tests has been rather quick, its incorporation to the setting of survival data has been much slower than in other settings which might be due to the intrinsic difficulties (censoring, truncation) and the specific theoretical tools (martingale and counting process theory, etc) that feature the study of Survival Analysis. Up to the best of our knowledge, current work applying kernel methods with Survival Analysis have only focused in simple settings such as Goodness-of-fit \citep{fernandez2019maximum,fernandez2020kernelized}, two-sample \citep{fernandez2020reproducing}, and independence problems \citep{fernandez2020kernel} for which a plethora of good methods has been developed since the 70's-80's (we refer to chapter 7 of the book of  \citet{klein2003survival} for a summary of classic methods in Survival Analysis). Moreover, even though the problems above are important, they are still far away from the setting of factorial designs, and much of the ideas used cannot be directly extended to such environment. In this paper, we fill the aforementioned gap by deriving the first kernel-based method for general factorial survival designs to infer main and interaction effects of different factors, e.g. treatments, genders, ethnic origins etc, that deals with complex hazard structures and allows a post-hoc analysis. 

The aim of this work is to incorporate kernel-based strategies to existing ideas in Survival Analysis such as log-rank tests. For that our first contribution is the derivation of a weighted log-rank type statistic for the problem of factorial designs, which up to the best of our knowledge has not been studied before. Our second contribution is to show how to enhance the previous log-rank statistic by `kernelising' it, which means that we choose the weight function of the log-rank statistic in the unit ball of a Reproducing kernel Hilbert space, and then we optimise to find the best weight function based on our data, resulting in what we call the kernel log-rank statistic. This statistic is computationally cheap, rather robust, and able to deal with complex hazard functions as we show in our experimental section. Based on this test-statistic a testing procedure can be derived by using a Wild Bootstrap resampling scheme which fit quite naturally due to the structure of the kernel log-rank.

For \textbf{post-hoc testing} (2), multiple contrast tests are well-established procedures used in uncensored data-settings \citep[e.g.][]{bretz2001numerical,hasler2008multiple,konietschke2012rank,gunawardana2019nonparametric} and do not suffer from a significant power loss such as Bonferroni correction. However, respective extensions to infer complex non-proportional hazards in factorial survival designs are, to the best of your knowledge, still pending. In our third contribution, we remedy this by combining the principle of multiple contrast testing with the kernel-based tests described in the previous paragraph. Contrary to classical multiple contrast tests, we have to deal with a vector of kernel log-rank statistics which are not asymptotically multivariate normal but each follow a more complex distribution, namely a (infinite) sum of weighted $\chi^2_1$-distributions. Consequently, critical values cannot be formulated, as typically, in terms of multivariate normal or $t$-quantiles \citep{mvtnorm}. To account for this, we  develop a Wild Bootstrap resampling scheme to estimate the unknown null distribution of the final multiple testing procedure. Up to the best of our knowledge the problem of multiple testing has not been accounted in the  literature of kernel-based testing, however, related ideas have featured in very recent works on adaptive tests \citep{
schrab2021mmd,albert2022adaptive}

\subsection{Structure of the Paper}

The paper is organised as following: 
\Cref{sec:framework} introduces the data-setting where we develop our new methodology, and introduces the factorial design testing problem. In \Cref{sec:teststat}, we introduce a  novel weighted log-rank test procedure for the factorial design setting. We present the kernel test-statistic in \Cref{section:KernelTestStatistic}. It combines an infinite collection of weight functions into one robust test.  We continue with \Cref{sec:wild}, where we show how to build a proper testing procedure by using a Wild Bootstrap resampling scheme. Later in \Cref{sec:multiple}, we combine the kernel method with the idea of multiple contrast tests and derive a post-hoc test for an in-depth analysis. The analysis of asymptotic properties of our test is performed in \Cref{sec:asymptotic}, while in \Cref{sec:experiments}, we empirically evaluate our method in several simulated data-settings, as well as real data-scenarios.

\section{Framework and Notation}\label{sec:framework}
In this work we consider the standard right-censoring data setting where we observe $n$ independent and identically distributed data points $D_i= (T_i,\Delta_i, X_i)$, $i \in [n]:=\{1,\ldots, n\}$. Here $T_i = \min\{Z_i,C_i\}$ is the observed time, defined as the minimum between the time of interest $Z_i$, and a censoring time $C_i$.  For simplicity, we assume both $Z_i$ and $C_i$ are continuous random variables. The variable $\Delta_i \in \{0,1\}$, known as the censoring indicator, takes the value $1$ {if the actual time of interest can be observed, i.e. if $Z_i\leq C_i$}, whereas it takes the value $0$ when {the observation is censored, i.e} $Z_i>C_i$. Finally, $X_i$ is a covariate taking values on the set $[k]$ and it encodes the membership of the $i$-th observation to one out of $k$ (sub-)groups. As we explain later in more detail, having covariates taking values in $[k]$ is enough to incorporate general factorial designs. We assume that  $Z_i$ and $C_i$ depend on the variable $X_i$, however, $Z_i$ and $C_i$ are independent given $X_i$, which is a standard assumption in Survival Analysis. {While it is beneficial for the proofs to consider randomised covariates, the theoretical derivations can be similarly performed for deterministic covariates with a little more technical effort \citep[c.f.][]{fernandez2020reproducing}.} For our analysis we set probability $p_j:= \Prob(X_i=j)$ for each $j\in [k]$. 

We denote by $F_j$ and $S_j=1-F_j$ the distribution and survival function, respectively, of the survival time $Z$ of an individual belonging to group $X=j$. Moreover, let $\Lambda_j$ be the respective  cumulative hazard function defined by $\Lambda_j(t) = \int_0^tS_j(s)^{-1}dF_j(s)$ for all $t\geq 0$. All cumulative hazard functions are collected together into the vector $\bLambda =(\Lambda_1,\ldots, \Lambda_k)$. Moreover, let $H_j$ be the distribution function of the observed time $T$ of an individual from group $X=j$, and set the vector $\Hb(t)=(H_1(t),\ldots,H_k(t))$.

To define the estimators and testing procedures, we adopt the standard counting process notation \citep{andersen2012statistical}. For each data point $D_i = (T_i,\Delta_i,X_i)$, we define the processes $N_i(t) = \Delta_i \ind_{\{T_i\leq t\}}$, $Y_{i}(t) = \ind_{\{T_i\geq t\}}$, and $M_{i}(t) = N_{i}(t)-\int_0^t Y_{i}(s)d\Lambda_{X_i}(s)$.  
We recall the reader that by choosing the appropriate filtration $\mathcal F_t$, the processes $(N_i(t))_{t\geq 0}$, $(Y_i(t))_{t\geq 0}$ and $(M_i(t))_{t\geq 0}$ are adapted, predictable, and a martingale, respectively, and so are their respective group-version.

Define the $k$-dimensional vector $\indv_{i}$ by $\indv_{i}=(\ind_{\{X_i=1\}},\ind_{\{X_i=2\}},\ldots,\ind_{\{X_i=k\}})$. Observe that all the entries of $\indv_{i}$ are equal to zero except the entry at the $X_i$th position, which has value one.  By using the previous definition, we introduce vector-valued processes which are analogue of the processes $(N_i(t))_{t\geq 0}$, $(Y_i(t))_{t\geq 0}$ and $(M_i(t))_{t\geq 0}$. We define the $k$-dimensional vectors-valued processes
\begin{align*}
\N^i(t)=N_i(t)\indv_{i}, \quad \Y^i(t)=Y_i(t)\indv_{i},\quad\text{and}\quad \M^i(t)=M_i(t)\indv_{i},
\end{align*}
that take the value of $N_i(t)$, $Y_i(t)$ and $M_i(t)$, respectively, in the coordinate $X_i$. We also consider the population version of those processes, namely, $(\N(t))_{t\geq0}$, $(\Y(t))_{t\geq0}$ and $(\M(t))_{t\geq0}$ defined by
\begin{align*}
\N(t) = \sum_{i=1}^n \N^i(t), \quad \Y(t)=\sum_{i=1}^n \Y^i(t), \quad \text{ and } \quad \M(t)=\sum_{i=1}^n \M^i(t).
\end{align*}
For readability purposes, we denote vectors and matrices (including random ones and processes) with bold letters, whereas scalar with light letters. For vector and matrices we write the entries in subindices, e.g. for a vector $\mathbf{v}$ we write $\mathbf{v}_i$ to represent the $i$-th component of it, and similarly, for a matrix $\mathbf{A}$, we denote by $\mathbf{A}_{ij}$ its entry $(i,j)$. If we have a time-dependent vector/matrix, we write the dependence on time after the subindices, for example $\N_j^i(t)$ represents the $j$-th entry of the vector-valued process $\N^i$ at time $t$.
\subsection{Factorial Designs and Additive Models}\label{sec:FactorialDesigns}
 
A factorial design considers one or more independent variables, known as factors, which take discrete possible values or levels. Common examples of factors in Survival Analysis are, e.g., gender, blood type, {ethnic origin, treatment group} etc.  In this setting, each survival time belongs to one of the subgroups that arise by considering all the combinations of levels across all factors. Factorial designs are important since  they allow us to study the individual effect of each factor, as well as the combined effect, known as interaction, of one or more factors on the survival times of interest. 

One of the simplest factorial designs we may study is the $2\times 2$ factorial design in which we consider two factors, $\mathcal{I}$ and $\mathcal{J}$, each of them having 2 levels, i.e, $\mathcal{I}=\{1,2\}$ and  $\mathcal{J}=\{1,2\}$.  For this example, the total number of groups is $k=4$, and thus there are 4 cumulative hazard functions $\Lambda_{ij}$ with $i\in\mathcal{I}$ and $j\in\mathcal{J}$. There are several questions which can be asked about this system, for instance, we may be interested on testing if there is no effect of the factor $\mathcal{I}$ on the survival times. Mathematically, this is equivalent to assessing if the cumulative hazard functions do not change for different levels of the factor $\mathcal{I}$, that is $\Lambda_{1j}=\Lambda_{2j}$ for all $j\in\mathcal{J}$. By writing the four hazards as the vector-valued function $\bLambda = (\Lambda_{11},\Lambda_{12},\Lambda_{21},\Lambda_{22})$, we can test the hypothesis that there is no effect on $\mathcal I$ as
\begin{align}\label{eqn:mainhypo}
H_0:\{\CMa\bLambda(t)=\bzero,\text{ for all }t\geq 0\}\quad\text{vs}\quad H_1:\{\CMa\bLambda(t)\neq \bzero,\text{for some }t\geq 0\}
\end{align}
where $\CMa =\begin{pmatrix}
1 & 0 & -1&0\\
0 & 1 & 0 & -1
\end{pmatrix}$. In general, $\CMa$ can be any contrast matrix, i.e. a matrix fulfilling 
$\CMa \bone=\bzero$ for the vectors $\bone$ and $\bzero$ consisting of 1' and 0's only.

{Let us know switch to a slightly more general case allowing more than two levels per group by considering two factors $\mathcal I$ and $\mathcal J$, with levels $\{1,\ldots, |\mathcal I|\}$ and $\{1,\ldots, |\mathcal J|\}$, respectively. Then, further interesting null hypotheses and related contrast matrices naturally arise by modelling the factorial design as an additive regression model.  Here, there are} $k=|\mathcal I||\mathcal J|$ combinations between the levels of the two factors, and the survival time of each group is associated to a cumulative hazard function $\Lambda_{ij}$.
Since there is a finite number of groups, the set of cumulative hazard functions $(\Lambda_{ij}: i \in \mathcal I, j \in \mathcal J)$ can be uniquely decomposed as
\begin{align}\label{eqn:AalenDecomposition0}
\Lambda_{ij}(t) = \Lambda_0(t)+ \Phi_i(t) + \Psi_j(t)+ \Sigma_{ij}(t), \quad\text{with } i \in \mathcal I \text{ and } j \in \mathcal J
\end{align}
satisfying $\sum_{i\in \mathcal I}\Phi_i(t) = \sum_{j\in \mathcal J}\Psi_j(t) = \sum_{i\in \mathcal I}\Sigma_{ij}(t) = \sum_{j\in \mathcal J}\Sigma_{ij}(t) = 0$. 
Indeed, by considering the following terms $\Lambda_{i\cdot}= \frac{1}{|\mathcal J|} \sum_{j\in \mathcal J}\Lambda_{ij}$, and $\Lambda_{\cdot j}= \frac{1}{|\mathcal I|} \sum_{i\in \mathcal I}\Lambda_{ij}$ and $\Lambda_{\cdot\cdot} = \frac{1}{|\mathcal J||\mathcal I|}\sum_{(i,j)\in \mathcal I\times \mathcal J}\Lambda_{ij}$, we can easily deduce that
\begin{align}\label{eqn:AalenDecomposition}
\Lambda_0 = \Lambda_{\cdot\cdot}, \quad \Phi_i= \Lambda_{i\cdot}-\Lambda_{\cdot \cdot}, \quad \Psi_j = \Lambda_{\cdot j}-\Lambda_{\cdot \cdot}, \quad {\text{and }} \quad\Sigma_{ij} = \Lambda_{ij}-\Lambda_{i\cdot}-\Lambda_{\cdot j} + \Lambda_{\cdot\cdot}.
\end{align}

Here $\Lambda_0$ is a common cumulative hazard for all the groups, whereas $\Phi_i$ and $\Psi_j$ are the cumulative hazards related only to factors $\mathcal I$ and $\mathcal J$, respectively, and $\Sigma_{ij}$ is the cumulative hazard associated to the interaction between the factors. Several question can be asked here and solved by posing the appropriate testing problem, the most common examples include:
\begin{enumerate}
\item No main effect of the factor $\mathcal I$. We shall test if $H_0:\{\Phi_i=0 \text{ for all }i\in \mathcal I\}$ holds. From \cref{eqn:AalenDecomposition}, such hypothesis is equivalent to $H_0:\{\Lambda_{i\cdot}=\Lambda_{j\cdot} \text{ for all }i,j\in \mathcal I\}$. 
\item No effect of factor $\mathcal I$. We shall test if $H_0:\{\Phi_i+\Sigma_{ij}=0\text{ for all }i\in \mathcal I \text{ and } j \in \mathcal J \}$ which is equivalent to $H_0:\{\Lambda_{ij} = \Lambda_{i'j}\text{ for all }i,i'\in \mathcal I \text{ and } j \in \mathcal J \}$.
\item No interaction effect. We shall test $H_0:\{\Sigma_{ij}=0 \text{ for all }(i,j)\in \mathcal I\times \mathcal J\}$ which is equivalent to test $H_0:\{\Lambda_{ij}-\Lambda_{i\cdot}-\Lambda_{\cdot j}+\Lambda_{\cdot \cdot} = 0\text{ for all }(i,j)\in \mathcal I\times \mathcal J\}$.
\end{enumerate}

We remark that the previous additive model can be generalised to more than two factors, and also covers hierarchical designs, leading to several natural and statistically meaningful testing problems that can be rephrased as \cref{eqn:mainhypo}. {For more details, we refer the readers to \cite{pauly2015asymptotic}.}

Since all the hypotheses above are written in terms of a homogeneous systems of equations $\CMa\bLambda = \bzero$, our main focus is the the general problem where $\CMa$ is any contrast matrix  of $k$ columns (not necessarily arising in a factorial design problem). Observe that due to the way the problem is displayed, without lost of generality, we can consider covariates taking value in the set $[k]$.

\section{A Log-rank Test-statistic}\label{sec:teststat}
The first contribution of this paper is the introduction of a log-rank test-statistic for the factorial design testing problem stated in \cref{eqn:mainhypo} where $\CMa$ is any contrast matrix of interest. The main idea is to find a vector of $k$ cumulative hazards $\bLambda$ that maximise the likelihood of our observations $(T_i,\Delta_i,X_i)_{i=1}^n$ subject to $\CMa \bLambda(t)=\bzero$ for all $t\geq 0$. Under the null hypothesis that  $\CMa \bLambda=\bzero$, maximising the likelihood function with or without the constrain should lead to the same maximum, however, under the alternative, we should observe different behaviours for the maximum likelihood problems,  differences we expect to capture with our testing procedure.

To put this idea in practice, let say we believe that our data is generated by the vector of hazards $\bLambda:[0,\infty)\to \R^k$ satisfying the null. In such case, we can model departures from $\bLambda$ by considering the family of hazard  functions $\bLambda(t;\theta)$ given by
\begin{align}
d\bLambda_i(t;\theta)&=(1+\theta\omega_i(t))d\bLambda_i(t),\quad i\in[k],\label{model1}
\end{align}
where $\theta$ is a scalar, and $\omega_1,\ldots,\omega_k$ are some fixed functions $\omega_i:[0,\infty)\to\R$.

Under the assumption that our data is actually generated by $\bLambda(t;\theta)$ for some $\theta$, testing $H_0:\CMa\bLambda(t)=\bzero$ is equivalent to test $H_0:\theta=0$. Notice that the latter can be done by using the score function obtained from \cref{model1} evaluated at $\theta=0$, i.e.,
\begin{align}
U_0&=\frac{1}{\sqrt{n}}\sum_{i=1}^k \int \omega_i(t)d\N_{i}(t)-\int \omega_i(t) \Y_{i}(t)d\bLambda_i(t)\nonumber\\
&=\frac{1}{\sqrt{n}}\bone^\transpose\int \W(t)\left(d\N(t)-\Diag{\Y}{t}d\bLambda(t)\right)\label{eqn:scoreTheo} ,
\end{align}
where $\W(t)=diag(\omega_1(t),\ldots,\omega_k(t))$ and $\Diag{\Y}{t}=diag(\Y(t))$ are $k\times k$ diagonal matrices. The factor $\frac{1}{\sqrt{n}}$ in $U_0$ is just for normalisation purposes. The term $U_0$ can be used as a test-statistic to assess the validity of the restriction  $H_0:\theta=0$.  Note, however, that $U_0$ cannot be directly evaluated from the data as it requires the unknown quantity $\bLambda$. To solve this issue, we propose to replace $\bLambda$ by a non-parametric estimate. Our non-parametric estimate of $\bLambda$ is based on the additive hazards model. For such, denote by $N(\CMa)$ the null space of $\CMa$, and let $d:= dim(N(\CMa))$. Let $\boldsymbol{v}_1,\ldots,\boldsymbol{v}_d\in\R^k$ be a basis of $N(\CMa)$, and define $\V=(\boldsymbol v_1,\ldots,\boldsymbol v_d)\in\R^{k\times d}$ as the matrix containing the vectors $\boldsymbol v_i$ as columns. Under the  null hypothesis, $\bLambda(t)$ can be written as
\begin{align*}
\bLambda(t)&= \boldsymbol v_1B_1(t)+\ldots+\boldsymbol v_d B_d(t)=\V\B(t),
\end{align*}
where $\B(t)=(\B_1(t),\ldots,\B_d(t))^\transpose$ are a collection of functions that are completely unspecified except by the fact that $\int_{0}^{\infty} |d\B_i|(t)<\infty$.
 
An estimator for $\B(t)$ can be obtained using ordinary least squares \cite[Section 4.2.1.]{aalen2008survival}. The main idea is to use the following equality
\begin{align}\label{eqn:martiMulti}
d\N(t)= \Diag{\Y}{t}d\bLambda(t)+d\M(t)=\Diag{\Y}{t}\V d\B(t)+d\M(t),
\end{align}
which resembles  a linear regression model, and where we can interpret $d\M(t)$ as a zero mean error. Let $\X(t)=\frac{1}{n}\Diag{\Y}{t}\V$ and denote by $\IndRank(t)$ the indicator variable that  $\X(t)$ has full rank. Then, we can estimate $\B(t)$ by multiplying both sides of \cref{eqn:martiMulti} by the pseudo inverse of $\X(t)$ to obtain the estimate $d\widehat{\B}(t)=\frac{1}{n}\IndRank(t)( \X(t)^\transpose\X(t))^{-1}\X(t)^\transpose{d\N}(t),$ and thus $\bLambda(t)$ can be estimated by
\begin{align}
d\widehat{ \bLambda}(t)=\V d\widehat{\B}(t)=\frac{1}n\IndRank(t)\V(\X(t)^\transpose\X(t))^{-1}\X(t)^\transpose{d\N}(t).\label{eqn:EstimatorLambda}
\end{align}
We remark that this estimator is independent of the basis $\boldsymbol v_1,\ldots, \boldsymbol v_d$ we chose for $N(\CMa)$. By plugging in the estimator of $\bLambda$ into \cref{eqn:scoreTheo} we get an estimator of $U_0$, i.e.:
\begin{align}
\widehat{U}_0
&=\frac{1}{\sqrt{n}}\bone^\transpose\int \W(t)\Q(t)d\N(t),\label{eqn:ourLogrank}
\end{align}
where $\Q(t)$ is a $k\times k$ matrix given by 
\begin{align}
\Q(t)&=\IndRank(t)\left(\I_k-\X(t)(\X(t)^\transpose\X(t))^{-1}\X(t)^\transpose\right).\label{eqn:defhatQ}
\end{align}
As $\widehat U_0$ is deduced similarly as the weighted log-rank statistic for the two-sample problem, we call the test-statistic $\widehat{U}_0$ the weighted log-rank statistic for the global null hypothesis $H_0:\CMa\bLambda(t) =\bzero$.

\section{Kernel log-rank test-statistic.}\label{section:KernelTestStatistic}

Observe that the log-rank test-statistic $\widehat U_0$ introduced in \cref{eqn:ourLogrank} is implicitly defined in terms of a function $w:[0,\infty)\times[k]\to \R$ that receives as inputs a time $t\geq 0$ and a group-label $i\in[k]$, that is, $\omega(t,i)=\omega_i(t)$. Thus, it will be convenient to write $\widehat U_0(w)$ to make this dependence explicit. 

The log-rank test is constructed based on the assumption that our data is generated by the model of \cref{{model1}}. Thus choosing the appropriate weight function $w$ in \cref{{model1}} is extremely relevant as a wrong choice could lead us to inferior results. This selection problem is well-known for the two-sample weighted log-rank test \citep{fernandez2020reproducing,ditzhaus2018more,dormuth2022test}. The problem here is that choosing the 'correct' weight function $\omega$ is even more difficult as there are many types of interactions between time and covariate that can be encoded by $w$. Typically, they are hard to visualise or to obtain via an exploratory analysis as we are dealing with several groups at the same time. {Moreover, choosing the weight after a first inspection leads to data adaptive weights, which are not covered by the standard theory.} Rather than dealing with the problem of choosing a specific weight function, we prefer to consider a large variety of weight functions at the same time. In particular, we wish to consider a test-statistic of the form
\begin{align}
\sup_{w\in \mathcal F} \left|\widehat U_0(w)\right|,\label{eqn:maxlogrank}
\end{align}
where $\mathcal F$ is a collection of weight functions. In principle, there is nothing wrong with choosing any space $\mathcal F$, however, it is very likely we will not be able to evaluate $\sup_{w\in \mathcal F} \left|\widehat U_0(w)\right|$, and even in that case, finding rejection regions for a test would be an intractable problem, either exactly, asymptotically or by a resampling scheme. We will see that  this is not the case when $\mathcal F$ is the unit ball of a reproducing kernel Hilbert space, and indeed, we will obtain a test-statistic that i) is easy to evaluate, ii) has desirable asymptotic properties, and iii) has a simple resampling scheme for finding rejection regions.

\subsection{Reproducing kernel Hilbert spaces}
We introduce some basic notions of reproducing kernel Hilbert spaces (RKHS). An RKHS is a space of functions $\omega:\mathcal X\to \R$ satisfying that the evaluation functional $E_x$, $E_x\omega \to \omega(x)$ is continuous for every fixed $x \in \mathcal X$. Since $E_x$ is continuous for any $x \in \mathcal X$, the Riesz representation theorem yields the reproducing property, which states that for any $x\in \mathcal X$ it exists a unique  $K_x \in \mathcal H$ such that $\omega(x) = \langle K_x, \omega\rangle_{\mathcal H}$ for all $\omega \in \mathcal H$. Since $K_x \in \mathcal H$ for all $x\in\mathcal X$, $K_x(y) = K_y(x) = \langle K_x, K_y \rangle_{\mathcal H}$ holds for any $x,y \in \mathcal X$. This allows us to define the so-called \emph{reproducing kernel} $K:\mathcal X^2\to  \R$ as
\begin{align}\label{eqn:repKernel}
K(x,y) =  \langle K_x, K_y \rangle_{\mathcal H}.
\end{align}
From now on, in order to ease the notation, we write $K(x,\cdot)$ instead of $K_x(\cdot)$, even though the former induces a slight abuse of notation.
 
For every RKHS $\mathcal H$ with inner product $\langle \cdot, \cdot \rangle_{\mathcal H}$ there exists a unique symmetric positive-definite reproducing kernel $K:\mathcal X^2\to\R$ satisfying ~\cref{eqn:repKernel}. Conversely, by the Moore-Aronszajn Theorem \citep{Aronszajn1950}, for any  symmetric positive-definite kernel function $K:\mathcal X^2\to \R$, there exists a unique RKHS $\mathcal H$ for which $K$ is its reproducing kernel. The Moore-Aronszajn Theorem is quite convenient for us as we do not need to describe the RKHS $\mathcal H$ but rather its kernel function $K$. Some common kernel functions defined on $\mathcal X = \R^d$ are the Gaussian kernel, $K(x,y) = \exp\{-\|x-y\|^2/\sigma^2\}$, and the Ornstein-Uhlenbeck kernel, $K(x,y) = \exp\{-\|x-y\|/|\sigma|\}$, where in each case $\sigma>0$. 

We introduce the important notion of $c_0$-universality \cite{sriperumbudur2011universality}. Suppose that $\mathcal X$ is a locally compact Hausdorff space (in particular $\mathcal X$ can be subset of $\R$, or a finite product of them), then a kernel $K:\mathcal X^2\to \R$ is said to be a $c_0$-kernel if it is bounded and $K(x,\cdot):\mathcal X\to \R$ is continuous and vanishing at infinity. Moreover, a $c_0$-kernel $K$ is $c_0$-universal if the RKHS $\mathcal H$ associated with $K$ is dense in $C_0(\mathcal X)$. This is equivalent to say that the embedding of finite signed measures $\mu$ into $\mathcal H$, defined as $\int K(x,\cdot)\mu(dx)$, is injective. While this definition is very technical, most of the usual kernels such as the squared exponential kernel, the Laplacian kernel, and the rational quadratic kernel are $c_0$-universal.

In this work we will consider kernels $K:([0,\infty)\times[k])^2 \to \R$, i.e. a kernel that is defined in the space of time and group labels. A rather simple way to construct kernels in this particular domain is to let $K$ be the product of two kernels, one defined for the times and another one for the group-labels. That is, let $L:[0,\infty)^2 \to \R$ and $J:[k]^2\to\R$, then we define the kernel function $K$ by 
\begin{align*}
K((t,i),(s,j))= L(t,s)J(i,j).
\end{align*}
Given a kernel $K$ as described above, we denote by $\K(t,s)$ the $k\times k$ matrix-valued process defined as $\K:[0,\infty)^2\to \R^{k\times k}$, where $\K_{ij}(t,s) = K((t,i),(s,j))$ denotes element in position $(i,j)$ of the matrix $\K(t,s)$. It is worth mentioning that the product of $c_0$-universal kernels is $c_0$-universal.

\subsection{Kernel log-rank statistic}
As we mentioned before, we want to avoid choosing a weight function $w$ in $\widehat U_0(w)$ by using several of them as in \cref{eqn:maxlogrank}. For that, we define
\begin{align}
\Upsilon_n(\CMa) =\left(\sup_{\omega\in B_1(\mathcal{H})}\widehat{U}_0(\omega)\right)^2,\label{eqn:statSup}
\end{align}
where $B_1(\mathcal{H})$ denotes the unit ball of a reproducing kernel Hilbert space of functions $\mathcal{H}$ associated to a kernel function $K:([0,\infty)\times [k])^2\to\R$. We call $\Upsilon_n(\CMa)$ the kernel log-rank test-statistic. Note we explicitly write the subindex $n$ to indicate the number of data points.

By using the reproducing property of RKHS, we can obtain a closed-form expression for $\Upsilon_n(\CMa)$ as following using the matrix-valued analogue of $K$ (which is denoted by $\K$).

\begin{proposition}\label{prop:closedform1}
{We can rewrite the statistic as follows}
\begin{align*}
\Upsilon_n(\CMa)  = \frac{1}{n}\int \int (\Q(t)d\N(t))^{\transpose} \K(t,s) (\Q(s)d\N(s))).
\end{align*}
\end{proposition}
To understand the expression above, recall that for each $t\geq 0$, $\Q(t)$ is a $k\times k$ matrix, whereas $\N(t)$ is a vector in $\R^k$, and similarly, for all $t,s\geq 0$, $\K(t,s)$ is a $k\times k$ matrix, so the above integral is indeed a scalar.

The next step is to study the distribution of the test-statistic $\Upsilon_n(\CMa)$ under the null hypothesis. By understanding how $\Upsilon_n(\CMa)$ behaves under the null hypothesis we will be able to determine rejection regions which will be fundamental to implement a testing procedure. Thus, our next result is important as it  characterises the asymptotic null distribution of our test statistic $\Upsilon_n(\CMa)$. 

From now on, we will assume that our kernel functions are bounded unless mentioned otherwise. This is mostly to avoid tedious computations in our proofs.

\begin{condition}\label{Cond:bounded} The kernel function $K$ is $c_0$-universal.
\end{condition}

\begin{theorem}\label{thm:asymptoticUpsilon}
Under the null hypothesis and \Cref{Cond:bounded}, it exists a random variable $\Upsilon(\CMa)$ with cumulative distribution function $P_{\Upsilon}$ such that $\Prob(\Upsilon_n(\CMa)\leq t) \to P_{\Upsilon}(t)$ for any $t\geq 0$, when the number of data points tend to infinity, that is, $\Upsilon_n(\CMa) \overset{\mathcal D}{\to} \Upsilon(\CMa)$.
\end{theorem}
The proof of \Cref{thm:asymptoticUpsilon} and other theoretical results{, such as the explicit form of the limiting distribution,} are postponed to \Cref{app:proofs}.

We continue by analysing the behaviour of our test-statistic under the alternative hypothesis. The natural alternative hypothesis is $H_1:\{\CMa \bLambda(t) \neq \bzero \text{ for some $t \in (0,\infty)$}\}$. Nevertheless, we will consider a slightly simpler alternative $H_1'$; define $\tau_H:=  \sup\{t\geq 0: \Xlim(t)=\Diag{\S}{t}\V\text{ has full rank}\}$, where $\S(t) = (p_1(1-H_1(t)),\ldots, p_k (1-H_k(t)))$ (which clearly depends on the distributions that generate the times of interest and censoring), then $H_1'$ is given by
\begin{align}
    H_1':\{\CMa\bLambda(t)\neq \bzero,\text{for some }t \in [0,\tau_H)\}.\label{eqn:defiHalternativeGood}
\end{align}
We call $H_1'$ the visible version of $H_1$. The idea is that the violation of the null hypothesis is not hidden by the censoring distribution (whereas in $H_1$ this might be possible). 

\begin{theorem}\label{thm:alternative}
Assume \Cref{Cond:bounded}, then if a visible alternative $H_1'$ holds, then $\Upsilon_n(\CMa)\overset{\Prob}{\to} \infty$.
\end{theorem}

{Even though \Cref{thm:asymptoticUpsilon,thm:alternative} provide the basis to develop an asymptotically valid and consistent testing procedure, the limiting distribution $P_\Upsilon$ is rather complex and depends on unknown quantities. Consequently, there is practically no possibility to determine the respective $(1-\alpha)$-quantile to serve as a critical value. However, the latter can be approximated by a Wild Bootstrap resampling scheme, which we introduce in the following section.}

\section{A Wild Bootstrap resampling scheme}\label{sec:wild}
In order to obtain a proper testing procedure we need to find rejection regions to decide if we shall reject the null hypothesis or not. To find a region we will use a Wild bootstrap resampling scheme. 

By \Cref{prop:closedform1}  $\Upsilon_n(\CMa)$ can be written in terms of the kernel and $\Q(t) d\N(t)$. We recall that under the null $\Q(t) d\N(t)$ is a vector of martingales (and hence it has 0 mean for each time $t$). The idea is now to resample from $\Q(t) d\N(t)$ as it contains all the relevant information about the  data $(T_{\ell},\Delta_{\ell},X_{\ell})$, indeed all the randomness involved in $\Upsilon_n(\CMa)$ is there. 

Since we have $n$ data points, consider weights $W_1,\ldots, W_n$, where the $W_i$'s are independent and identically distributed with $\E(W_i) = 0$ and $\Var(W_i)=1$. Then, define a vector-valued process $\widetilde \N$ by
\begin{align*}
\widetilde{\N}(t) = \sum_{i=1}^n W_i \N^i(t).
\end{align*}
Thus, we can define the Wild Bootstrap version of $\Upsilon_n(\CMa)$ as
\begin{align}
\widetilde \Upsilon_n(\CMa)& = \frac{1}{n}\left(\sup_{\omega \in B_1(\mathcal H)}\bone^\transpose\int \W(t)\Q(t)d\widetilde{\N}(t) \right)^2=\frac{1}{n}\int \int (\Q(t)d\widetilde{\N}(t))^{\transpose} \K(t,s) (\Q(s)d\widetilde{\N}(s))),\label{eqn:closeformwild1}
\end{align}
where recall that $\W(t)=diag(\omega_1(t),\ldots,\omega_k(t))$, and the second equality follows from an analogous result to \Cref{prop:closedform1}.

Our Wild Bootstrap resampling scheme is asymptotically correct, in the sense that it approximates the limit distribution of $\Upsilon_n(\CMa)$, as we show in the following theorem.
\begin{theorem}\label{thm:wildAymp}
Assume \Cref{Cond:bounded}, then under the null hypothesis $H_0$ it holds that
\begin{align*}
 \Bigl | \Prob(\widetilde{\Upsilon}_n(\CMa) < t| D_1,\ldots,D_n) - \Prob(\Upsilon(\CMa)< t) \Bigr | \to 0,
\end{align*}
for almost all data points $D_1,D_2,\ldots$ as $n$ grows to infinity, where $\Upsilon(\CMa)$ is the limiting random variable of \Cref{thm:asymptoticUpsilon}.
\end{theorem}

Note that in the previous theorem the result is conditional on the data we observe, and hence, we can resample as many times as we want (using independent copies of $(W_1,\ldots, W_n)$) to obtain a good empirical representation of the distribution $\Prob(\widetilde{\Upsilon}_n(\CMa) < t|D_1,\ldots,D_n)$, which is asymptotically the same as $\Prob(\Upsilon_n(\CMa) < t)$ {under the null hypothesis}. 

With our resampling scheme defined, we are ready to describe our testing procedure. Our procedure relies on approximating the quantiles of the distribution of ${\Upsilon}(\CMa)$ by resampling from $\widetilde{\Upsilon}_n(\CMa)$. The algorithm for our testing procedure is as follows.

\noindent\rule{\linewidth}{0.4pt}\\
\textbf{Algorithm 1: Testing Procedure for a Single Contrast Matrix}\\
\noindent\rule{\linewidth}{0.4pt}
\begin{enumerate}
\item Set the desired level of the test:  $\alpha\in (0,1)$.
\item Generate $M$ independent copies of $(W_1,\ldots, W_n)$, and  use them to compute $M$ copies of $\widetilde{\Upsilon}_n(\CMa)$, to then compute the $(1-\alpha)$-quantile of the sample. Call such a quantile $\widetilde{q}^M_n({1-\alpha})$.
\item Compute $\Upsilon_n(\CMa)$.
\item Reject the null hypothesis if $\Upsilon_n(\CMa)>\widetilde{q}^M_n({1-\alpha})$, otherwise do not reject.
\end{enumerate}
\noindent\rule{\linewidth}{0.4pt}

In practice, kernel functions are evaluated in data points, so the implementation of our testing procedure involves mostly manipulation of matrices.  Implementing our procedure is rather fast as many operations can be recycled (in particular matrix multiplications) to be used by all the Wild bootstrap samples.

The next theorem shows that our algorithm is asymptotically correct when testing against a visible alternative. Note that we reject the null hypothesis if the event $\{\Upsilon_n(\CMa)> \widetilde{q}^M_n({1-\alpha})\}$ holds true, then 
\begin{align*}
    \Prob(H_0 \text{ is rejected}) = \Prob\left(\Upsilon_n(\CMa)> \widetilde{q}^M_n({1-\alpha})\right).
\end{align*}
The previous expression is the power of the test when a visible alternative $H_1'$ holds.

\begin{theorem}\label{thm:Algo1Works}
 Assume \Cref{Cond:bounded} and suppose we are testing $H_0$ against a visible alternative $H_1'$. Let $\alpha \in (0,1)$, then Algorithm 1 is asymptotically correct, that is, under the null hypothesis
$$\limsup_{n\to \infty}\limsup_{M\to \infty}  \Prob(\Upsilon_n(\CMa)> \widetilde{q}^M_n({1-\alpha})) \leq \alpha,$$
and under the alternative hypothesis $H_1'$ we have that for any fixed $M\in\mathbb{N}$, it holds
$$\lim_{n\to \infty}  \Prob(\Upsilon_n(\CMa)> \widetilde{q}^M_n({1-\alpha})) = 1.$$
\end{theorem}

\section{Multiple contrast tests}\label{sec:multiple}
In the previous sections, we derived a testing procedure for the global null hypothesis $H_0:\{ \CMa \bLambda = \bzero \}$ for a given contrast matrix $\CMa$ of $k$ columns. However, upon rejection, the procedure is not able to tell us which component of $\CMa \bLambda$ is different from zero, that is, which equations of the linear system $\CMa \bLambda=\bzero$ are not satisfied. This is, in particular, of interest in multiple comparisons. For the latter, various contrast matrices can be used, e.g. Tukey-type \citep{tukey1953problem} or of many-to-one Dunnett-type \citep{dunnett1955multiple} matrices. For example, if we consider the $k$-sample problem, we can consider the following contrast matrix, which is equivalent to $k-1$ equations written in terms of row-vectors $\cm_i$:
 \begin{eqnarray} \label{Dunnett}
 \CMa = \left(\begin{array}{c}
\cm_1\\ \vdots \\ \cm_{k-1}
\end{array} \right) = \left(\begin{array}{ccccc}
-1 &1 &0 & \ldots & 0\\
-1 & 0 & 1 & \vdots&\vdots\\
\vdots & \vdots & \vdots& \vdots & \vdots\\
-1 & 0 & \vdots & \vdots & 1\\
\end{array}
 \right) \Longrightarrow
H_0:  \CMa\bLambda = \bzero
= \left\{ \begin{array}{c}
H_{01}: \Lambda_2-\Lambda_1 =0\\
H_{01}: \Lambda_3-\Lambda_1 =0\\
\vdots\\
H_{0k-1}: \Lambda_{k}-\Lambda_1 =0.\\
\end{array}\right.
\end{eqnarray}

A first naive idea for this multiple comparison problem is to infer every single null hypothesis $H_{0i}$ by using $\Upsilon_n(\cm_i)$ and to adjusted them by a Bonferroni correction. However, it is well known that this leads to a significant loss in power. Contrary, multiple contrast tests fully exploiting the dependence structure between the single testing procedures leading to more satisfactory results in various completely observable data scenarios \citep[e.g.][]{bretz2001numerical,hasler2008multiple,konietschke2012rank,gunawardana2019nonparametric}. In the spirit of these positive results, we combine the idea of multiple contrast testing with the novel kernel log-rank testing strategy to derive a powerful multiple contrast procedure for complex survival data. Since the limiting distribution of the kernel log-rank test is not a simple normal distribution, we cannot follow the classical way in terms of a (studentised) maximum-type statistic but follow a different strategy which is explained more detailed below.

To be more concrete, we like to infer the multiple testing problem
\begin{align*}
H_{0i}: \CMa_i \bLambda = \bzero\quad \text{holds for all $i \in [b]$.}
\end{align*} 
for contrast matrices $\CMa_1,\ldots,\CMa_b$ with $k$ columns, where we explicitly allow matrices and not just contrast vectors. The global null hypotheses then becomes $H_0:\CMa \bLambda = \bzero$ where $\CMa$ is the matrix resulting by concatenating 
the matrices $\CMa_1,\ldots,\CMa_b$ by row (i.e $\CMa$ still has $k$ columns), but is of minor practical interest. As already mentioned, the individual null hypotheses $H_{0i}$, also denoted as local null hypotheses, can be tested by $\Upsilon_n(\CMa_i)$. These individual statistics are collected together into a joint (random) vector $\bUpsilon_n = (\Upsilon_n(\CMa_1),\ldots, \Upsilon_n(\CMa_b))$.

In view of \Cref{thm:alternative}, we expect that under the global null, the vector $\bUpsilon_n$ convergences in distribution (as \Cref{thm:asymptoticUpsilon} ensures convergence of each coordinate). Unfortunately, this is not a direct consequence of our previous developments as  ${\Upsilon}_n(\CMa_1),\ldots, \Upsilon_n(\CMa_b)$ are dependent test-statistics being all computed for the same data. We prove in \Cref{thm:multipleV-statsConvergence} in \Cref{app:proofMultiple} that, indeed, $\bUpsilon_n$ converges in distribution under the null hypothesis. For now, let's assume that the limit distribution exists, then the high-level theoretical idea to implement the testing procedure is very simple: we should find a vector $\mathbf{q}=(q_1,\ldots,q_b)\in (0,\infty)^b$ such that under the global null hypothesis
\begin{align*}
    \lim_{n\to\infty}P\left(\bigcup_{i=1}^b \{\Upsilon_n(\CMa_i)\geq q_i\}\right)=\alpha,
\end{align*}
where $\alpha\in[0,1]$ is the desired level of the global test. That is, we want to identify the vector $\mathbf{q}$ satisfying that the probability that at least one component $\Upsilon_n(\CMa_i)$ exceeds the value $q_i$ under the null is exactly $\alpha$ for large $n$. 
Notice that the main advantage of using this approach is that the multiple contrast test is able to identify which individual hypotheses are not true, while maintaining a correct Type-I error of $\alpha$ for the global null hypothesis. Indeed, by using the vector $\mathbf{q}=(q_1,\ldots,q_b)$, we can partition the rejection region in $2^b-1$ disjoint subsets as illustrated in \Cref{Fig:Multiple contrast} (for $b=2$). Each of these regions represents a different way of rejecting the global null hypothesis. For example, in \Cref{Fig:Multiple contrast}, we can reject because: (1) only $H_{01}$  is false, (2) only $H_{02}$ is false and (1,2) {both $H_{01}$ and $H_{02}$} are false at the same time. Under the alternative hypothesis, we can identify the source of the rejection by identifying to which region our test statistic $\bUpsilon_n$ belongs. In our experimental section we will see that this method is less powerful than testing all the hypothesis combined in one matrix, i.e. to test directly the global null hypothesis, however, in simulated data -where we know which hypotheses are failing- the algorithm is able to correctly identify all the unsatisfied hypotheses  when enough data points are provided.

\begin{figure}[h]
\centering
\includegraphics[scale=0.8]{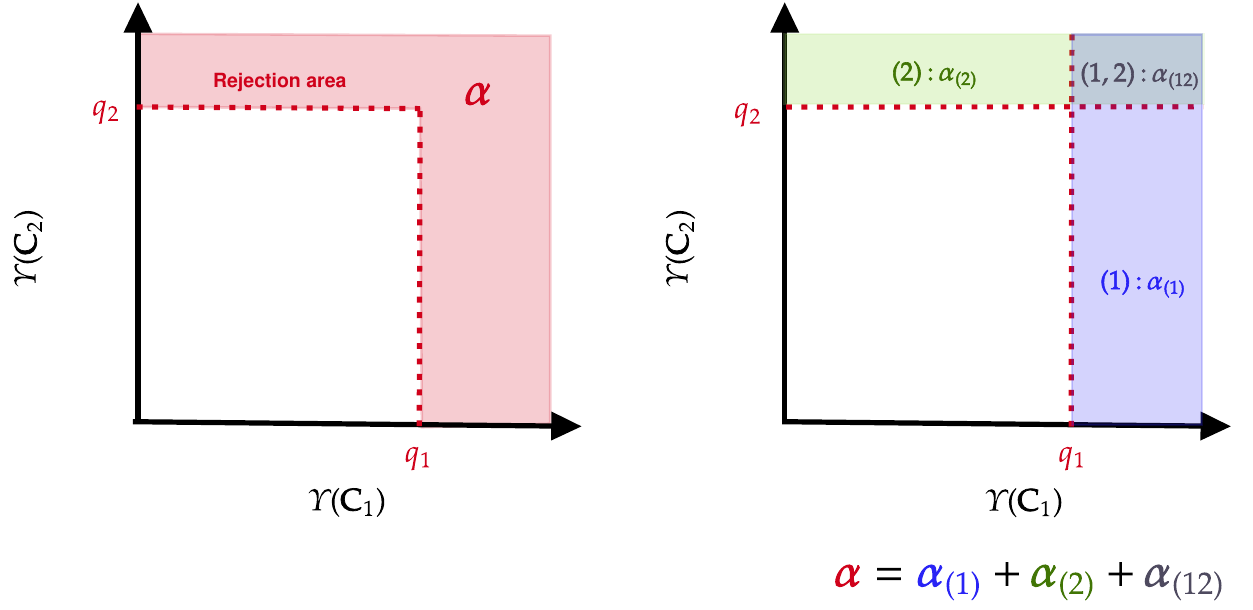}
\caption{Illustration of the rejection region of the multiple contrast test procedure for $b=2$.}\label{Fig:Multiple contrast}
\end{figure}

In principle, the vector $\mathbf{q}=(q_1,\ldots,q_b)\in (0,\infty)^b$ can be chosen in many ways. Here we present a simple way of doing that: define $\mathbf{q}=(q_{n,1}({1-\beta_\alpha}),\ldots,q_{n,b}(1-\beta_\alpha))$ where $q_{n,i}(1-\beta_\alpha)
$ denotes the $(1-\beta_\alpha)$-quantile of the limit distribution of ${\Upsilon}_n(\CMa_i)$ under the null, and $\beta_{\alpha}$ is defined as
\begin{align}\label{eqn:betaalphadefi}
    \beta_{\alpha} = \sup\left\{\beta\in [0,1]: \Prob\left(\bigcup_{i=1}^b \{\Upsilon(\CMa_i)\geq q_{n,i}(1-\beta)\}\right)\leq\alpha\right\},
\end{align}
where $\bUpsilon=(\Upsilon(\CMa_1),\ldots,\Upsilon(\CMa_b))$ is the limit (in distribution) of the random vector  $\bUpsilon_n$.

To implement the previous testing procedure we will rely on a Wild Bootstrap approximation of the rejection region under the null. Let $(W_1,\ldots, W_n)$ be $n$ independent Rademacher random variables. Then, using these weights, we obtain the Wild Boostrap test-statistics associated to each single test $\Upsilon_n(\CMa_i)$ for each $i\in[b]$ by using the procedure explained in \Cref{sec:wild} (in particular \cref{eqn:closeformwild1}) , and denote each of them by $\widetilde{\Upsilon}_n(\CMa_i)$.   We remark that all the components of $\widetilde \bUpsilon_n = (\widetilde{\Upsilon}_n(\CMa_1),\ldots,\widetilde{\Upsilon}_n(\CMa_b) )$ are obtained by using the same set of weights $(W_1,\ldots, W_n)$, so, the components of $\widetilde\bUpsilon_n$ are dependent given the observed data. We refer to the vector $\widetilde \bUpsilon_n$ as a Wild bootstrap sample of $\bUpsilon_n$.  Our testing procedure is then based on the following algorithm:\\

\noindent\rule{\linewidth}{0.4pt}
\textbf{Algorithm 2: The Multiple Contrast testing procedure. }\\
\noindent\rule{\linewidth}{0.4pt}
\begin{enumerate}
\item Set the desired level of the test (for the global hypotesis):  $\alpha\in (0,1)$.
\item Consider $M$ independent copies of $(W_1,\ldots, W_n)$, and use them to compute the $M$ corresponding Wild Bootstrap samples of $\bUpsilon_n$, say ${\widetilde{\bUpsilon}}_n^{1},\ldots, {\widetilde{\bUpsilon}}_n^{M}$.
\item Define by $\qwei{\beta}$ the empirical  $(1-\beta)$-quantile of $\widetilde{\Upsilon}_n(\CMa_i)$ obtained with the previous $M$ independent Wild Bootstrap samples ${\widetilde{\bUpsilon}}_n^{1},\ldots, {\widetilde{\bUpsilon}}_n^{M}$.
\item  Find $\widehat \beta_{\alpha}$ defined as the supremum over all $\beta$ such that
\begin{align}
\widehat \beta_{\alpha} = \sup\left\{\beta\in [0,1]: \frac{1}{M}\sum_{\ell=1}^M \ind\{{\widetilde{\Upsilon}}^{\ell}_n(\CMa_i)>\qwei{\beta} \text{ for at least one $i \in [b]$}\}\leq \alpha\right\} \label{eqn:conditionBeta1}
\end{align}
\item Reject the individual hypotheses $H_{0i}:\CMa_i\bLambda=\bzero$ whenever $\Upsilon_n(\CMa_i) \geq \qwei{\widehat \beta_{\alpha}}$. 
\item Reject the global hypothesis $H_0: \CMa\bLambda = \bzero$ if at least one hypothesis $H_{0i}$ is rejected.
\end{enumerate}
\noindent\rule{\linewidth}{0.4pt}

The only seemingly difficult step in the previous algorithm is to find $\widehat \beta_\alpha$. However, this can be done quite efficiently by noting that $\widehat \beta_\alpha$ belongs to the set $\{0,1/M,2/M,\ldots, 1\}$, so we can perform a binary search since each time we try some candidate value for $\widehat \beta_\alpha$ we already know if we shall try a larger or smaller value for it depending on the value of the left-hand-side of  \cref{eqn:conditionBeta1}.

The next theorem states the asymptotic correctness of our algorithm. Note that our algorithm rejects the hypothesis $H_{0i}$ if and only if the event $\left\{\Upsilon_n(\CMa_i) \geq \qwei{\widehat \beta_{\alpha}} \right\}$ holds true. In detail, we have shown that under the global null hypothesis, the Type-I error of the test is at most $\alpha$ (fixed by the user) when both the number of Wild Bootstrap samples $M$, and the number of data points $n$ tend to infinity. Additionally, if the global null is not true, then the test detects all the local hypothesis that are not true, and moreover it rejects a true local hypothesis with probability at most $\alpha$.
\begin{theorem}\label{thm:multiMatrixAlgorithmWorks}
Assume \Cref{Cond:bounded}, and suppose that if a hypothesis $H_{0i}$ is false then a visible alternative $H_{1i}'$ holds true. Then, given a fixed level $\alpha\in (0,1)$, we have that:
\begin{enumerate}
    \item If the global null holds true then the asymptotic type-1 error is smaller than or equal to $\alpha$, i.e.:
    \begin{align}
    \limsup_{n\to \infty}\limsup_{M\to \infty} \Prob\left(\bigcup_{i=1}^b \left\{H_{0i} \text{ is rejected}\right\}\right)\leq \alpha\label{thm7:eq1}
\end{align}
    \item Suppose the first $b'\leq b$ hypothesis are true, and the rest are false (and so their corresponding visible alternatives hold true), then 
    \begin{align}
       \lim_{n\to \infty} \Prob\left(\bigcap_{i=b'+1}^{b} \left\{H_{0i} \text{ is rejected}\right\}\right)=1, \quad \text{and } \quad        \limsup_{n\to \infty}\limsup_{M\to \infty} \Prob\left(\bigcup_{i=1}^{b'} \left\{H_{0i} \text{ is rejected}\right\}\right)\leq \alpha\label{thm7:eq2}
    \end{align}
   that is, the algorithm asymptotically is able to identify all false hypothesis, and rejects a true local hypothesis with probability at most $\alpha$.
\end{enumerate}
\end{theorem}

Note that Algorithm 2 recovers Algorithm 1, and \Cref{thm:Algo1Works} reduces to \Cref{thm:multiMatrixAlgorithmWorks} when $b = 1$.

\section{Experiments on  Simulated Data}\label{sec:experiments}

\subsection{Experimental Setup}\label{sec:experimentalSetup}

We describe three simulated data-settings that are used to evaluate the presented methodology, as well as the choice of hyperparameters and description of other methods used for comparison purposes.

\subsubsection{Data-settings $A$ and $B$}\label{subsec:datasettings}
We describe the data set $A$, which was proposed in \cite{ditzhaus2020permutation} {and represents a proportional hazards setting}, and data set $B$, which is a variation of data set $A$ with {non-proportional} hazard functions. In both cases the data is generated considering two factors: $\mathcal{I}\in\{1,2\}$ and $\mathcal{J}\in\{1,2,3\}$, leading to a total of $k=6$ groups. For each scenario we use the hazard functions described in \Cref{Table:data} to generate the data for each group (combination of factors). In \Cref{Table:data}, and in general, the group corresponding to the hazard $\lambda_{ij}$ with $i\in\mathcal{I}$ and $j\in\mathcal{J}$ is represented by the tuple $(i,j)$.

\begin{table}[H]
\centering
\begin{tabular}{|cc|cc|}
\hline
\multicolumn{2}{|c|}{Data $A$}&\multicolumn{2}{c|}{\footnotesize Factor $\mathcal{I}$}\\
&&1&2\\
\hline
\multirow{3}{*}{\rotatebox{90}{\footnotesize Factor $\mathcal{J}$}}&1&$\lambda_{11}(x)=1$&$\lambda_{21}(x)=2$\\
&2&$\lambda_{12}(x)=2$&$\lambda_{22}(x)=1$\\
&3&$\lambda_{13}(x)=1$&$\lambda_{23}(x)=1$\\
\hline
\end{tabular}
\quad
\begin{tabular}{|cc|ll|}
\hline
\multicolumn{2}{|c|}{Data $B$}&\multicolumn{2}{c|}{\footnotesize Factor $\mathcal{I}$}\\
&&1&2\\
\hline
\multirow{3}{*}{\rotatebox{90}{\footnotesize Factor $\mathcal{J}$}}&1&$\lambda_{11}(x)=\cos(2x)^2$&$\lambda_{21}(x)=\sin(2x)^2$\\
&2&$\lambda_{12}(x)=\sin(2x)^2$&$\lambda_{22}(x)=\cos(2x)^2$\\
&3&$\lambda_{13}(x)=1$&$\lambda_{23}(x)=1$\\
\hline
\end{tabular}\caption{Hazard functions used to generate Data $A$ and Data $B$. 
}
\label{Table:data}
\end{table}
For both data sets, and for each experiment, we consider balanced and unbalanced data-settings. In the \textbf{balanced} setting, all groups have the same number of data points, whereas in the \textbf{unbalanced} setting the groups have different sample size: for the groups $(1,1),(2,1),(1,2),(2,2),(1,3),(2,3)$ (enumerated from 1 to 6) we consider the sample size to be proportional to $15,9,5,9,7$ and $6$, respectively. In this setting we vary the sample sizes by multiplying the previous sizes by a factor greater or equal than 1, and by taking the floor function when needed. 

In both data-settings we consider censoring, which we assume is generated by an exponential distribution with rate parameter $\gamma$, independent of the groups. The parameter $\gamma$ is chosen in such a way it generates low ($5-10\%$ of censored observations ), medium ($25-35\%$) and high ($40-60\%$) percentage of censored observations. See \Cref{App:censoringExplanation} for more details about the censoring mechanism used in the experiments. 

We have two goals regarding these experiments. First, we will test the global null hypothesis that there is no main effect on the factor $\mathcal I$ i.e, our null hypothesis is $H_0:\Lambda_{1\cdot}(t)=\Lambda_{2\cdot}(t)$ for all $t\geq 0$, which is encoded by a contrast matrix with only one row. For both data-settings, $A$ and $B$, this hypothesis should not be rejected which can be verified by observing that in \Cref{Table:data} the columns of each table sum up to the same value (4 and 2 respectively). The second goal is to analyse the global hypothesis that there is no effect of factor $\mathcal I$, that is, $H_0: \Lambda_{1j}(t) = \Lambda_{2j}(t)$ for all $j\in \{1,2,3\}$. From \Cref{Table:data} we observe that the hypothesis should be rejected (e.g. by observing that the elements in the first row are different).

\subsubsection{Data setting $C$}\label{subsec:datasettingC}
For this data set, we consider two factors $\mathcal{I}$ and $\mathcal{J}$ with three levels each, i.e. $\mathcal I = \mathcal J = \{1,2,3\}$, leading to $k=9$ groups.  The set of hazard functions $\{\lambda_{ij},i\in \mathcal I, j\in \mathcal J\}$ is given in terms of the additive model described in \cref{eqn:AalenDecomposition0}, where $\lambda_0(t) = (29/8) + \theta$, 
\begin{align*}
  \phi_i(t) = \begin{cases}
 -\frac{5}{24}+ \frac{3x}{2(1+x^2)} & i = 1\\
  \frac{13}{24}-\frac{3x}{2(1+x^2)}& i=2\\
   -\frac{8}{24}& i=3 \end{cases},
    \quad \psi_j(t) = \begin{cases} 
      -1/2& j=1 \\
      0 & j=2 \\
      1/2    & j=3 
   \end{cases}, \quad\text{and}\quad \sigma_{ij}(t) = \begin{cases} \frac{5}{6}\theta & i=1,j=2\\
   -\theta/6 & \text{otherwise}
   \end{cases}.
\end{align*}
Note that \cref{eqn:AalenDecomposition0} is described in terms of cumulative hazards, whereas the equations above describe the hazard function, however both descriptions are equivalent. The constant $\theta$ is chosen such that $\theta \geq -1$ which ensures that $\lambda_{ij}(t)\geq 0$ for all $t\geq 0$, and for all $i,j\in[3]$. Note that the model has interaction terms if and only if $\theta \neq 0$, so testing for no interaction is equivalent to test for $\theta  = 0$. 

Similarly to what we did for data A and B, we consider three censoring regimes low (5-20\% of censored observations), medium (20-50\%), and high (40-60\%). This time however, we make censoring depend on the factor $\mathcal I$. A detailed explanation of the censoring mechanism and the specific percentages of censored observation per group are given in \Cref{App:censoringExplanation}.

We perform experiments under balanced and unbalanced data-settings. In the \textbf{balanced} setting all groups have the same sample size, whereas in the \textbf{unbalanced} setting, groups have different size. In particular, for the unbalanced data setting, the size of each group  $(1,1)$, $(2,1)$, $(3,1)$, $(1,2)$, $(2,2)$, $(3,2)$,$ (1,3)$, $(2,3)$ and $(3,3)$ is chosen proportional to  $15,9,5,9,7,6,8,5$ and $11$, and in order to modify the size of the groups we apply a multiplicative factor (and take floor if the new size is not integer) to obtain new sample sizes.

The main objective of this simulated data set is to evaluate our methods in the task of detecting interaction terms. Then, the null hypothesis is that there is no interaction terms, that is, $H_0:\Sigma_{ij}=0$ for all $i,j\in \{1,2,3\}$ {which is equivalent to $\theta =0$.} Note moreover that, intuitively, the larger is $|\theta|$ the easier it should be to reject the null hypothesis.

\subsubsection{Hyperparameter}\label{sec:hyperparameters} For the implementation of our method we use a kernel $K$ that factorises into the product of two kernels. In particular, we choose $K((t,x),(s,y))=K_{SE}(t,s)K_{RQ}(x,y)$ where $K_{SE}(t,s)=\exp\left(-\frac{\|t-s\|^2}{\ell^2}\right)$ and $K_{RQ}(x,y)=\left(1+\frac{\|x-y\|^2}{2ab^2}\right)^{-a}$
are the squared-exponential kernel and the rational quadratic kernel, respectively. The squared-exponential kernel is used to model survival times, and the rational quadratic kernel is used to model the group labels. We implement 5 different kernel test-statistics which we denote by $K_1$ to $K_5$.  Each test-statistic uses the same parameters $a=2$ and $b=1$, and the length-scale parameter $\ell^2$ of $K_1$ to $K_5$ are chosen as 10, 1, 0.1, 0.05, and 0.02 respectively.

For the Multiple Contrast test (denoted by $M$ in the experiments), the length-scale parameter used is $\ell^2=0.1$. In all cases we consider  one matrix per equation that defines the global null hypothesis.

Finally, to estimate the power of our test we perform 1000 independent repetitions of each experiments (fpr each combination of censoring and sample sizes). Also, for our Wild Bootstrap sample scheme we use $M=1000$ independent samples to approximate the rejection region.

\subsubsection{Comparison}\label{sec:casa}
We compare our methods with the very recent CASANOVA permutation test proposed by \cite{ditzhaus2020permutation}, which is a permutation-based test for {the global null hypothesis only}. Such is based on combining several weighted statistics into a single test-statistic. We denote by \textbf{Per2} the CASANOVA test based on the weight functions $\{1,x\}$ and by \textbf{Per4} the test that uses $\{1,x,x^2,x^3\}$, which are the ones used by the authors in their original experiments. We use the implementation provided by the authors, {which is also available in the R-package GFDsurv \citep{GFDsurv}}, and to approximate rejection region we use 1000 permutations as suggested by the authors.

\subsection{Results under Null Distributions}\label{sec:exp_results_H0}
In our first experiment, we aim to show that all the methods evaluated attain a correct Type-I error -which in our experiments is fixed to the value $\alpha=0.05$- for all the null hypotheses studied.  Recall that, as previously discussed, both data sets $A$ and $B$ satisfy the null hypothesis corresponding to `\textit{there is no main effect of the factor $\mathcal{I}$}', and note that the data set $C$, generated according the description in \cref{subsec:datasettingC} with $\theta=0$, satisfies the global null hypothesis corresponding to `\textit{there is no interaction effect between factors $\mathcal{I}$ and $\mathcal{J}$}'. \Cref{table:nullBalanced} and \Cref{table:nullUnbalanced} of \Cref{appendix:ExperimentsNull} show the rejection rates obtained for each of these null hypotheses, for each of the censoring regimes considered, and in the balanced and unbalanced data-settings. Overall, the kernel-based tests attain the correct Type-I error  for most combinations of sample sizes and censoring, for both the balanced and unbalanced data-settings, from which we deduce that the tests are well-calibrated. The  CASANOVA tests also attain a rejection rate of approximately $\alpha$, however they tend to fail slightly more than the kernel approach, i.e., we tend to observe more rejections rate above the level.

\subsection{Results under Alternative Distributions}\label{sec:experiments_alternative}
We analyse the performance of our methods when data is generated under a visible alternative. In such scenario, we expect  rejection rates- which coincide with the power of the test- to be as large as possible. Within this context, we study the null hypothesis that `\textit{there is no an effect of the factor $\mathcal{I}$}', which we already discussed it does not hold for data sets A and B. \Cref{Fig:DataAEI} and \Cref{Fig:DataBEI} show the power obtained for the data sets $A$ and $B$, respectively, in the unbalanced data-setting. For data A, where hazard functions are more simple (indeed, they are constant over time), the best results are obtained by  kernel-based tests that use larger length-scale parameters, whereas for data B, where hazards are generated using more complex functions, the overall best results are associated with smaller length-scale parameters. As expected, all  test increase their power as the sample size increases, and the power decreases as the amount of censoring increases for a fixed sample size. Results for balanced groups are given in the \Cref{Appendix:ExperimentsAlternativeAB}.

\begin{figure}
\centering
\includegraphics[scale=\size]{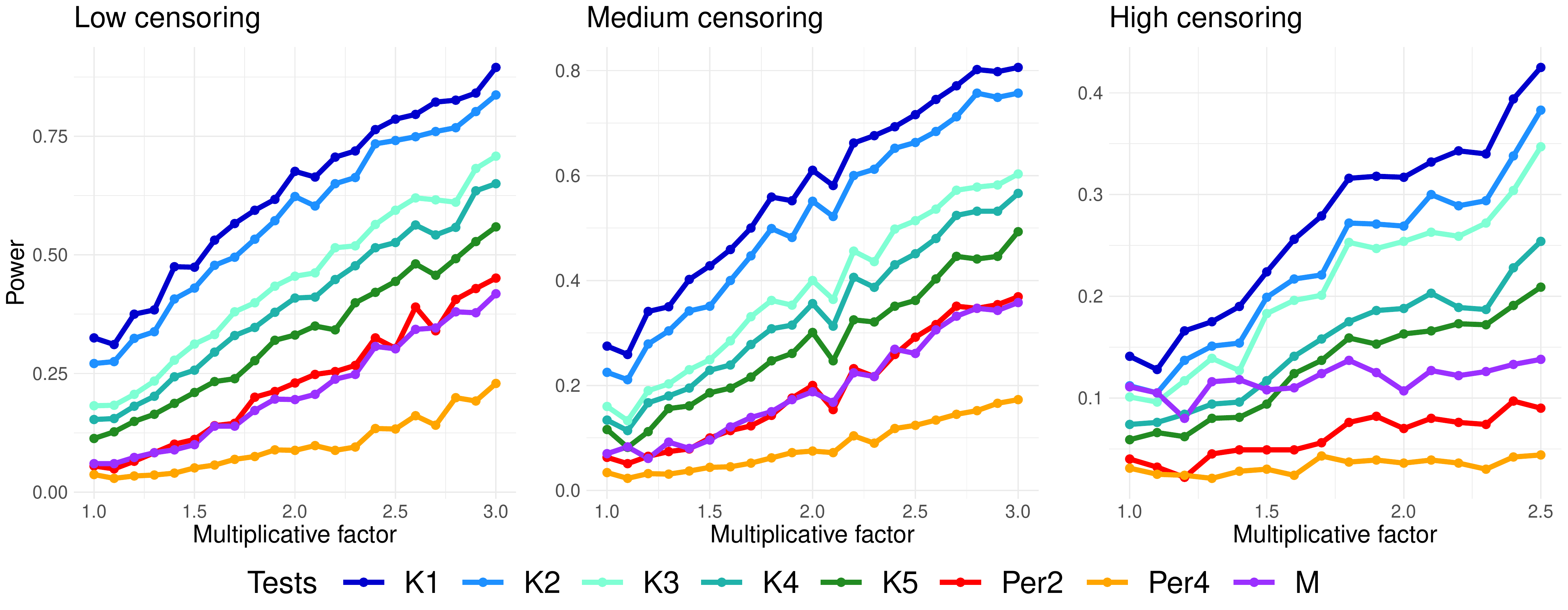}\caption{Test power versus sample size for the hypothesis there is no  effect of the factor $\mathcal{I}$ in Data $A$ in the unbalanced case.}\label{Fig:DataAEI}
\end{figure}

\begin{figure}
\centering
\includegraphics[scale=\size]{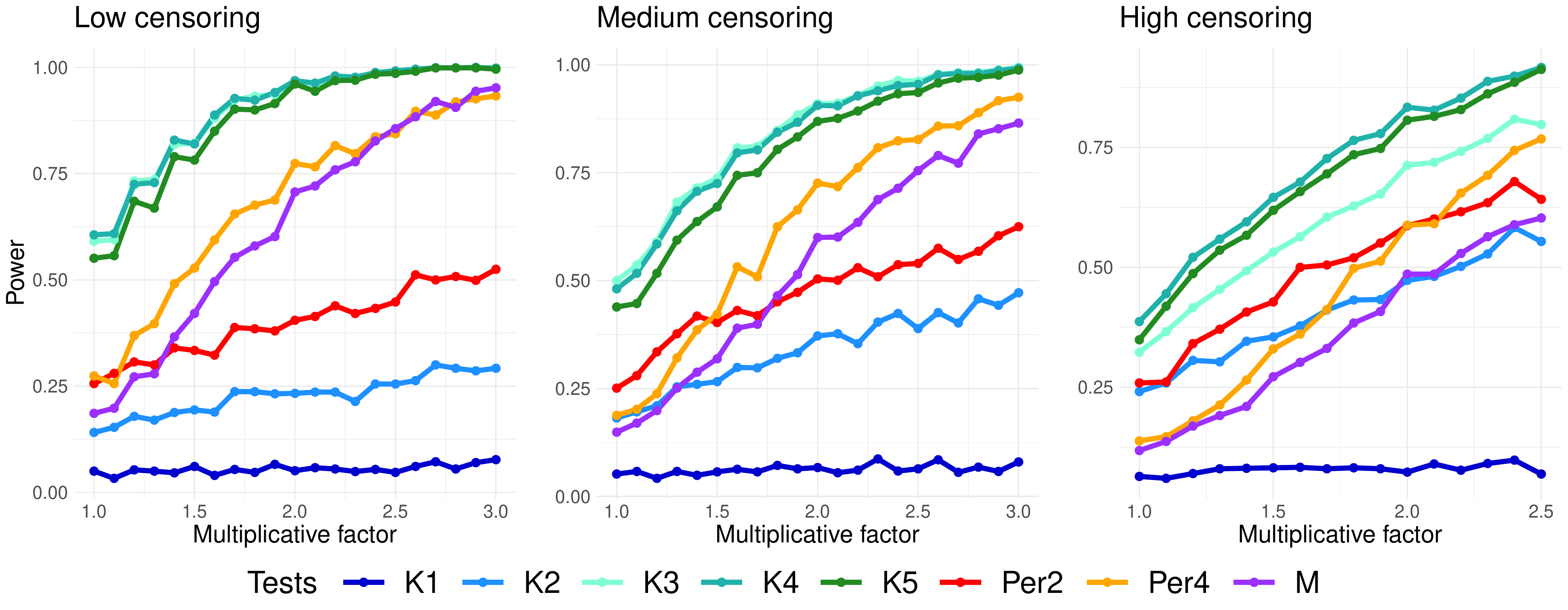}\caption{Test power versus sample size for the hypothesis there is no  effect of the factor $\mathcal{I}$ in Data $B$ in the unbalanced case.}\label{Fig:DataBEI}
\end{figure}

For data C, the we still consider the null hypothesis  `\textit{that there is no interaction effect between factors $\mathcal{I}$ and $\mathcal{J}$}'. For this experiments we want to assess how the rejection rate grows as the sample size increases under the alternative hypothesis. For this, we generate under the alternative hypothesis using $\theta=2$ and $\theta=1$. Recall that $\theta\neq 0$ controls the strength of the interaction, and indeed, the larger $|\theta|$ it should be easier to reject the null hypothesis.  \Cref{Fig:DataCt2Unb} shows the rejection rates obtained for $\theta=2$ in the unbalanced data-setting. As expected the rejection rates increase with the sample size, and decrease for a fixed sample size when censoring increases. Overall the best performance is attained by the kernel tests and, in particular, the performance improves from smaller to larger length-scale parameters.  Results for the balanced data setting when $\theta=2$, and $\theta=1$ are given in \Cref{Appendix:ExperimentsAlternativeC}

\begin{figure}
\centering
\includegraphics[scale=\size]{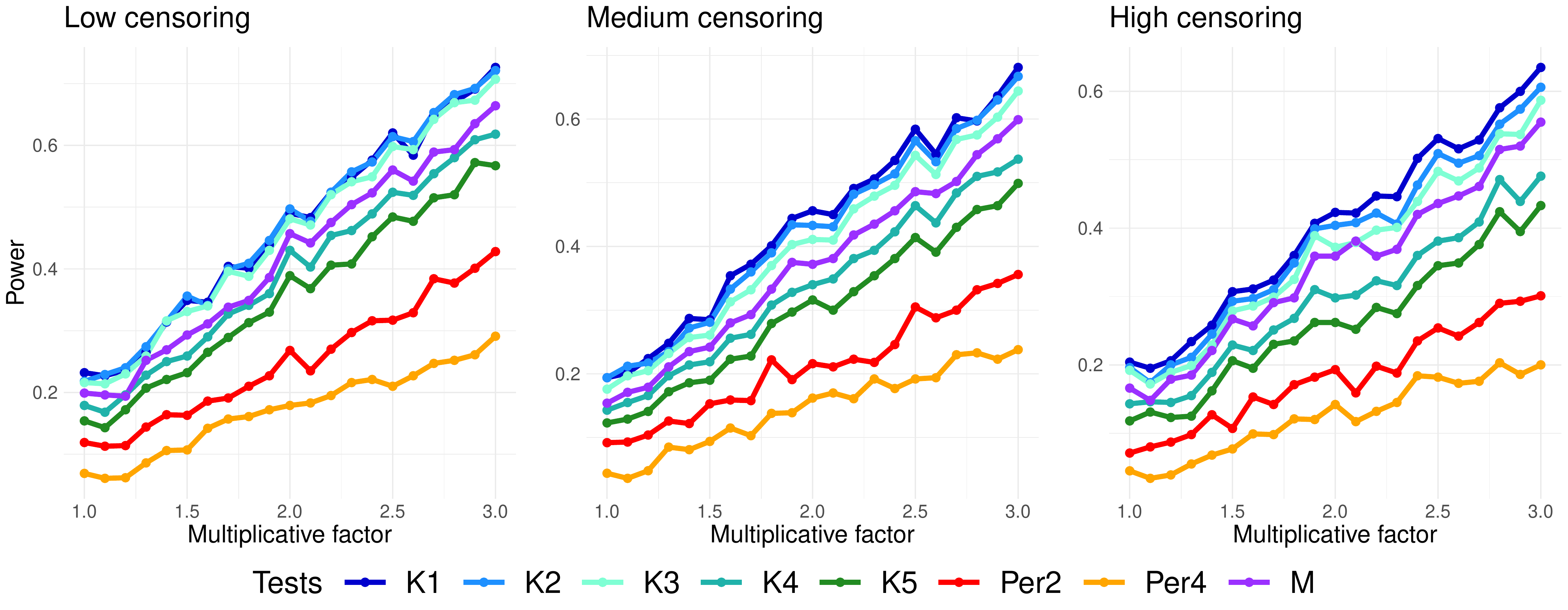}\caption{Test power versus sample size for the hypothesis there is no interaction term, equivalent to test $\theta = 0$,  when the alternative $\theta = 2$ holds in Data $C$ in the unbalanced setting. }\label{Fig:DataCt2Unb}
\end{figure}

Additionally,  since the size $|\theta|$ controls the strength of the interaction, we analyse how the power of the tests change for small deviations from the null (i.e. when $\theta$ moves away from $0$). For that, we fix the sample size of the groups, and analyse the power of our test for values of $\theta \in [-1,2]$. \Cref{Fig:DataCtvarUnb} shows the power of the tests on the three censoring regimes in the unbalanced data-setting. We can see that when $\theta = 0$ all the tests show a rejection rate very close to the Type-I error $\alpha=0.05$ (which agrees with our previous statement that all tests are well-calibrated), and when $\theta$ moves away from $0$ - meaning the strength of the interactions becomes more visible- the rejection rate (which can be interpreted as the power since the alternative holds) increases. Note that in this experiment kernel tests gain power much faster than the CASANOVA approach.

\begin{figure}
\centering
\includegraphics[scale=\size]{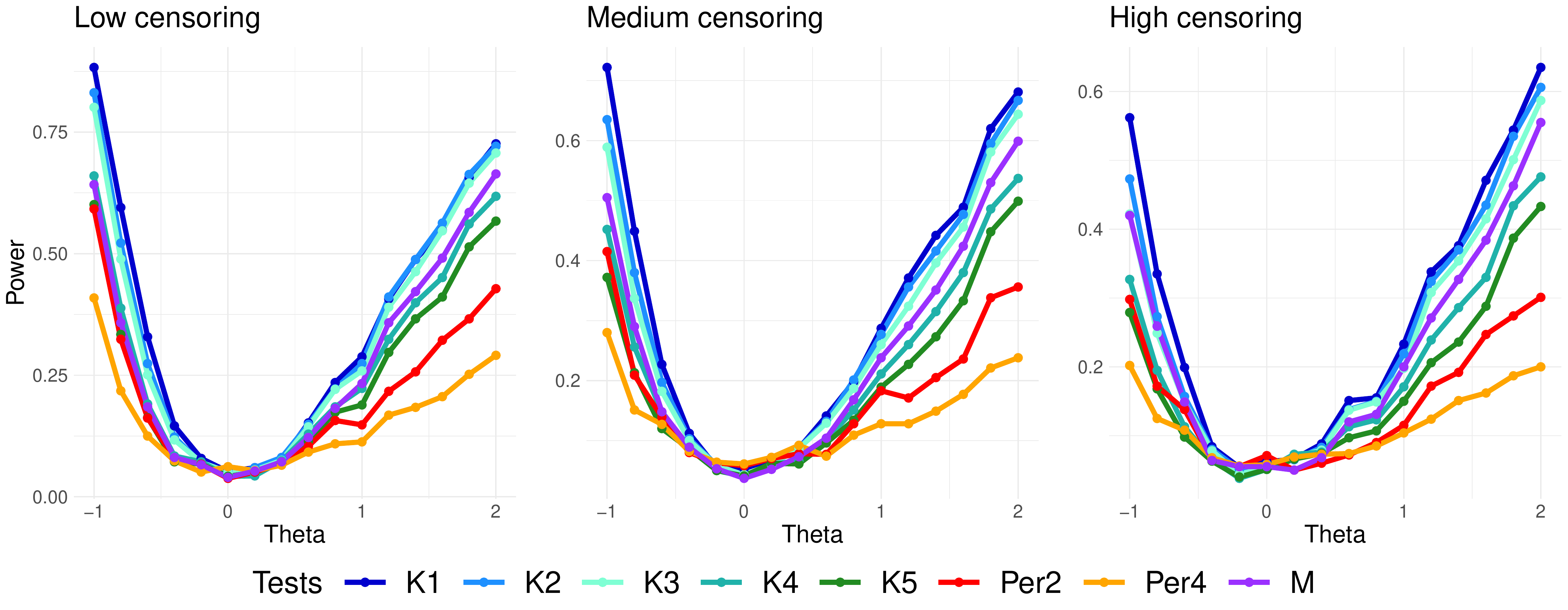}\caption{Test power versus $\theta$ for the hypothesis there is no interaction term in Data $C$. The sample size of the nine groups are 45, 27, 15, 27, 21, 18, 24, 15, and 33.}\label{Fig:DataCtvarUnb}
\end{figure}

From the observed results we conclude the following:
\begin{enumerate}
    \item For problems with simple hazard structure (smooth functions with small fluctuations), such as data set $A$ and $C$, simple methods tend to perform better. In these cases kernels with large length-scale parameter, which tend to be very flat perform very well. This is consistent with the fact that indeed, those kernels are associated with space of very flat functions, whereas kernels with small length-scale are associated with spaces of functions with a lot of fluctuations.  Something similar can be observed in the CASANOVA procedure, where considering  2 functions (Per2) works better than considering 4 functions (Per4)
    \item On the other hand, for problems with complex hazard structures (functions with a lot of fluctuations), such as Data set $B$, we need to use more complex structures. In these cases, kernels with small length-scale parameter have a very good performance, whereas kernels with large length-scale perform very poorly, in particular $K_1$, which length parameter is $10$, has almost no power at all.
    \item In general, we observe that kernel methods has much better performance than the CASANOVA in the settings considered. This difference seems to be larger for unbalanced group sizes.
   
\end{enumerate}

\subsection{Results for the Multiple Contrast Test}\label{sec:experiment_marginal}

Recall that the Multiple Contrast test not only is able to reject the (global) null hypothesis, but also it is able to distinguish which equations of the system $\CMa \bLambda(t) = \bzero$ do not hold. 

For the data sets $A$ and $B$, for the  null hypothesis that there is no interaction effect of $\mathcal I$, the contrast matrix $\CMa$ has three rows, and so three equations compose the null hypothesis. Thus, for the Multiple Contrast test we consider three contrast vectors $\cm_1,\cm_2$ and $\cm_3$, corresponding to the local hypotheses:
\begin{align*}
H_{01}:\cm_1\bLambda(t)=0,\qquad H_{02}:\cm_2\bLambda(t)=0\qquad\text{and}\qquad H_{03}:\cm_3\bLambda(t)=0,
\end{align*}
which together test the global null hypothesis that there is no effect of the factor $\mathcal{I}$. In particular, we choose $\cm_i$ such that the hypothesis $H_{0i}$ corresponds to $\Lambda_{1i}(t)=\Lambda_{2i}(t)$, for $i \in \{1,2,3\}$. Notice there are seven ways of rejecting the global null hypothesis: (1) only $H_{01}$ is false; (2) only $H_{02}$ is false; (1,2) only $H_{01}$ and $H_{02}$ are false; (3) only $H_{03}$ is false; (1,3) only $H_{01}$ and $H_{03}$ are false; (2,3) only $H_{02}$ and $H_{03}$ are false; and (1,2,3) all hypotheses are false.

In \Cref{BarplotDataA_Unb,BarplotDataB_Unb} we show  how the power of the Multiple Contrast test splits among the different ways of rejecting the global null
in the unbalanced sample size setting (Figures for the balanced case can be found in~\Cref{Appendix:ExperimentsAlternativeAB}). In particular, we observe that as the sample size grows, the Multiple Contrast test correctly identifies that only $H_{01}$ and $H_{02}$ are false, which is represented by (1,2). Indeed, most of the power is allocated to: i) hypothesis $H_{01}$ is false, ii) hypothesis $H_{02}$ is false, and iii) both  $H_{01}$ and  $H_{02}$ are false, and the more data we have the more the test is able to identify that both  $H_{01}$ and  $H_{02}$ are false at the same time. Observe as well that the problem becomes harder the more censoring we consider.

\begin{figure}
\centering
\includegraphics[scale=\size]{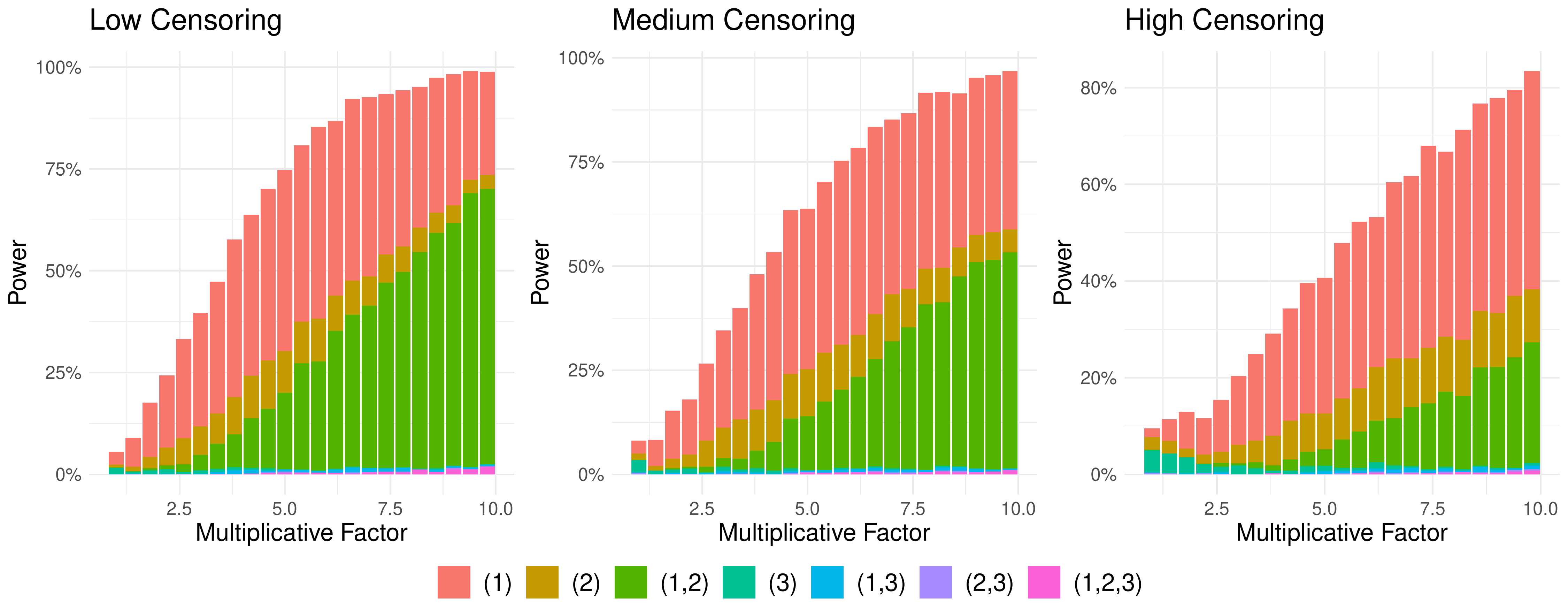}\caption{Distribution of the power attained by the Multiple Contrast test in data setting $A$. The sample size of the 6 groups are multiples of $15,9,5,9,7$ and $6$ (taking floor if needed).}\label{BarplotDataA_Unb}
\end{figure}

\begin{figure}
\centering
\includegraphics[scale=\size]{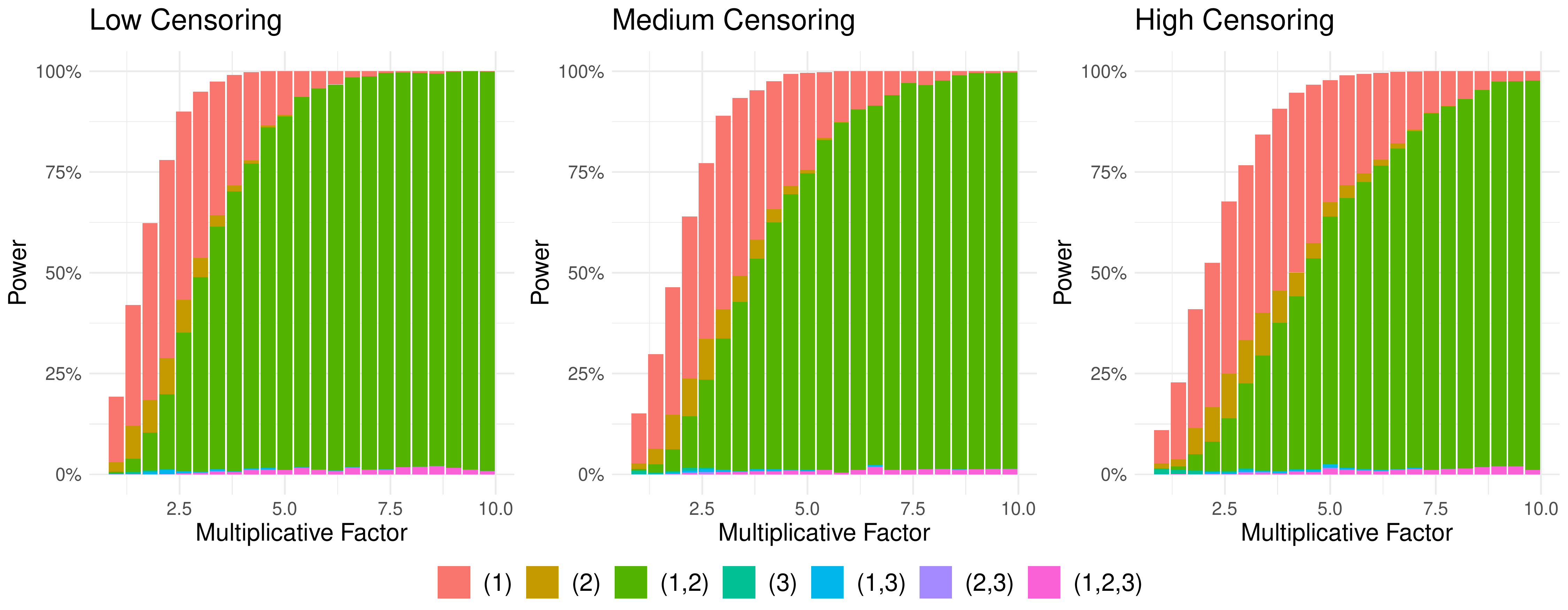}\caption{Distribution of the power attained by the Multiple Contrast test in data setting $B$. The sample size of the 6 groups are multiples of $15,9,5,9,7$ and $6$ (taking floor if needed).}\label{BarplotDataB_Unb}
\end{figure}

In the case of data set $C$ we test that there is no interaction term. The global null hypothesis $H_0$ consists of 9 equations, one per each combination of the factors, i.e. $9$ local hypotheses. In the case the null is false, i.e. $\theta \neq 0$, all 9 local hypotheses are false. \Cref{BarplotDataC_t2_Unb} shows how the power of the proposed Multiple Contrast test is spread among the different ways of rejecting the global null hypothesis in the unbalanced case when $\theta = 2$ (for the balanced case and other values of $\theta$ we refer the reader to \Cref{Appendix:ExperimentsAlternativeC}). Since there are $2^9-1$ ways to reject the null (choosing at least 1 hypothesis out of 9), we just report the number of local hypotheses that were rejected by the Multiple Contrast test. We can see that as the sample size increases, more hypothesis are being rejected.

\begin{figure}[h]
\centering
\includegraphics[scale=\size]{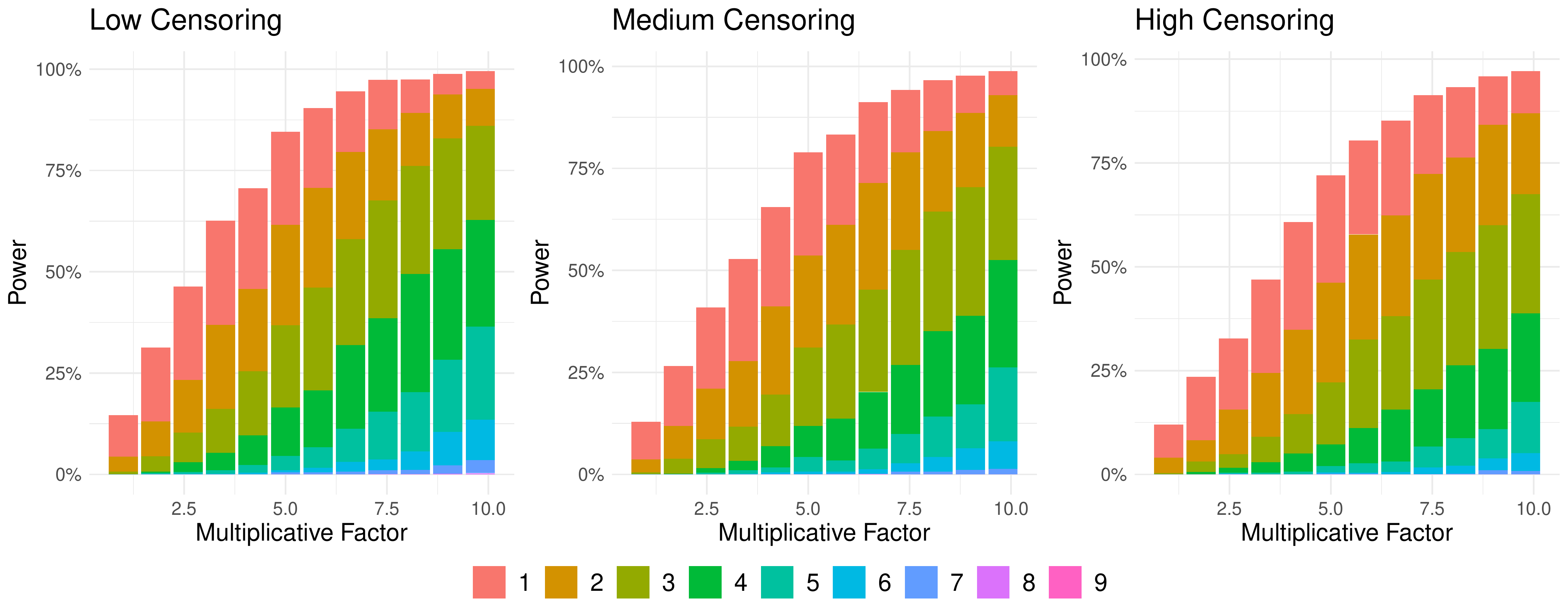}\caption{Distribution of the power attained by the Multiple Contrast test in data setting $C$ over the number of rejected local hypotheses for $\theta = 2$. The size of the 9 groups are multiples of  $15,9,5,9,7,6,8,5$ and $11$, respectively. 
\label{BarplotDataC_t2_Unb}
}
\end{figure}

From our experiments we conclude that
\begin{enumerate}
    \item As predicted by our theoretical results, the test is actually identifying the local hypotheses that are actually false.
    \item The method is rather data expensive, and it should be used in settings where a lot of data points are available.
\end{enumerate}

\section{Real experiments}
We consider the data corresponding to the lung cancer study from Pretince (1978), which includes information about the survival times of male patients with advanced inoperable lung cancer, the treatment given to those patients: standard therapy or a test chemotherapy, and the histological type of their tumors: smallcell, adeno, large and squamous. 
\begin{table}[h]
\centering
\begin{tabular}{|l|cccc|}
\hline
             &\multicolumn{4}{c|}{Factor $\mathcal{J}$: Celltype}\\
    Factor $\mathcal{I}$: Treatment &1: smallcell &2: adeno & 3: large&4: squamous\\
    \hline
  1(standard) &          30 &    9 &   15&15\\
  2(test) &           18 &   18 &   12&20\\
  \hline
\end{tabular}
\caption{Sample sizes for each group $(i,j)$ with $i\in\{1,2\}$ and $j\in\{1,2,3,4\}$.}\label{Table:real}
\end{table}

For this data we consider a $2\times 4$ factorial design where the first factor $\mathcal{I}$ (with levels $1,2$) corresponds to the treatment and the second factor $\mathcal{J}$ (with levels $1,2,3,4$) is the celltype.  Information about the factors and a summary of the sample size for each group is given in \Cref{Table:real}. We apply the kernel methods as well as CASANOVA to test five global null hypotheses: i) Trt (ME): there is no main effect of the treatment, ii) Trt (E): there is no effect of the treatment, iii) Celltype (ME): there is no main effect of the celltype, iv) Celltype (E): there is no effect of the celltype, and iv) Interaction: there is no an interaction effect of the treatment and celltype. We refer the reader to \Cref{sec:FactorialDesigns} to recall the different null hypotheses in the setting of factorial designs.

The specification of the tests is as follows. For the kernel methods we implement the kernels $K_1$-$K_5$ used in our simulated data experiments with the hyperparameters as described in \cref{sec:hyperparameters}. For the CASANOVA test, we implement Per2 and Per4 described in \Cref{sec:casa}, but we also include two extra tests: one including only the weight function $\omega(x) = 1$ (denoted by LR in our experiments), and other only with the weight function $\omega(x) = 1-2x$ (denoted by CROSS) which is helpful in data generated by a hazard functions that cross around the median of the data.

In \Cref{Table:realresults} we show the $p$-values in percentages. We notice that at a level $\alpha=0.05$ (or $5\%$), all the tests reject the null hypotheses Celltype (ME) and Celltype (E), from which we can infer that there is a  main effect and a effect of the celltype on the survival times. In both cases the smallest $p$-values are attained by the kernel tests. For the hypotheses Trt (ME), Trt (E) and Interaction, all tests agree on not rejecting the null hypothesis at a level $\alpha=0.05$ (or $5\%$).

\begin{table}[]
    \centering
    \begin{tabular}{|l|lllll|}
    \hline
        &\multicolumn{5}{|c|}{Global null hypothesis}\\
         &Trt (ME)&Trt (E)&Celltype (ME)&Celltype (E)&Interaction\\  
    \hline
    LR&93.68&14.35&0.04&0.14&48.372\\
    CROSS&69.2&68.1&0.254&4.104&85.423\\
    \hline
    \multicolumn{6}{|c|}{CASANOVA (permutation tests)}\\
    \hline
    Per2&52.01&16.3&1.298&0.185&24.459\\
    Per4&64.6&11.36&1.584&0.729&24.294\\
    \hline
    \multicolumn{6}{|c|}{Kernel tests}\\
    \hline
    K1&25.303&11.076&0.003&0.118&21.087\\
    K2&15.309&6.981&0.011&0.056&20.156\\
    K3&19.391&12.839&0.089&0.080&18.180\\
    K4&15.961&15.981&0.090&0.133&23.883\\
    K5&14.078&15.125&0.134&0.246&26.476\\
    \hline
    \end{tabular}
    \caption{$P$-values (in $\%$) attained by each test, for each of the global null hypotheses.
    }\label{Table:realresults}
    \label{tab:my_label}
\end{table}
Upon rejection of the global null hypothesis Celltype (E), we proceed to apply the Multiple Contrast test to obtain more information about which component in the contrast matrix $\CMa$ associated  to this problem is being rejected. To implement the Multiple Contrast test we choose a squared exponential kernel for the times with length-scale parameter $\ell^2=0.1$, and a rational quadratic kernel for the group labels with parameters $a=2$ and $b=1$.

Let $\Lambda_{ij}$ be the risk function associated the treatment $i\in\{1,2\}$ and the celltype $j\in\{1,2,3,4\}$ (see \Cref{Table:real}). Then the global null hypothesis Celltype (E) is equivalent to test each single hypothesis $\Lambda_{11}=\Lambda_{1j}$, and $\Lambda_{21}=\Lambda_{2j}$, for all $j\in\{2,3,4\}$.
\begin{table}
\centering
\begin{tabular}{|l|cccccc|}
\hline
 &\multicolumn{6}{c|}{Global null: Celltype (E)}\\
Local   $H_{0i}$&$\Lambda_{11}=\Lambda_{12}$&$\Lambda_{11}=\Lambda_{13}$&$\Lambda_{11}=\Lambda_{14}$&$\Lambda_{21}=\Lambda_{22}$&$\Lambda_{21}=\Lambda_{23}$&$\Lambda_{21}=\Lambda_{24}$\\
     \hline
     $\Upsilon_n(\CMa_i)$&0.007&0.823&0.282&0.017&0.337&0.816\\
       $\qwei{\widehat\beta_{\alpha}}$  &0.266&0.430&0.581 &0.309&0.486&0.723\\
     Reject$_i$&No&Yes&No&No&No&Yes\\
     $p$-value$_i$&85.239&0.025&8.054&64.101&3.362&0.559\\
\hline
\end{tabular}
    \caption{$P$-values (in $\%$) attained by the Multiple Contrast test for each local null hypothesis. Reported $\widehat\beta_{\alpha}=0.949\%$ for $\alpha=5\%$}
    \label{tab:my_label}
\end{table}

\begin{figure}[H]
    \centering
\includegraphics[scale=0.27]{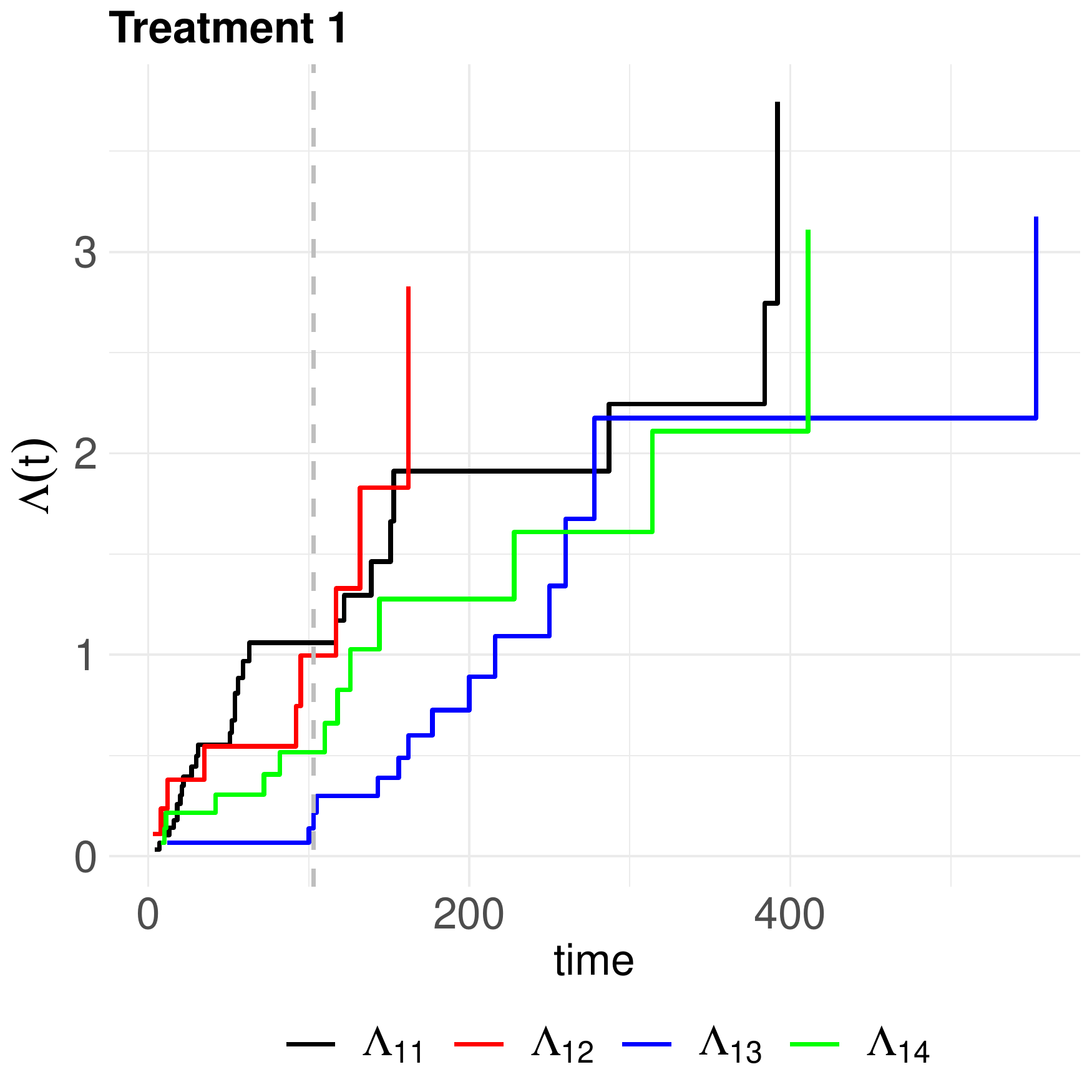}
\includegraphics[scale=0.27]{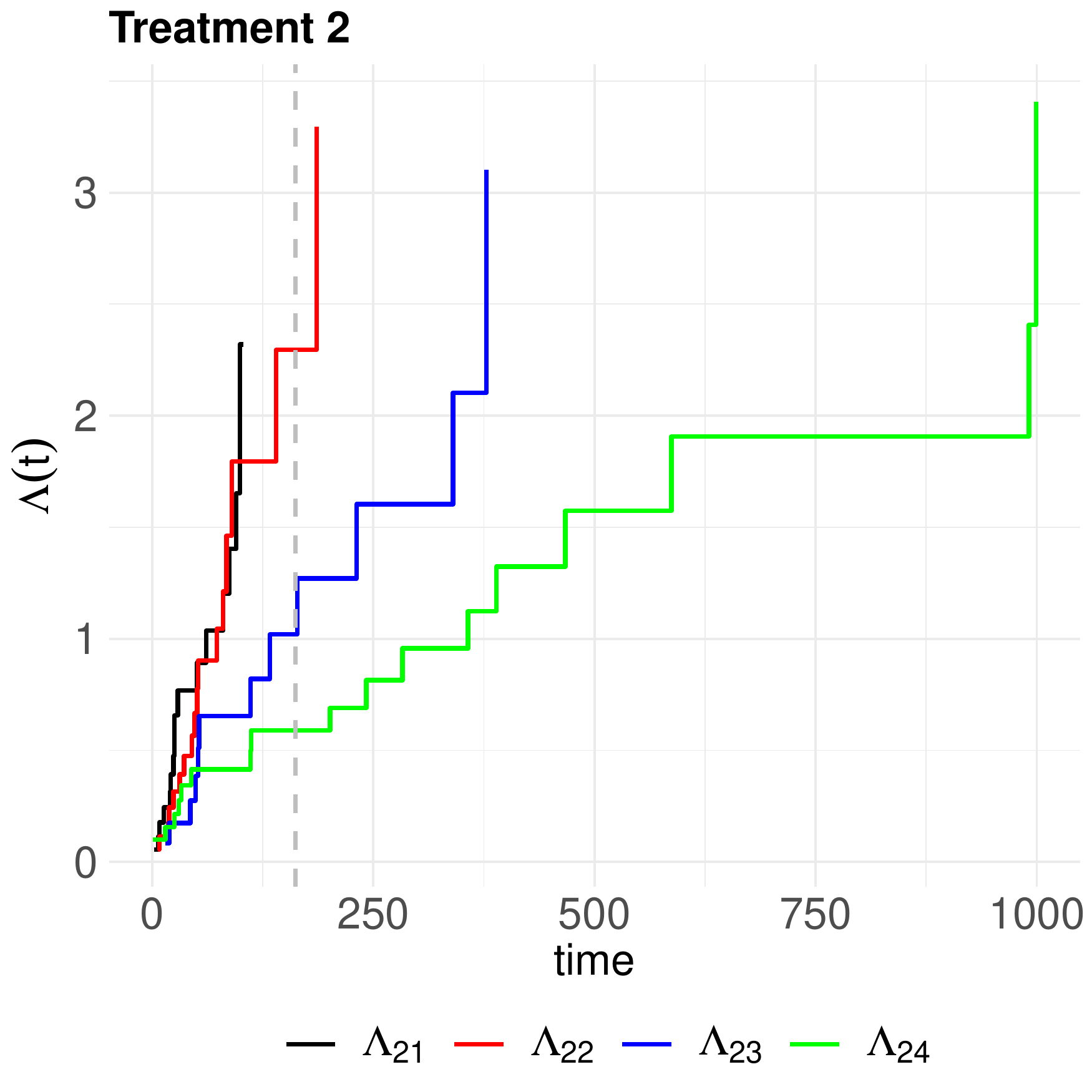}
    \caption{Cumulative hazard functions estimated for each Treatmeant and Celltype (ranging from 1:smallcell, 2:adeno, 3:large and 4:squamous). In both plots the dashed line represents $\tau_H$.}
    \label{fig:my_label}
\end{figure}
 In \Cref{tab:my_label} we show the results obtained by the Multiple Contrast Test. We report the test-statistics $\Upsilon_n(\CMa_i)$, the quantiles $\qwei{\widehat{\beta}_\alpha}$, the individual $p$-values (in \%), and the decision made for each single hypothesis $H_{0i}$. The quantiles $\qwei{\widehat{\beta}_\alpha}$ where computed using $M=100000$ Wild Bootstrap samples, $\alpha=0.05$ and the estimate $\widehat\beta_\alpha=0.949$. Notice that the decision of rejecting each individual null hypothesis $H_{0i}$ is made whenever $\Upsilon_{n}(\CMa_i)>\qwei{\widehat{\beta}_\alpha}$ or, equivalently, when $\text{$p$-value}$ {(in $\%$)} associated to $\CMa_i$ is smaller than $\widehat\beta_\alpha=0.949{\%}$.

From  \Cref{tab:my_label}, note that we reject the hypotheses: $\Lambda_{11}=\Lambda_{13}$ and $\Lambda_{21}=\Lambda_{24}$. These results align with what is being shown in   \Cref{fig:my_label}, as we can observe that for treatment 1, $\Lambda_{11}$ (in black) and $\Lambda_{13}$ (in blue) appear to be the most dissimilar risk functions. The same conclusion can be drawn for the risk functions $\Lambda_{21}$ (in black) and $\Lambda_{24}$ (in green) for treatment 2. The highest $p$-values are obtained for the hypotheses: $\Lambda_{11}=\Lambda_{12}$ and $\Lambda_{21}=\Lambda_{12}$, which in \Cref{tab:my_label} appear to be the most similar pairs of risk functions (black vs red for both treatments). Finally, the hypotheses $\Lambda_{11}=\Lambda_{14}$ and $\Lambda_{21}=\Lambda_{23}$ are not rejected at level $\alpha=0.05$ (or 5\%) for the global test Celltype (E) even though visually we might think they should be rejected. We believe this is a possible defect of our estimation method as it ignores all the data point after the time represented by the dashed grey line in \Cref{fig:my_label} due to the technical condition that the matrix $\X(t)=\frac{1}{n}\Diag{\Y}{t}\V$ has to have full rank (and we do not have such property after such times), and so several points are being ignored, especially the ones involving $\Lambda_{13}$, $\Lambda_{14}$, $\Lambda_{23}$, and $\Lambda_{24}$. Fixing this issue will be part of future work on this topic.

\section{Conclusion}
In this paper we introduced a novel nonparametric method for testing in the factorial design setting for survival data with the aim of dealing with complex hazard alternatives and perform a post-hoc analysis in order to not only reject a hypothesis but to also give reasons why such hypothesis fails to hold. Our methodology combined a novel weighted log-rank statistic for factorial designs with the state-of-the-art kernel-based testing approach, leading to a test that is powerful, robust and able to deal with complex hazard structures. In order to perform an in-depth analysis and finding reasons why the a (global) hypothesis is rejected, we extended our previous construction to a multiple contrast test that is able to test a hypothesis by analysing several local hypothesis,  distinguishing between the ones that are rejected and not. We experimentally showed that our method is rather powerful and very robust, being able to deal with complex hazard functions, including hazards with multiple crossings, and intricate dependence on the factors. We also provided asymptotic guarantees for our methods.

Future research ideas include finding other log-rank-type statistic that can be `kernelised' instead of the one we proposed. That may give better results in practice, as well as extending the kernel approach to other practical problems in Survival Analysis such as testing proportionality, which cannot be posed as a factorial design testing problem, among others. Extending our approach to continuous covariates is of interest as well. From the kernel testing approach, exploring our multiple testing approach in other setting may lead to interesting developments.

\newpage

\appendix
\section{Deferred Proofs}\label{app:proofs} 
We give proof to the results stated in the main document. We begin by introduce some preliminary results that feature our proofs.

\subsection{Preliminary Results}

\begin{proposition}\label{prop:projection}
Let $\mathbf{P}$ be a $k\times k$ orthogonal projection matrix. Then, for any $i,j\in[k]$, it holds $|P_{ij}|\leq 1$.
\end{proposition}

\begin{proof}
Let $(\boldsymbol{\delta}_i)_{i=1^k}$ be the canonical orthonormal base of $\R^k$. Using that orthonomal projection matrices have $\|P\|_2=1$ we get
$$
|\mathbf{P}_{ij}|=|\langle\boldsymbol{\delta}_i,\mathbf{P}\boldsymbol{\delta}_j\rangle_{2}|\leq \|\boldsymbol{\delta}_i\|_2\|\mathbf{P}\boldsymbol{\delta}_j\|_2\leq\|\mathbf{P}\|_2\|\boldsymbol{\delta}_j\|_2=\|\mathbf{P}\|_2=1.$$\end{proof}

\begin{proposition}\label{prop:ProbBounds}
Let $U_i$ be random variables taking values on $[0,\infty)$ with cumulative distribution function $H$, and let $Y(t) = \sum_{i=1}^n \ind_{\{U_i\geq t\}}$. Let $\beta \in (0,1)$, then
\begin{enumerate}
\item[i)] $\Prob\left(Y(t)/n \leq \beta^{-1}(1-H(t-)),\quad\forall t\leq\tau_n\right)\geq1-\beta,$ and \label{prop:ProbBound2}
\item[ii)] $\Prob\left(Y(t)/n \geq \beta(1-H(t-)),\quad\forall t\leq\tau_n\right)\geq1-e(1/\beta) e^{-1/\beta}.$ \label{prop:ProbBound3} 
\end{enumerate}
That is, i) and  ii) deduce $\sup_{t\leq \tau_n}Y(t)/(n(1-H(t-)) = \Theta_p(1)$.
\end{proposition}
The proofs of item ii) is due \citet{gill1983large}, and  item i) follows from \citet[Theorem 3.2.1]{gill1980censoring}.

\subsubsection{Advanced and Backward operators}
Recall that in our setting we observe censored data $D_i = (T_i,\Delta_i, X_i)$. The censored mechanism is due to the existence of a censored time $C_i$ that might not allow us to observe the time of interest $Z_i$ (as explained in \cref{sec:framework}). While in practice we do not observe the triple $(Z_i,C_i,X_i)$ (otherwise our setting is pointless), for our theoretical analysis is quite useful to work with this triple. We denote such triple with the letter $U_i$ and we call $\nu$ the probability measure in $\R\times\R\times[k]$ associated with it.

We define the operators $A,B:\mathcal{L}_2(\nu)\to\mathcal{L}_2(\nu)$ by
\begin{align*}
(Af_1)(z,c,x)&=f_1(z,c,x)-\frac{1}{S_x(z)}\int_z^\infty f_1(s,c,x)dF_x(s),\text{ and}\\
(Bf_1)(z,c,x)&=f_1(z,c,x)-\int_0^z f_1(s,c,x)d\Lambda_x(s),
\end{align*}
for any $f_1\in\mathcal{L}_2(\nu)$. The previous operators were introduced in \cite{ritov1988censoring, efron1990fisher}, and are known the advanced and backward operator, respectively.
\begin{proposition}\label{prop:fowardbackward} The operators defined above satisfy the following properties:
\begin{enumerate}
\item Let $f_1,f_2\in\mathcal{L}_2(\nu)$, then
\begin{align*}
    \int (Af_1)(u)f_2(u)\nu(du)=\int f_1(u)(Bf_2)(u)\nu(du)
\end{align*}
i.e. $A$ and $B$ are adjoint operators on $\mathcal L_2(\nu)$.
\item Let $f_1\in\mathcal{L}_2(\nu)$, then $ABf_1=f_1$. Moreover, if $\E(f_1(Z,C,X)|C,X)=0$ for almost all $(C,X)$, then
\begin{align*}
ABf_1=BAf_1=f_1.
\end{align*}
\item If $f_1\in\mathcal{L}_2(\nu)$ is such that $f_1(z,c,x)=\ind_{\{z\leq c\}}f(z,x)$ for some function $f$. Then
\begin{align*}
(Bf_1)(z,c,x)=\ind_{\{z\leq c\}}f(z,x)-\int_{0}^{\min\{z,c\}}f(s,x)d\Lambda_x(s).
\end{align*}
In particular, $(Bf_1)(Z_i,C_i,x_i)=\int_0^\infty f(s,x_i)dM_i(s)$. 
\end{enumerate}
\end{proposition}

\begin{proof}[\textbf{Proof of \Cref{prop:fowardbackward}}]
We start with property 1. Notice that
\begin{align*}
&\E((Af_1)(Z,C,X)f_2(Z,C,X)|C,X)\\
&\qquad=\int_0^\infty (Af_1)(z,C,X)f_2(z,C,X)dF_X(z)\\
&\qquad=\int_0^\infty \left(f_1(z,C,X)-\frac{1}{S_X(z)}\int_z^\infty f_1(s,C,X)dF_X(s)\right)f_2(z,C,X)dF_X(z)\\
&\qquad=\int_0^\infty f_1(z,C,X)f_2(z,C,X)dF_X(z)-\int_0^\infty\left(\frac{1}{S_X(z)}\int_z^\infty f_1(s,C,X)dF_X(s)\right)f_2(z,C,X)dF_X(z)\\
&\qquad=\int_0^\infty f_1(s,C,X)f_2(s,C,X)dF_X(s)-\int_0^\infty\left(\int_0^sf_2(z,C,X)\frac{dF_X(z)}{S_X(z)}\right)f_1(s,C,X)dF_X(s)\\
&\qquad=\int_0^\infty f_1(s,C,X)\left(f_2(s,C,X)-\int_0^sf_2(z,C,X)\frac{dF_X(z)}{S_X(z)}\right)dF_X(s)\\
&\qquad=\E(f_1(Z,C,X)(Bf_2)(Z,C,X)|C,X).
\end{align*}
The previous set of equations show that $\E((Af_1)(Z,C,X)f_2(Z,C,X)|C,X)=\E(f_1(Z,C,X)(Bf_2)(Z,C,X)|C,X)$, then by taking expectation in both sides we get the result stated in i).

For property 2., note
\begin{align*}
ABf_1(z,c,x)&=(Bf_1)(z,c,x)-\frac{1}{S_x(z)}\int_z^\infty (Bf_1)(s,c,x)dF_x(s).
\end{align*}
The first term of the right-hand-side of the previous equation satisfies
\begin{align}
(Bf_1)(z,c,x)&=f_1(z,c,x)-\int_0^zf_1(s,c,x)d\Lambda_x(s)\label{eqn:p21}, 
\end{align}
and the second term satisfies
\begin{align}
\int_0^z&f_1(s,c,x)d\Lambda_x(s)=\frac{1}{S_x(z)}\int_z^\infty (Bf_1)(s,c,x)dF_x(s)\nonumber\\
&\qquad=\frac{1}{S_x(z)}\int_z^\infty \left(f_1(s,c,x)-\int_0^sf_1(s',c,x)d\Lambda_x(s')\right)dF_x(s)\nonumber\\
&\qquad=\frac{1}{S_x(z)}\int_z^\infty f_1(s,c,x)dF_x(s)-\frac{1}{S_x(z)}\int_0^\infty \ind_{\{z\leq s\}}\left(\int_0^sf_1(s',c,x)d\Lambda_x(s')\right)dF_x(s)\nonumber\\
&\qquad=\frac{1}{S_x(z)}\int_z^\infty f_1(s,c,x)dF_x(s)-\frac{1}{S_x(z)}\int_0^\infty S_x(\max\{z,s'\})f_1(s',c,x)d\Lambda_x(s')\nonumber\\
&\qquad=-\int_0^z f_1(s',c,x)d\Lambda_x(s')\label{eqn:p22},
\end{align}
Thus, by substracting \cref{eqn:p21} and \cref{eqn:p22} leads to $ABf_1(z,c,x)=f_1(z,c,x)$.

We continue by computing 
\begin{align*}
(BAf_1)(z,c,x)&=(Af_1)(z,c,x)-\int_0^z (Af_1)(s,c,x)d\Lambda_x(s).
\end{align*}
Observe that the first term in the right-hand-side of the previous equation satisfies
\begin{align*}
(Af_1)(z,c,x)&=f_1(z,c,x)-\frac{1}{S_x(z)}\int_z^\infty f_1(s,c,x)dF_x(s).
\end{align*}
The second term satisfies
\begin{align*}
&\int_0^z (Af_1)(s,c,x)d\Lambda_x(s)\\
&\quad=\int_0^z \left(f_1(s,c,x)-\frac{1}{S_x(s)}\int_s^\infty f_1(s',c,x)dF_x(s')\right)d\Lambda_x(s)\\
&\quad=\int_0^zf_1(s,c,x)d\Lambda_x(s)-\int_0^\infty\ind_{\{s\leq z\}}\frac{1}{S_x(s)}\int_s^\infty f_1(s',c,x)dF_x(s')d\Lambda_x(s)\\
&\quad=\int_0^zf_1(s,c,x)d\Lambda_x(s)-\int_0^\infty f_1(s',c,x)\left(\int_0^{s'}\ind_{\{s\leq z\}}\frac{1}{S_x(s)^2}dF_x(s)\right)dF_x(s')\\
&\quad=\int_0^zf_1(s,c,x)d\Lambda_x(s)-\int_0^\infty f_1(s',c,x)\left(\frac{1}{S_x(\min\{z,s'\})}-1\right)dF_x(s')\\
&\quad=\int_0^zf_1(s,c,x)d\Lambda_x(s)-\int_0^z f_1(s',c,x)d\Lambda_x(s')-\frac{1}{S_x(z)}\int_z^\infty f_1(s',c,x)dF_x(s')-\E(f_1(Z,C,X)|C=c,X=x)\\
&\quad=-\frac{1}{S_x(z)}\int_z^\infty f_1(s',c,x)dF_x(s')+\E(f_1(Z,C,X)|C=c,X=x).
\end{align*}
By subtracting the second term to the first term we deduce $(BAf_1)(z,c,x)=f_1(z,c,x)-\E(f_1(Z,C,X)|C=c,X=x)$. Furthermore, if $\E(f_1(Z,C,X)|C=c,X=x)=0$ for almost all $(c,x)$, we conclude
\begin{align*}
(ABf_1)(z,c,x)=(BAf_1)(z,c,x)=f_1(z,c,x).
\end{align*}
Finally, we check property 3. Let $f_1(z,x,c)=\ind_{\{z\leq c\}}f(z,x)$, then
\begin{align*}
(Bf_1)(z,c,x)&=f_1(z,c,x)-\int_0^z f_1(s,c,x)d\Lambda_x(s)\\
&=\ind_{\{z\leq c\}}f(z,x)-\int_0^z \ind_{\{s\leq c\}}f(s,x)d\Lambda_x(s)\\
&=\ind_{\{z\leq c\}}f(z,x)-\int_0^{\min\{z,c\}}f(s,x)d\Lambda_x(s).
\end{align*}
Finally, consider the above equation using $(z,c,x) = (Z_i,C_i,X_i)$ and recall that $T_i = \min\{Z_i,C_i\}$ and $\Delta_i= \ind_{\{T_i\leq C_i\}}$, then 
\begin{align*}
(Bf_1)(Z_i,C_i,X_i)&= \Delta_i f(T_i,X_i)-\int_0^{T_i}f(s,X_i)d\Lambda_{X_i}(s)\\
&=\Delta_i f(T_i,X_i)-\int_0^{\infty}\ind_{\{s\leq T_i\}}f(s,X_i)d\Lambda_{X_i}(s)\\
&=\Delta_if(T_i,X_i)-\int_0^{\infty}Y_i(s)f(s,X_i)d\Lambda_{X_i}(s)\\
&=\int_0^\infty f(s,X_i)\left(dN_i(s)-Y_i(s)d\Lambda_{X_i}(s)\right)=\int_0^\infty f(s,X_i)dM_i(s).
\end{align*}
\end{proof}

\subsection{Proofs of \Cref{section:KernelTestStatistic}}\label{sec:asymptotic}

We recall some of the definitions of \Cref{sec:framework} and \Cref{sec:teststat} that recurrently appear in our proofs.

Observe that the matrix $\X(t)=\frac{1}{n}\Diag{\Y}{(t)}\V$ that has a.s. limit when the number of data points tends to infinity, and such limit is given by $\Xlim(t) = \Diag{\S}{(t)} \V$ where $\S(t) = (p_1(1-H_1(t)),\ldots, p_k (1-H_k(t)))$. The previous result holds since $\E(Y_i(t)|X_i=\ell) = p_{\ell}(1-H_{\ell}(t))$, and by the law of large numbers. Additionally, recall that $\Q(t)$ is the $k\times k$ matrix defined as
\begin{align*}
\Q(t)&=\IndRank(t)\left(\I_k-\X(t)(\X(t)^\transpose\X(t))^{-1}\X(t)^\transpose\right),
\end{align*}
where $\IndRank(t)$ is the indicator that $\X(t)$ has full rank. Note then that $\Q(t)$ converges a.s. to $\Qlim(t)$ given by
\begin{align*}
\Qlim(t) = \IndRanklim(t)\left(\I_k-\Xlim(t)(\Xlim(t)^\transpose\Xlim(t))^{-1}\Xlim(t)^\transpose\right)
\end{align*}
where $\IndRanklim(t)$ is the indicator that $\Xlim(t)$ has full rank. Observe that when the matrices are not trivially zero, $\Q$ and $ \Qlim$ are orthogonal projection matrices on the null space of $\X$ and $\Xlim$ respectively.

\begin{proof}[\textbf{Proof of \Cref{prop:closedform1}}]
By definition 
\begin{align*}
   \Upsilon_n(\CMa)=\sup_{\omega \in B_1(\mathcal H) }\widehat U_0(\omega)^2&=\sup_{\omega \in B_1(\mathcal H) }\left( \frac{1}{\sqrt{n}}\bone^\transpose\int \W(t)\Q(t)d\N(t)\right)^2\\
   &=\sup_{\omega \in B_1(\mathcal H) }\left(\frac{1}{\sqrt{n}}\int \sum_{i=1}^k \sum_{j=1}^k \omega(t,i) \Q_{ij}(t)d\N_j(t)\right)^2
\end{align*}
where the second equality is obtained by replacing $\widehat{U}_0(\omega)$ by its definition, given in \cref{eqn:ourLogrank}, and the third equality follows from expanding the matrix products. 

Since $\omega(t,i):[0,\infty)\times [k]\to \R$ is a element of the RKHS $\mathcal H$ with kernel $K$, we know that $\omega(t,i)=\langle\omega,K((t,i),\cdot)\rangle_{\mathcal H}$. Then, define the function $\xi \in \mathcal H$ as
\begin{align*}
    \widehat \xi(\cdot)  = \frac{1}{\sqrt n}\int \sum_{i=1}^k \sum_{j=1}^k K_{(t,i)}(\cdot) \Q_{ij}(t)d\N_j(t)
\end{align*}
which is well-defined since the integral  with respect to $d\N_j$ is no more than  a finite sum.
Then, clearly
\begin{align*}
     \Upsilon_n(\CMa)&=\sup_{\omega \in B_1(\mathcal H) }\left(\frac{1}{\sqrt{n}}\int \sum_{i=1}^k \sum_{j=1}^k \omega(t,i) \Q_{ij}(t)d\N_j(t)\right)^2=\sup_{\omega \in B_1(\mathcal H)}\langle \widehat \xi, \omega\rangle_{\mathcal H}^2= \|\widehat \xi\|^2,
\end{align*}
where the last equality follows from the fact that we are taking supremum over the unit ball of an RKHS. By carefully expanding $ \|\widehat \xi\|^2$ we obtain that
\begin{align}
    \|\widehat \xi\|^2 &= \frac{1}{n} \int\int \sum_{i,i'=1}^k\sum_{j,j'=1}^k K((t,i),(s,i'))\Q_{ij}(t)\Q_{i'j'}(s)dN_i(t)dN_{i'}(s) 
\end{align}
since $\langle K_{(t,i)}, K_{(s,i')}\rangle_{\mathcal H} = K((t,i),(s,i'))$. Finally, recall that $\K(t,s)$ represents the matrix $(\K_{i,i'}(t,s))_{i,i'=1}^k$ where $\K_{i,i'}(t,s) = K((t,i),(s,i'))$, hence, rearranging the previous equation yields
\begin{align}
     \Upsilon_n(\CMa)=\|\widehat \xi\|^2 &= \frac{1}{n}\int \int (\Q(t)d\N(t))^{\transpose} \K(t,s) (\Q(s)d\N(s))).
\end{align}

\end{proof}
\subsubsection{Proof of \Cref{thm:asymptoticUpsilon}}
The next theorem is a refinement of \Cref{thm:asymptoticUpsilon}.
\begin{theorem}\label{thm:asymptoticUpsilonRefi}
Under the null hypothesis and  \Cref{Cond:bounded} it holds that
\begin{align}
\Upsilon_n(\CMa) \overset{\mathcal D}{\to} \Upsilon(\CMa):= \sum_{i=1}^{\infty} \lambda_i \xi_i^2\label{eqn:asytrue}
\end{align}
as $n$ tends to infinity, where $\xi_1,\xi_2,\ldots$ are independent and identically distributed $\mathcal N(0,1)$ random variables, and $\lambda_1,\lambda_2,\ldots$ are the eigenvalues of an integral operator $T_J:\mathcal L^2(\mu) \to \mathcal L^2(\mu)$, where $\mu$ denotes the probability measure on $[0,\infty)\times \{0,1\}\times [k]$ associated with the observed triple $(T_{i},\Delta_{i},x_{i})$.
\end{theorem}

To prove \Cref{thm:asymptoticUpsilonRefi} we require the following intermediate results

\begin{proposition}\label{prop:limitApproxTestStat} Under the null hypothesis and \cref{Cond:bounded}, we have that
$$\sup_{\omega\in B_1(\mathcal{H})}\left(  \frac{1}{\sqrt{n}}\mathbf{1}^\transpose\int_0^\infty \W(t){\Qlim}(t) d\M(t)\right)^2\overset{\mathcal D}{\to} \Upsilon(\CMa)$$
where $\Upsilon(\CMa)$ is as defined in \cref{thm:asymptoticUpsilonRefi}.
\end{proposition}

\begin{proposition}\label{prop:approxTestStat}
Under the null hypothesis and \Cref{Cond:bounded}, 
\begin{align*}
\Upsilon_n(\CMa) = \sup_{\omega\in B_1(\mathcal{H})}\left(  \frac{1}{\sqrt{n}}\mathbf{1}^\transpose\int_0^\infty \W(t){\Qlim}(t) d\M(t)\right)^2 + o_p(1).
\end{align*}
\end{proposition}

\begin{proof}[\textbf{\textbf{Proof of \Cref{thm:asymptoticUpsilonRefi}}}]

Combine the \Cref{prop:approxTestStat} and  \Cref{prop:limitApproxTestStat} with Slutsky's theorem.
\end{proof}

We proceed to prove the intermediate results.

\begin{proof}[\textbf{Proof of \Cref{prop:limitApproxTestStat}} ]

Denote by $Z_n$ the random variable 
$$Z_n=\sup_{\omega\in B_1(\mathcal{H})}\left(  \frac{1}{\sqrt{n}}\mathbf{1}^\transpose\int_0^\infty \W(t){\Qlim}(t) d\M(t)\right)^2$$
then, following the same steps of \Cref{prop:closedform1}, we have
\begin{align}
Z_n  &= \frac{1}{n}\int_0^{\infty}\int_0^{\infty} ({\Qlim}(t)d\M(t))^{\transpose} \K(t,s) ({\Qlim}(s)d\M(s)))\nonumber\\
&= \frac{1}{n} \sum_{i=1}^n \sum_{j=1}^n \int_0^{\infty}\int_0^{\infty} ({\Qlim}(t)d\M^j(t))^{\transpose} \K(t,s) ({\Qlim}(s)d\M^i(s)))\nonumber\\
&= \frac{1}{n} \sum_{i=1}^n \sum_{j=1}^n {J}((T_i,\Delta_i,x_i),(T_j,\Delta_j,x_j)),\label{eqn:UpsilonApproxVstats1}
\end{align}
where $J:(\R_+\times\{0,1\}\times[k])^2\to\R$ is defined by
\begin{align*}
  {J}((T_i,\Delta_i,x_i),(T_j,\Delta_j,x_j))=\int_0^{\infty}\int_0^{\infty} ({\Qlim}(t)d\M^j(t))^{\transpose} \K(t,s) ({\Qlim}(s)d\M^i(s))).  
\end{align*}
Notice that $J$ only depends on the observed data points $D_i=(T_i,\Delta_i,X_i)$ and $D_j=(T_j,\Delta_j,X_j)$, and thus it can be deduced that $Z_n$ is a V-statistic with V-statistic kernel $J$. Since $Z_n$ is a V-statistic, we can use the standard theory of V-statistics, in particular Theorem B of \cite[Section 6.4.1]{serfling80}, to obtain the limit distribution of $Z_n$. Notice however that in order to apply the previous result, we need to check the following properties: 
\begin{itemize}
\item[i)]$\E(J(D_i,D_j)|D_j)=0$ for any $i\neq j$,
\item[ii)] $\E(|J(D_i,D_i)|)<\infty$ and,
\item[iii)]$\E(J(D_i,D_j)^2)<\infty$ for any $i\neq j$.
\end{itemize}
We defer the proof of the properties above to the end of this proof. From property i) we deduce that $Z_n$ is a degenerate V-statistic. Thus, by \cite[Section 6.4.1, Theorem B]{serfling80} together with properties ii) and iii), we deduce that
\begin{align*}
Z_n \overset{\mathcal D}{\to}\sum_{i=1}^{\infty} \lambda_i \xi_i^2,
\end{align*}
where $\xi_i$ are independent and identically distributed $\mathcal N(0,1)$ random variables, and $\lambda_i$ are the eigenvalues of the integral operator $T_J:\mathcal L_2(\mu) \to \mathcal L_2(\mu)$,
\begin{align}\label{eqn:defiTJOperator}
T_Jf(\cdot) = \int J(\cdot, u)f(u)\mu(du),
\end{align}
where $\mu$ is the probability measure on $[0,\infty)\times \{0,1\}\times [k]$ associated with the triple $D_i=(T_{i},\Delta_{i},X_{i})$.

We finish our proof by showing properties i), ii) and iii) hold true. We start by checking i). Let $i\neq j$, then
\begin{align*}
\E\left(J(D_i,D_j)|D_i\right)&=\E\left(\left.\int \int ({\Qlim}(t)d\M^j(t))^{\transpose} \K(t,s) ({\Qlim}(s)d\M^i(s)))\right|D_i\right)\\
&=\E\left(\left.\int \int \sum_{\ell,\ell'=1}^k\sum_{m,m'=1}^k{\Qlim}_{\ell m}(t)d\M^j_m(t) \K_{\ell\ell'}(t,s){\Qlim}_{\ell'm'}(s)d\M^i_{m'}(s)\right|D_i\right)\\
&=\E\left(\left.\int \left(\int \sum_{\ell,\ell'=1}^k{\Qlim}_{\ell X_j}(t) \K_{\ell\ell'}(t,s){\Qlim}_{\ell'X_i}(s)dM_i(s))\right)dM_j(t)\right|D_i\right)\\
&=\E\left(\left.\int_0^{T_j} \underbrace{\left(\int_0^{T_i} (\Qlim(t)\K(t,s)\Qlim(s))_{X_jX_i} dM_i(s))\right)}_{g_{ij}(t)}dM_j(t)\right|D_i\right)\\
&=\E\left(\left.\int_0^{T_j} g_{ij}(t)dM_j(t)\right|D_i\right),
\end{align*}
where the second equality follows from matrix multiplication, and the third equality holds since $d\M^i(t)=dM_i(t)\mathbf{e}_i$ where $\indv_{i}=(\ind_{\{X_i=1\}},\ind_{\{X_i=2\}},\ldots,\ind_{\{X_i=k\}})$. Finally, notice that given $D_i$, $g_{i,j}(t)$ is a predictable process, and that under the null hypothesis $M_j(t)$ is a  zero-mean martingale. Thus, by the optional stopping time theorem, it holds $\E\left(\left.\int_0^{T_j} g_{ij}(t)d\M_j(t)\right|D_i\right)=\E\left(\left.\int_0^{0} g_{ij}(t)d\M_j(t)\right|D_i\right)=0$, from which we deduce the desired result.

To check properties ii) and iii), recall that first, $\Qlim$ is a projection matrix, and thus \Cref{prop:projection} tells us that $|\Qlim_{ij}(t)|\leq 1$ for any $t\geq 0$ and $i,j\in[k]$; second, by \Cref{Cond:bounded}, the kernel $K$ is bounded, and third, $M_i(t)$ is a squared-integrable martingale for each $i\in [n]$. By using these observations claims ii) and iii) can be easily checked from straightforward computations.
\end{proof}

\begin{proof}[\textbf{Proof of \Cref{prop:approxTestStat}}]
By \cref{eqn:statSup} and \cref{eqn:ourLogrank}, it holds
\begin{align*}
\Upsilon_n(\CMa) = \sup_{\omega\in B_1(\mathcal{H})}\left( \frac{1}{\sqrt{n}}\mathbf{1}^\transpose\int_0^\infty \W(t)\Q(t) d\N(t)\right)^2.
\end{align*}
We start by claiming that we can replace $d\N(t)$ by $d\M(t)$ in the previous equation. Indeed, from \cref{{eqn:martiMulti}} we have $d\N(t)=d\M(t) +n\X(t)d\B(t)$, and thus
\begin{align*}
\int_0^\infty \W(t)\Q(t) d\N(t)&=\int_0^\infty \W(t)\Q(t)(d\M(t) +n\X(t)d\B(t)).
\end{align*}
Now, note that
\begin{align*}
\Q(t)\X(t)&=\IndRank(t)\left(\I_k-\X(t)(\X(t)^\transpose\X(t))^{-1}\X(t)^\transpose\right)\X(t)=\IndRank(t)\left(\X(t)-\X(t)\right)=0,
\end{align*}
which proves our claim.

We continue by  proving that $\Q(t)$ can be replaced by $\Qlim(t)$ in $\int \W(t)\Q(t) d\M(t)$. Define $\D = \Q-\Qlim$, and write
\begin{align*}
\Upsilon_n(\CMa) = \sup_{\omega\in B_1(\mathcal{H})}\left(  \frac{1}{\sqrt{n}}\mathbf{1}^\transpose\int_0^\infty \W(t)(\D(t)+\Qlim(t)) d\M(t)\right)^2.
\end{align*}
Notice then that, by the triangle inequality, it is enough to show that
\begin{align}\label{eqn:asymcon1}
 \sup_{\omega\in B_1(\mathcal{H})}\left( \frac{1}{\sqrt{n}}\mathbf{1}^\transpose\int_0^\infty \W(t)\D(t) d\M(t)\right)^2 = o_p(1),
\end{align}
and that
\begin{align}\label{eqn:asymcon2}
\sup_{\omega\in B_1(\mathcal{H})}\left(\frac{1}{\sqrt{n}}\mathbf{1}^\transpose\int_0^\infty \W(t)\Qlim(t) d\M(t)\right)^2 = O_p(1).
\end{align}

Note that \cref{eqn:asymcon2} holds immediately as \Cref{prop:limitApproxTestStat} states the above quantity converges in distribution (and hence it is $O_p(1)$). We proceed to prove \cref{eqn:asymcon1}. Denote by $R$ the term in the left-hand-side of \cref{eqn:asymcon1}, and note that vector and matrix multiplication yields
\begin{align*}
R\leq \left(\sum_{i=1}^k \sum_{j=1}^k \sup_{\omega\in B_1(\mathcal{H})}\frac{1}{\sqrt{n}}\int_0^{\infty} w(t,i)\D_{i,j}(t)d\M_j(t)\right)^2.
\end{align*}
Hence, to prove \cref{eqn:asymcon1}, it is enough to show that $\sup_{\omega\in B_1(\mathcal{H})}\int_0^{\infty} w(t,i)\D_{i,j}(t)d\M_j(t)= o_p(1)$ for any $i,j\in [k]$. Consider some fixed $i,j\in[k]$, then observe that
\begin{align*}
\left(\sup_{\omega\in B_1(\mathcal{H})}\frac{1}{\sqrt{n}}\int_0^{\infty} w(t,i)\D_{i,j}(t)d\M_j(t)\right)^2 = \frac{1}{n}\int_0^{\infty}\int_0^{\infty} \D_{i,j}(t)K((t,i), (s,i))  \D_{i,j}(t)d\M_j(t) d\M_j(s),
\end{align*}
follows from the fact that we are taking supremum over the unit ball of a RKHS. To avoid a long expression, we write $\kappa(t,s)$ instead of $\D_{i,j}(t)K((t,i),(s,i))\D_{i,j}(s)$. Then, we proceed to prove that $\int \int \kappa(t,s)d\M_j(t)d\M_j(s) = o_p(1)$. The previous result can be verified using Theorem 17 of \cite{fernandez2020reproducing} which tells us that it is enough to verify that
\begin{align*}
\frac{1}{n}\int_0^{\infty} \kappa(t,t)d \langle \M_j \rangle (t)  = o_p(1),
\end{align*}
when $\kappa$ is positive definite, which it is the case as $K((t,j),(s,j))$ is positive definite for every fixed $j$, and, in particular, $k(t,t)\geq 0$.

Recall that $d\langle \M_j\rangle (t) = \Y_{j}(t) d\bLambda_j(t)$. Then, the following equalities hold
\begin{align}
\frac{1}{n}\int_0^{\infty}  \kappa(t,t)d \langle \M_j \rangle (t) = \frac{1}{n}\int_0^{\infty}  \kappa(t,t) \Y_{j}(t)d \bLambda_j(t)&= O_p\left(\int_0^{\infty}  \kappa(t,t) p_j(1-\Hb_j(t))d \bLambda_j(t) \right)\nonumber\\
&= O_p\left(\int_0^{\infty}  \kappa(t,t) d \F_j(t) \right)\label{eqn:random192jf9},
\end{align}
where the second equality holds since  by  \Cref{prop:ProbBounds}, it holds that $\sup_{t\leq\tau_n}Y_j(t)/(n p_j(1-H_j(t)))=O_p(1)$, where recall that $\tau_n=\max\{T_1,\ldots,T_n\}$. Note that the result holds since the $t^\star=\sup\{t\geq0:\IndRank(t)=1\}\leq \tau_n$, and recall that $\kappa(t,t)=\IndRank(t)\kappa(t,t)$.

To prove the later expression is $o_p(1)$, we use Lebesgue's dominated convergence theorem in sets of large probability. Then, we need to verify that i) $\kappa(t,t)$ tends to $0$ for each $t$ when the number of data points tends to infinity, and that ii) it exists a non-negative deterministic function $g(t)$ such that $\sup_{t\geq 0} \kappa(t,t)/g(t) = O_p(1)$, and that $ \int g(t) d\F_j(t) \leq \infty$. 

To verify the conditions above we use that under \Cref{Cond:bounded} the kernel $K$ is bounded by some constant $C\geq0$. Thus
\begin{align*}
|\kappa(t,s)| = |\D_{i,j}(t)K((t,i), (s,i))  \D_{i,j}(t)|\leq C \D_{i,j}(t)^2 \to 0,
\end{align*}
since $\D_{i,j}(t)\to 0$ for each $t$, and indeed, uniformly on every interval $[0,T]$, with $T\geq0$, proving i). Notice also that $|\D_{ij}(t)|=|\Q_{ij}(t)-\Qlim_{ij}(t)|\leq 2$  by \Cref{prop:projection} since $\Q(t)$ and $\Qlim(t)$ are orthogonal projection matrices. Then, we may choose $g(t)=4C$, and ii) follows.
\end{proof}

\subsubsection{Proof of \Cref{thm:alternative}}
The proof of \Cref{thm:alternative} is a direct consequence of \Cref{Lemma1} and \Cref{Lemma3} stated below. To introduce these lemmas, we first define the measure $\mu$ on $(0,\infty)\times [k]$, equipped with the standard product $\sigma$-algebra, by
\begin{align*}
\mu([s,t)\times A)&:= \sum_{i=1}^k\sum_{j=1}^k \ind_{\{i\in A\}}\int_s^t\Qlim_{ij}(x)p_j(1-\mathbf{H}_j(x))d\bLambda_j(x),
\end{align*}
for any $0\leq s\leq t\leq \tau_H$, and $A\subseteq [k]$. In matrix notation, we can re-write the measure $\mu$ as follows
\begin{align*}
\mu([s,t)\times A) = \bone_A^{\transpose}\int_{s}^t\Qlim(t) \Diag{\S}{t}d\bLambda(x),
\end{align*}
where $\S(t) = (p_1(1-H_1(t)),\ldots, p_k (1-H_k(t)))$, and where $\bone_A$ is a $k$-dimensional vector such that $(\bone_A)_i=\Ind_{\{i\in A\}}$ for any $i\in[k]$.

\begin{lemma}\label{Lemma1} Under \Cref{Cond:bounded}, we have that $n^{-1}\Upsilon_n(\CMa)\overset{\Prob}{\to} c$, where $c\geq0$ is the constant given by
\begin{align}\label{eqn:limitcValue}
  c=\sup_{\omega\in B_1(\mathcal H)}\left(\int_{(0,\tau_H)\times [k]} \omega(t,i)d\mu(t,i)\right)^2.
\end{align}

\end{lemma}

\begin{lemma}\label{Lemma3}
 If a visible alternative $H_1'$ holds, then $\mu$ is not be the zero measure.
\end{lemma}

\begin{proof}[\textbf{Proof of \Cref{thm:alternative}}]
By \Cref{Lemma1},  under \Cref{Cond:bounded},  we have that
\begin{align*}
n^{-1}\Upsilon(\CMa)&\overset{\Prob}{\to}
\sup_{\omega\in B_1(\mathcal H)}\left(\int_{(0,\tau_H)\times [k]} \omega(t,i)d\mu(t,i)\right)^2,
\end{align*}
where note that
\begin{align*}
   \int_{(0,\tau_H)\times [k]} \omega(t,i)d\mu(t,i)=\left\langle\omega, \int_{(0,\tau_H)\times [k]} K((t,i),\cdot)d\mu(t,i)\right\rangle_{\mathcal{H}},
\end{align*}
where $\int_{(0,\tau_H)\times [k]} K((t,i),\cdot)d\mu(t,i)$ is the embedding of the signed measure $\mu$ into the RKHS $\mathcal{H}$. Since we are taking supremum over the unit ball, we have
\begin{align*}
n^{-1}\Upsilon(\CMa)&\overset{\Prob}{\to} \left\|\int_{(0,\tau_H)\times [k]} K((t,i),\cdot)d\mu(t,i)\right\|^2_{\mathcal{H}}.
\end{align*}

By \Cref{Lemma3}, since under the assumptions of \Cref{thm:alternative} an alternative hypothesis $H_1'$ holds, we deduce that $\mu\neq0$.  Then, since $K$ is $c_0$-universal kernel (\Cref{Cond:bounded}), the embedding of signed measures is injective and thus $\int_{(0,\tau_H)\times [k]} K((t,i),\cdot)d\mu(t,i)\neq0$, from which we deduce that $n^{-1}\Upsilon(\CMa)\overset{\Prob}{\to} c>0$, which implies the main result, i.e., $\Upsilon(\CMa)\overset{\Prob}{\to} \infty$.
\end{proof}

We proceed to prove \Cref{Lemma1} and \Cref{Lemma3}.
\begin{proof}[\textbf{Proof of \Cref{Lemma1}}]
From \cref{eqn:ourLogrank} and \eqref{eqn:statSup} we have
\begin{align*}
n^{-1}\Upsilon_n(\CMa)&=\sup_{\omega\in B_1(\mathcal H)}\left(\frac{1}{n}\bone^{\transpose}\int_0^{\infty} \W(t)\Q(t)d\N(t)\right)^2.
\end{align*}
We first prove that we can replace $\Q(t)$ by its limit $\Qlim(t)$. Let $\D(t)=\Q(t)-\Qlim(t)$, then the result follows directly from proving that
\begin{align}
\sup_{\omega\in B_1(\mathcal H)}\left(\frac{1}{n}\bone^{\transpose}\int_0^{\infty} \W(t)\D(t)d\N(t)\right)^2=o_p(1)\label{eqn:AlteD},
\end{align}
and that
\begin{align}
\sup_{\omega\in B_1(\mathcal H)}\left(\frac{1}{n}\bone^{\transpose}\int_0^{\infty} \W(t)\Qlim(t)d\N(t)\right)^2\to c\quad\label{eqn:OpQ} a.e.
\end{align}
for some constant $c\geq 0$.
We start by proving \cref{eqn:AlteD}. Denote by $R$ the left-hand-side of \cref{eqn:AlteD}.  Then, vector and matrix multiplication yields
\begin{align}
R&:=\left(\sum_{i=1}^k \sum_{j=1}^k \sup_{\omega \in B_1(\mathcal H)} \frac{1}{n} \int_0^{\infty} \omega(t,i)\D_{i,j}(t)d\N_j(t)\right)^2\leq \left(\sum_{i=1}^k \sum_{j=1}^k \frac{1}{n} \int_0^{\infty} C|\D_{i,j}(t)|d\N_j(t)\right)^2,\label{eqn:R}
\end{align}
where the last inequality holds for some constant $C>0$ by \Cref{Cond:bounded}, as 
\begin{align*}
\sup_{\omega \in B_1(\mathcal H)}\omega(t,i)=\sup_{\omega \in B_1(\mathcal H)}\langle\omega(\cdot),K(\cdot,(t,i))\rangle_{\mathcal{H}}\leq \sup_{\omega \in B_1(\mathcal H)}\|\omega\|\|K(\cdot,(t,i))\|\leq C.  
\end{align*}
Let $T>0$, then notice that the integral in the right-hand-side term of \cref{eqn:R} can be decomposed into two parts leading to
\begin{align*}
    R
    &\leq\left(\sum_{i=1}^k \sum_{j=1}^kC\left(\frac{1}{n}\int_0^{T} |\D_{i,j}(t)|d\N_j(t)+\frac{1}{n}\int_T^{\infty}|\D_{i,j}(t)|d\N_j(t)\right)\right)^2\\
    &\leq \left(\sum_{i=1}^k\sum_{j=1}^kC\left(\frac{1}{n}\int_0^{T} \epsilon d\N_j(t)+\frac{1}{n}\int_T^{\infty}2d\N_j(t)\right)\right)^2\leq \left(Ck^2\epsilon+2Ck\sum_{j=1}^k\frac{1}{n}\int_T^{\infty}d\N_j(t)\right)^2
\end{align*}
where the second inequality holds since for any $i,j\in[k]$ i) $\D_{i,j}(t)\to 0$ uniformly on every finite interval $[0,T]$, and ii) $|\D_{ij}(t)|=|\Q_{ij}(t)-\Qlim_{ij}(t)|\leq 2$  by \Cref{prop:projection} since $\Q(t)$ and $\Qlim(t)$ are orthogonal projection matrices.

By the law of large numbers, for any $j\in[k]$,
\begin{align*}
    \frac{1}{n}\int_T^{\infty}d\N_j(t)\overset{a.s.}{\to}p_j\int_T^\infty (1-G_j(t))dF_j(t),
\end{align*}
where the previous limit is upper bounded by $1-F_j(T)$. Then, by choosing $T>0$ and $n$ large enough,  we get that almost surely
\begin{align*}
    R&\leq \left(Ck^2\epsilon+2Ck\sum_{j=1}^k\frac{1}{n}\int_T^{\infty}d\N_j(t)\right)^2\leq \left(Ck^2\epsilon+2Ck^2\epsilon\right)^2=9C^2k^4\epsilon^2,
\end{align*}
from which we conclude that $R\overset{a.s.}{\to}0$ since $\epsilon>0$ can be chosen arbitrarily small.

We continue by proving \cref{eqn:OpQ}. Note that $\frac{1}{n}\bone^{\transpose}\int_0^{\infty} \W(t) \Qlim(t)d\N(t)$ is no more than the sum of independent and identically distributed random variables. Now, the supremum can be represented as a $V$-statistic (the analogue of \Cref{prop:closedform1}, with $\Qlim$ instead of $ \Q$), with bounded $V$-statistic kernel. Hence the left-hand side of \cref{eqn:OpQ} converges, and indeed, the limit $c$ is given by
\begin{align}
c=\sup_{\omega\in B_1(\mathcal H)}\left(\sum_{i=1}^k\sum_{j=1}^k\int_0^{\tau_H} \omega(t,i) \Qlim_{ij}(t)p_j(1-\mathbf{H}_j(t))d\bLambda_j(t)\right)^2,\label{eqn:ceq}
\end{align}
where $p_j$ is the probability that a data point has group $j$ and recall that $\tau_H=\sup\{t\geq 0:\Xlim(t)=\Diag{\boldsymbol\eta}{t}\V\text{ has full rank}\}$.

By repeating the argument of \Cref{prop:closedform1} (in reverse) we get \cref{eqn:limitcValue}.
\end{proof}

\begin{proof}[\textbf{Proof of \Cref{Lemma3}}]
By definition, if a visible alternative holds, then there exists a $t^\star\leq \tau_H$ such that 
\begin{align}
\CMa\bLambda(t^\star)\neq 0.\label{eqn:Valte}    
\end{align}
 We will prove that if $\mu$ is not the zero measure, then $d\bLambda(t)=\V d\B(t)$ for all $t\leq \tau_H$, which contradicts \cref{eqn:Valte} as 
 \begin{align*}
     \CMa\bLambda(t^\star)=\int_0^{t^\star}\CMa\V d\B(t)=0,
 \end{align*}
 since $\V$ contains in its columns a basis of the null space of $\CMa$. Thus, $\mu$ cannot be the zero measure.

Suppose that $\mu=0$, let's verify that $d\bLambda(t)=\V d\B(t)$ for all $t\leq \tau_H$. To prove this,  write $\bLambda(t)$ as
\begin{align*}
d\bLambda(t) = \V d\B(t)+\V^{\perp} d\B^{\perp}(t), \quad t\geq 0
\end{align*}
where $\V^\perp$ is a base of the orthogonal complement of the columns of $\V$, and $\B^{\perp}(t)$ is a real vector-valued function. We shall prove now that $\V^{\perp} d\B^{\perp}(t)=0$ for all $t\leq\tau_H$.  

Observe by definition of the measure $\mu$ that $\mu=0$ implies that
\begin{align*}
 \Qlim(t)\Diag{\S}{t}d\bLambda(t)=0,\quad\text{for all }t\leq\tau_H,
\end{align*}
and by replacing $d\bLambda(t)$ above, we obtain that for all $t\leq \tau_H$
\begin{align}
\Qlim(t)\Diag{\S}{t}\left(\V d\B(t)+\V^{\perp} d\B^{\perp}(t)\right)&=0\nonumber\\
\Longrightarrow\Qlim(t)\Diag{\S}{t}\V^{\perp} d\B^{\perp}(t)&=0\label{eqn:projector},
\end{align}
where the second line follows since $\Qlim(t)$ is the orthogonal projector onto the null space of $\Xlim(t)$, and thus $\Qlim(t)\Diag{\S}{t}\V=\Qlim(t)\Xlim(t)=0$. By \cref{eqn:projector} we deduce that $\Diag{\S}{t}\V^{\perp}d\B^{\perp}(t)$ belongs to the range of $\Xlim(t)$, that is, for every $t\leq \tau_H$ it exists for some $u\in\R^d$ such that
\begin{align*}
    \Xlim(t)u&=\Diag{\S}{t}\V^{\perp} d\B^{\perp}(t)
\end{align*}
and by the definition of $\Xlim(t)$, we get $$ \Diag{\S}{t}\V u=\Diag{\S}{t}\V^{\perp} d\B^{\perp}(t)$$
implying that $\V u=\V^{\perp} d\B^{\perp}(t)$, which is a contradiction since $\V^\perp$ is a base of the orthogonal complement of the columns of $\V$.
\end{proof}

\subsection{Proofs of \Cref{sec:wild}}\label{app:proofMultiple}
In this section we prove  \Cref{thm:wildAymp}, the proof of \Cref{thm:Algo1Works} is omited as it is a special case of  \Cref{thm:multiMatrixAlgorithmWorks} which will be proven later.

Our proof of \Cref{thm:wildAymp} is a direct consequence of the following two results.

\begin{lemma}\label{thm:Xi}
Consider 
$$\Xi_n = \sup_{\omega\in B_1(\mathcal{H})}\left(\frac{1}{\sqrt{n}}\mathbf{1}^\transpose\int_0^\infty \W(t){\Qlim}(t) d\widetilde{\N}(t)\right)^2.$$
Then, under the null hypothesis and \Cref{Cond:bounded}, we have that for any $t\geq 0$
\begin{align}
    \Prob(\Xi_n\leq t|D_1,\ldots, D_n) \to \Prob(\Upsilon(\CMa)\leq t)
\end{align}
for almost every sequence of data points $D_1,D_2,\ldots$, where recall $\Upsilon(\CMa)$ is the limit of $\Upsilon_n(\CMa)$ under the null hypothesis.
\end{lemma}

\begin{lemma}\label{prop:WBapprox}Under the null hypothesis and \Cref{Cond:bounded}, it holds
\begin{align*}
\widetilde{\bUpsilon}_n(\CMa) = \sup_{\omega\in B_1(\mathcal{H})}\left( \frac{1}{\sqrt{n}}\mathbf{1}^\transpose\int_0^\infty \W(t){\Qlim}(t) d\widetilde{\N}(t)\right)^2+o_p(1),
\end{align*}
where $\Qlim(t)$ is the limit of ${\Q}(t)$ when $n$ tends to infinity.
\end{lemma}

\begin{proof}[\textbf{Proof of \Cref{thm:wildAymp}}]
The proof of \Cref{thm:wildAymp} follows by noting that  $\widetilde\Upsilon_n=\Xi_n+o_p(1)$  by \cref{prop:WBapprox}. Then combining \Cref{thm:Xi} with Slutsky's theorem yields the result.
\end{proof}

We proceed to prove \Cref{prop:WBapprox,thm:Xi}.
\begin{proof}[\textbf{Proof of \Cref{prop:WBapprox}}]
Let $\D(t)=\Q(t)-\Qlim(t)$, and notice that by the triangle inequality, we just need to verify that
\begin{align}
\sup_{\omega\in B_1(\mathcal{H})}\left( \frac{1}{\sqrt{n}}\mathbf{1}^\transpose\int_0^\infty \W(t){\D}(t) d\widetilde{\N}(t)\right)^2 = o_p(1),\label{eqn:WE1}
\end{align}
and 
\begin{align}
\sup_{\omega\in B_1(\mathcal{H})}\left( \frac{1}{\sqrt{n}}\mathbf{1}^\transpose\int_0^\infty \W(t){\Qlim}(t) d\widetilde{\N}(t)\right)^2 = O_p(1).\label{eqn:WE2}
\end{align}

We first notice that \cref{eqn:WE2} is a consequence of \Cref{prop:WBapprox}, so we will just prove \cref{eqn:WE1}. Denote by $R$ the left-hand-side of \cref{eqn:WE1}, and note that since $R$ is positive, it suffices to prove that its expectation is $o(1)$. 

By using an analogue of \cref{eqn:closeformwild1}, we have that
\begin{align*}
     \E\left(R\right)&=\E\left(\frac{1}{n}\int_0^\infty \int_0^\infty ({\D}(t)d\widetilde{\N}(t))^{\transpose} \K(t,s) ({\D}(s)d\widetilde{\N}(s))\right).
\end{align*}
Then, by using the definition of $d\widetilde{\N}(t)=\sum_{i=1}^nW_id\N^i(t)$ we obtain
\begin{align}
 \E(R)&=\frac{1}{n}\sum_{i=1}^n \sum_{j=1}^n\E\left( W_iW_j\int_0^\infty \int_0^\infty ({\D}(t)d{\N^i}(t))^{\transpose} \K(t,s) ({\D}(s)d{\N^j}(s)) \right)\nonumber\\
 &=\frac{1}{n}\sum_{i=1}^n \E\left( \int_0^\infty \int_0^\infty ({\D}(t)d{\N^i}(t))^{\transpose} \K(t,s) ({\D}(s)d{\N^i}(s))\right)\nonumber,
\end{align}
where the second equality holds since $W_1,\ldots,W_n$ are i.i.d. Rademacher random variables. 

Finally, we use the previous equation and the definition of  $d\N^i(t)=dN_i(t)\mathbf{e}_i$ where recall that $\mathbf{e}_i$ is defined as $\indv_{i}=(\ind_{\{x_i=1\}},\ind_{\{x_i=2\}},\ldots,\ind_{\{x_i=k\}})$ to obtain
\begin{align*}
\E\left(R\right)
=\frac{1}{n}\sum_{i=1}^n \E\left(\Delta_i \sum_{\ell=1}^k\sum_{\ell'=1}^k \D_{\ell x_i}(T_i) \K_{\ell\ell'}(T_i,T_i)\D_{\ell' x_i}(T_i) \right)&\leq \frac{C}{n}\sum_{i=1}^n \E\left( \left(\sum_{\ell=1}^k |\D_{\ell x_i}(T_i)| \right)^2  \right)\\
&=o_p(1).\end{align*}
 The inequality in the previous equation  is due to \Cref{Cond:bounded} since the kernel function is bounded by some constant $C\geq0$. Finally, we conclude the last term is $o(1)$ by an application of the Dominated convergence theorem, since each coordinate of $\D$ tends to $0$ uniformly on every interval $[0,T]$ with $T>0$, and since $\D$ is bounded (recall that $\D=\Q-\Qlim$, and both $\Q$ and $\Qlim$ are projection matrices).
\end{proof}

\begin{proof}[\textbf{Proof  of \Cref{thm:Xi}}]
Following the same argument of \Cref{prop:closedform1}, we have that
\begin{align}
\Xi_n & = \sup_{\omega\in B_1(\mathcal{H})}\left( \frac{1}{\sqrt{n}}\mathbf{1}^\transpose\int_0^\infty \W(t)\Qlim(t) d\widetilde{\N}(t)\right)^2\nonumber\\
&= \frac{1}{n} \int_0^{\infty} \int_0^{\infty} (\Qlim(t)d\widetilde{\N}(t))^{\transpose}\K(t,s)(\Qlim(s)d\widetilde{\N}(s))\nonumber\\
&=\frac{1}{n}\sum_{i=1}^n \sum_{j=1}^n W_iW_j\int_0^{\infty} \int_0^{\infty} (\Qlim(t)d\N^i(t))^{\transpose}\K(t,s)(\Qlim(s)d\N^j(s))\nonumber\\
&=\frac{1}{n}\sum_{i=1}^n \sum_{j=1}^n W_iW_jJ'((Z_i,C_i,x_i),(Z_j,C_j,x_j)),\label{eqn:Wild}
\end{align}
where the third equality holds since $d\widetilde{\N}(t)=\sum_{i=1}^nW_id\N^i(t)$, and  where the function $J':([0,\infty)\times[0,\infty)\times[k])^2\to \R$ is defined by
\begin{align}\label{eqn:defJprima}
J'((Z_i,C_i,x_i),(Z_j,C_j,x_j))&=\int_0^{\infty} \int_0^{\infty} (\Qlim(t)d\N^i(t))^{\transpose}\K(t,s)(\Qlim(s)d\N^j(s)).
\end{align}
Note that $J'$ is defined in terms of the unobserved data points $U_i=(Z_i,C_i,x_i)$ and $U_j=(Z_j,C_j,x_j)$. Nevertheless, in order to evaluate $J'(U_i,U_j)$, it is enough to have access the observed data $D_i=(T_i,\Delta_i,x_i)$ and $D_j=(T_j,\Delta_j,x_j)$. The previous statement can be easily verified from the following computations: 
\begin{align}
J'((Z_i,C_i,x_i),(Z_j,C_j,x_j))&=\int_0^{\infty} \int_0^{\infty} (\Qlim(t)d\N^i(t))^{\transpose}\K(t,s)(\Qlim(s)d\N^j(s))\nonumber\\
&=\int_0^{\infty} \int_0^{\infty} \sum_{\ell,\ell'=1}^k(\Qlim(t)d\N^i(t))^{\transpose}_{\ell}\K_{\ell\ell'}(t,s)(\Qlim(s)d\N^j(s))_{\ell'}\nonumber\\
&=\int_0^{\infty} \int_0^{\infty} \sum_{\ell,\ell'=1}^k\Qlim_{\ell,x_{i}}(t)\K_{\ell\ell'}(t,s)\Qlim_{\ell'x_j}(s)dN_i(t)dN_j(s)\nonumber\\
&=\ind_{\{Z_i\leq C_i\}}\ind_{\{Z_j\leq C_j\}}\sum_{\ell,\ell'=1}^k\Qlim_{\ell,x_{i}}(Z_i)\K_{\ell\ell'}(Z_i,Z_j)\Qlim_{\ell'x_j}(Z_j).\label{eqn:kernelJ'}
\end{align}
Observe that in order to evaluate \cref{eqn:kernelJ'}, we only need to know $\Delta_i,\Delta_j,T_i$ and $T_j$, where recall that $\Delta_i=\ind_{\{Z_i\leq C_i\}}$ and $T_i=\min\{Z_i,C_i\}$. Also, note that for any function $g$ we have $\Delta_i g(Z_i)=\Delta_i g(T_i)$.

We continue by using \cite[Theorem 3.1]{Dehling1994RandomQF} to prove that, conditioned on the data $(D_i)_{i=1}^{\infty}$, \cref{eqn:Wild} converges to the same asymptotic null distribution as 
\begin{align}
\frac{1}{n}\sum_{i=1}^n \sum_{j=1}^nJ'((Z_i,C_i,x_i),(Z_j,C_j,x_j)).\label{eqn:Vstat}
\end{align}
Note, however, that in order to be able to apply \cite[Theorem 3.1]{Dehling1994RandomQF}, we need to verify that $J'$ is a degenerate $V$-statistic kernel, which follows from checking $\E(J'(U_i,U_j)|U_j)=0$. Also, we need to verify standard integrability conditions such as
\begin{align}
\E(|J'(U_i,U_i)|)<\infty\quad\text{and}\quad\E(J'(U_i,U_j)^2)<\infty. \text{ for any $i\neq j$}.\label{eqn:J'Integrability.}
\end{align}

Assuming such properties, by standard convergence results for degenerate $V$-statistics (in particular \cite[Section 6.4.1, Theorem B]{serfling80}) imply that
\begin{align}
\Xi_n \overset{\mathcal D}{\to}\sum_{i\geq 1} \lambda'_i\xi_i^2,\label{eqn:asywild}
\end{align}
where $\xi_1,\xi_2,\xi_3,\ldots$ are i.i.d. standard normal random variables, and $\lambda_i'$ are the eigenvalues associated to the integral operator $T_{J'}\in\mathcal{L}_2(\nu)$ given by 
\begin{align*}
T_{J'}f(\cdot)&=\int J'(\cdot,u)f(u)\nu(du),
\end{align*}
where $\nu$ here denotes the measure induced by $(Z,C,x)$ under the null.

We proceed to prove that $J'$ is a degenerate V-statistic kernel. For that it suffices to check $\E(J'(U_i,U_j)|U_j)=0$. Observe that
\begin{align*}
\E(J'(U_i,U_j)|U_i)&=\E\left(\left.\int_0^{\infty} \int_0^{\infty} (\Qlim(t)d\N^i(t))^{\transpose}\K(t,s)(\Qlim(s)d\N^j(s))\right|U_i\right)\\
&=\int_0^{\infty}(\Qlim(t)d\N^i(t))^{\transpose}\E\left(\left.\int_0^{\infty}\K(t,s)(\Qlim(s)d\N^j(s))\right|U_i\right)=0
\end{align*}
since we claim that
\begin{align}
\E\left(\int_0^{\infty} \K(t,s)(\Qlim(s)d\N^j)(s)\right)=0,\label{eqn:0claim}
\end{align}
for any fixed $t\geq 0$. 

To verify \cref{eqn:0claim} note that coordinate-wise $d\N^j(s) = (dN^j_{1}(s),dN^j_{2}(s),\ldots, dN^j_{k}(s))$ is compensated by $d\A^j(s) = (Y^j_{1}(s)d\bLambda_1(s),\ldots, Y^j_{k}(s)\bLambda_k(s))$. Also, note that $\E(d\A^j(s)) = \Diag{\S}{s}d\bLambda(s)$ since $\E(Y^j_{i}(s)) = p_i(1-H_i(s))$. The  previous observations lead to
\begin{align*}
\E\left(\int_0^{\infty} \K(t,s)\Qlim(s)d\N^j(s)\right)&= \E\left(\int_0^{\infty} \K(t,s)\Qlim(s) d\A^j(s)\right)=\int_0^{\infty} \K(t,s)\Qlim(s)\Diag{\S}{s}d\bLambda(s).
\end{align*}
Recall that, by definition, $d\bLambda(s) = \V d\B(s)$. Thus
\begin{align*}
\E\left(\int_0^{\infty} \K(t,s)\Qlim(s)d\N^j(s)\right)
=\int_0^{\infty} \K(t,s)\Qlim(s)\Diag{\S}{s}\V d\B(s)&=\int_0^{\infty} \K(t,s)\Qlim(s)\Xlim(s)d\B(s)&\\
&=0,
\end{align*}
where the last equality follows from the simple observation that $\Qlim(s)$ is a projection matrix onto the null space of $\Xlim(s)$.

To verify the integrability conditions stated in \cref{eqn:J'Integrability.}, notice that $|\K_{ij}(t)|\leq C$ by some constant $C\geq0$ by \Cref{Cond:bounded}, and that $|\Qlim_{ij}(t)|\leq 1$ for any $i,j\in[k]$ by  \Cref{prop:projection}, since $\Qlim(t)$ is an orthogonal projection matrix.

We proceed to show that $T_J$ the operator defined in ~\cref{eqn:defiTJOperator}  and $T_{J'}$ have the same set of non-zero eigenvalues, and so the limit distributions stated in \cref{eqn:asytrue} and \cref{eqn:asywild} are the same, which will conclude our proof. 

Let $(Z_1,C_1,x_1),(Z_2,C_2,x_2)\overset{i.i.d.}{\sim}\nu$, and observe that by \Cref{prop:fowardbackward}.3 and \cref{eqn:kernelJ'} it holds
\begin{align*}
J((Z_1,C_1,x_1),(Z_2,C_2,x_2))&=(B_1B_2J')((Z_1,C_1,x_1),(Z_2,C_2,x_2)),
\end{align*}
where we understand $B_i$ as the operator $B$ applied on the i-th coordinate. Then, it follows that
\begin{align}
(T_Jf)(\cdot)=\int J(\cdot,u)f(u)\nu(du)=\int (B_1B_2J')(\cdot,u)f(u)\nu(du)&=\int (B_1J')(\cdot,u)(Af)(u)\nu(du)\nonumber\\
&=B(T_{J'}(Af))(\cdot),\label{eqn:TjfBTJ'}
\end{align}
where the third equality is due to \Cref{prop:fowardbackward}.1, and the fourth equality is due to the linearity of the operator $B$, and from the definition of the operator $T_{J'}$. We conclude then that $T_J = BT_{J'}A$, and we also conclude that $T_{J'} = AT_JB$ since $AB = I$.
$$\langle T_J f, g\rangle_{L_2(\nu)} =\langle T_J f, AB g\rangle_{L_2(\nu)} $$

Recall that $T_J$ is compact since $\|J\|_{L_2(\nu\times \nu)}<\infty$ and self-adjoint since $J$ is symmetric, so the eigenvectors of $T_J$ form a basis of $L_2(\nu)$, and the same holds for $T_{J'}$. We shall prove that if $f$ is an eigenfunction with eigenvalue $\lambda\neq 0$, then $Af$ is an eigenfunction of $T_{J'}$ with the same eigenvalue. Indeed,
\begin{align*}
T_{J'}Af=ABT_{J'}Af=AT_Jf=\lambda Af,
\end{align*}
We just need to check that $Af\neq 0$ so $\lambda$ is an eigenvalue of $T_{J'}$. For that, note that 
$$\|Af\|_{L_2(\nu)}^2 = \langle f, BA f \rangle_{L_2(\nu)} = \langle f, f\rangle_{L_2(\nu)}=1,$$
so, moreover, the same computation shows that $Af$ are orthonormal. Similarly, if $g$ is a eigenfunction of $T_{J'}$ associated with an eigenvalue $\lambda$, then
$$T_{J}Bg = BT_{J'}ABg = BT_{J'}g = \lambda Bg,$$
and we can also verify that the functions $Bg\neq 0$, so $\lambda$ is also an eigenvalue of $T_{J}$. We just show that there is a bijection between the non-zero eigenvalues (and corresponding eigenvectors) of $T_J$ and $T_{J'}$, concluding the proof.
\end{proof}

\subsection{Proof of \Cref{sec:multiple}}
The main objective of this section is to prove \Cref{thm:multiMatrixAlgorithmWorks}. To achieve this, we introduce the following results which will be needed in the proof.

\begin{lemma}\label{thm:joinconvergenceDegeVStats}
Let $\bGamma:\mathcal X \times \mathcal X \to \R^b$ be defined by $\bGamma(x,y) = (\bGamma_1(x,y),\ldots, \bGamma_b(x,y))$, where $\bGamma_k:\mathcal X \times \mathcal X \to \R$  for each $k \in \{1,\ldots, b\}$. Let $X_1,\ldots,X_n$ be a collection of i.i.d random variables taking values in $\mathcal X$ such that for each $k\in \{1,\ldots, b\}$, $\E(\bGamma_k(x,X_i))=0$  for every $x\in \mathcal X$, and $\E(\bGamma_k(X_i,X_j)^2)<\infty$ for $i\neq j$, and $\E(|\bGamma_k(X_i,X_i)|)< \infty$.  

Then exists a $\R^b$-valued random variable $Z$ such that 
\begin{align*}
Z_n=\frac{1}{n}\sum_{i=1}^n \sum_{j=1}^n \bGamma(X_i,X_j) \to Z.
\end{align*} 
Moreover, given i.i.d. real random variables $W_i$ independent of $X_i$ with $\E(W_i) = 0$ and $\E(W_i^2) = 1$ then, conditioned on almost every sequence $X_1,X_2,\ldots$, we have
\begin{align*}
Z_n^W=\frac{1}{n}\sum_{i=1}^n \sum_{j=1}^n W_iW_j\bGamma(X_i,X_j)\to Z
\end{align*}
\end{lemma}

Recall that in the context of \Cref{sec:multiple}, $\bUpsilon_n$ is a b-dimensional vector defined as
$$\bUpsilon_n = (\Upsilon_n(\CMa_1),\ldots, \Upsilon_n(\CMa_b)),$$
i.e. a vector containing $b$ test-statistics. We denote by $\widetilde{\bUpsilon}_n$ its corresponding Wild Bootstrap version, and we denote by $\bUpsilon$ the potential limit in distribution of $\bUpsilon_n$ which exists by the following result.

 \begin{lemma}\label{thm:multipleV-statsConvergence} Assume \Cref{Cond:bounded} holds. Then, it exists a random vector $\bUpsilon$ taking values in $\R^b$ such that $\bUpsilon_n $ converges in distribution to $\bUpsilon$ when $n$ tends to infinity.

Moreover, for any $f:\R^b\to \R$ bounded and continuous, we have that
$$\E(f(\widetilde{\bUpsilon}_n)|D_1,\ldots, D_n) \to \E(f(\bUpsilon)),$$
as $n$ grows to infinity, for almost every sequence $D_1,D_2,\ldots$. That is, conditioned on the data, $\widetilde{\bUpsilon}_n$ also converges in distribution to $\bUpsilon$.
\end{lemma}
This section features three random vectors: $\bUpsilon_n$, $\bUpsilon$ and $\widetilde \bUpsilon_n$. From these random vectors we define 4 functions and 4 quantile functions that feature in our analysis.
\begin{enumerate}
    \item $P_n:\R^b \to \R^b$ will denote the vector of marginal cumulative distribution functions of the random vector $\bUpsilon_n = (\Upsilon_n(\CMa_1),\ldots, \Upsilon_n(\CMa_b))$, that is, for any $\boldsymbol{x}=(x_1,\ldots,x_b)\in \R^b$ we have 
    \begin{align*}
      P_n(\boldsymbol{x}) = \left( \Prob(\Upsilon_n(\CMa_1)\leq x_1),\ldots,\Prob( \Upsilon_n(\CMa_b)\leq x_b)\right)\in \R^d.
    \end{align*}
    We denote by $q_{n,i}(\alpha)$ the $\alpha$-quantile of $\Upsilon_n(\CMa_i)$, defined as 
    \begin{align*}
        q_{n,i}(\alpha) =\inf\left\{x\geq 0: \Prob\left(\Upsilon_n(\CMa_i)\leq x\right)\geq \alpha \right\}
    \end{align*}
    \item $P:\R^b \to \R^b$ will denote the vector with the cumulative distribution functions of the marginal distributions of $\bUpsilon$. We denote by $q_{i}$ the quantile function of $\Upsilon(\CMa_i)$ for $i\in[b]$.
    \item $\widetilde P_n:\R^b \to \R^b$ represents the vector containing the cumulative distribution functions of the marginal distributions associated with $\widetilde \bUpsilon_n$, which recall is the Wild Bootstrap version of the vector $\bUpsilon_n$. Note that each coordinate of $\widetilde P_n$ is a random distribution that depends on the $n$ data points $D_1,\ldots, D_n$. We denote by $\widetilde q_{n,i}$ the quantile function of $\widetilde \Upsilon(\CMa_i)$.
    \item $\widetilde P_n^M:\R^b \to \R^b$ represents the vector of cumulative distribution functions associated to the empirical measure given by $M$ independent random samples from $\widetilde \bUpsilon_n$, say, $\widetilde \bUpsilon_n^{(1)},\ldots, \widetilde \bUpsilon_n^{(M)}$. For any $i\in[b]$, we denote by $\widetilde q_{n,i}^M$ the empirical quantile associated with the component $i$.
\end{enumerate}
For all the functions described above, we use the subindex $i$ to denote its $i$-th coordinate, e.g. $P_{n,i}(t) = \Prob(\Upsilon_n(\CMa_i)\leq t)$. 

\begin{lemma}\label{lemma:PUtoPU}
 Assume \Cref{Cond:bounded}. Then $\widetilde P_{n}(\widetilde \bUpsilon_n) \overset{\mathcal D}{\to} P(\bUpsilon)$ when the number of of data points $n$ tends to infinity.  Moreover, conditioned on almost any sequence of data points, $(D_i)_{i=1}^{\infty}$, we have that $ \widetilde P_{n}( \bUpsilon_n) \overset{\mathcal D}{\to} P(\bUpsilon)$ 
\end{lemma}

Finally, recall the definition of $\beta_{\alpha}$ and $\widehat \beta_{\alpha}$ from \cref{eqn:betaalphadefi} and \cref{eqn:conditionBeta}, respectively.

\begin{lemma}\label{lemma:hatbetaconvergesup}  Assume \Cref{Cond:bounded}, and let $\varepsilon>0$, then
\begin{align*}
    \limsup_{n\to\infty}\limsup_{M\to\infty}\Prob\left(\left.\widehat\beta_\alpha>\beta_\alpha+\varepsilon\right|D_1,\ldots,D_n\right)=0.
\end{align*}
for almost all sequence of data points $(D_i)_{i=1}^\infty$.
\end{lemma}

With the results above we are ready to proceed with the proof of \Cref{thm:multiMatrixAlgorithmWorks}.

\begin{proof}[\textbf{Proof of \Cref{thm:multiMatrixAlgorithmWorks}}]

We start by proving \cref{thm7:eq1}. Consider the left-hand side of \cref{thm7:eq1}
\begin{align*}
    \limsup_{n\to \infty}\limsup_{M\to \infty}\Prob\left(\bigcup_{i=1}^b \left\{\Upsilon_n(\CMa_i) \geq \qwei{\widehat{\beta}_{\alpha}}\right\}\right),
\end{align*}
and note that by the Reverse Fatou's Lemma, it holds
\begin{align*}
    &\limsup_{M\to \infty}\Prob\left(\bigcup_{i=1}^b \left\{\Upsilon_n(\CMa_i) \geq \qwei{\widehat{\beta}_{\alpha}}\right\}\right) \\
    &\qquad\leq \E\left(\limsup_{M\to \infty}\Prob\left(\bigcup_{i=1}^b \left\{\Upsilon_n(\CMa_i) \geq \qwei{\widehat{\beta}_{\alpha}}\right\}\given D_1,\ldots, D_n\right)\right).
\end{align*}
Henceforth, for the sake's of notation, we denote by $\Prob_D$ probability conditioned on the $n$ data points $D_1,\ldots, D_n$. Consider any $\varepsilon>0$ sufficiently small, then
\begin{align*}
    &\Prob_D\left(\bigcup_{i=1}^b \left\{\Upsilon_n(\CMa_i) \geq \qwei{\widehat{\beta}_{\alpha}}\right\}\right) \\&\qquad\leq \underbrace{\Prob_D\left(\bigcup_{i=1}^b \left\{\Upsilon_n(\CMa_i) \geq \qwei{\widehat{\beta}_{\alpha}}\right\},\widehat \beta_{\alpha}\leq \beta_{\alpha} + \varepsilon\right)}_{(1)}+\underbrace{\Prob_{D}\left(\widehat \beta_{\alpha}> \beta_{\alpha} + \varepsilon\right)}_{(2)}.
\end{align*}
By \Cref{lemma:hatbetaconvergesup}, the term $(2)$ defined above tends to 0 as the number of Wild Bootstrap samples $M$ and the sample size $n$ grow to infinity. For the term (1), notice that { $\widetilde q_{n,i}^M(\alpha)$ is a non-decreasing function of $\alpha$ and thus $\widehat\beta_{\alpha}\leq \beta_{\alpha}+\varepsilon$ implies $\qwei{\widehat\beta_{\alpha}}\geq\qwei{\beta_{\alpha}-\varepsilon}$ for any $\varepsilon\geq 0$}, and thus 
\begin{align}
    (1)\leq \Prob_D\left(\bigcup_{i=1}^b \left\{\Upsilon_n(\CMa_i) \geq \qwei{\beta_{\alpha}-\varepsilon}\right\}\right).\nonumber
\end{align}
We use the previous equation to obtain
\begin{align}
    \limsup_{M\to \infty} \Prob_D\left(\bigcup_{i=1}^b \left\{\Upsilon_n(\CMa_i) \geq \qwei{ \beta_{\alpha}-\varepsilon}\right\}\right)& = \limsup_{M\to \infty} \Prob_D\left(\bigcup_{i=1}^b \left\{\widetilde P_{ni}^M(\bUpsilon_n) \geq 1-({\beta}_{\alpha}+\varepsilon)\right\}\right)\nonumber \\
    &\leq \Prob_D\left(\bigcup_{i=1}^b \left\{\widetilde P_{ni}(\bUpsilon_n) \geq  1-({\beta}_{\alpha}+\varepsilon)\right\}\right),\label{eqn:randomport12151}
\end{align}
where the first equality in the previous equation holds by definition of the quantile function $\widetilde q_{n,i}^M$, and the inequality holds by the law of large numbers by noticing that 
\begin{align*}
\widetilde P_{n,i}^M(\bUpsilon_n) = \frac{1}{M}\sum_{{\ell}=1}^M\ind\{\widetilde\Upsilon_n^{\ell}(\CMa_i) \leq \Upsilon_n(\CMa_i)\},    
\end{align*}
where $(\widetilde\bUpsilon_n^{\ell})_{\ell=1}^M$ are $M$ independent Wild Bootstrap samples, and thus, given the data, completely independent of $\bUpsilon_n$. 


Compiling the above computations yields that for any $\varepsilon>0$ sufficiently small
\begin{align*}
    \limsup_{n\to \infty}\limsup_{M\to \infty}\Prob\left(\bigcup_{i=1}^b \left\{\Upsilon_n(\CMa_i) \geq \qwei{\widehat\beta_{\alpha}}\right\}\right)\leq   \limsup_{n\to \infty} \Prob\left(\bigcup_{i=1}^b \left\{\Upsilon_n(\CMa_i) \geq \qwi{\beta_{\alpha}-\varepsilon}\right\}\right),
\end{align*}
where notice that $\widetilde q_{n,i}$ is the quantile function associated with $\widetilde{\bUpsilon}_n$.

By definition of the quantile function we have 
\begin{align*}
    \left\{\Upsilon_n(\CMa_i) \geq \qwi{{\beta}_{\alpha}-\varepsilon}\right\} = \left\{\widetilde P_{n,i}(\bUpsilon_n) \geq 1-(\beta_{\alpha}+\varepsilon)\right\}
\end{align*}
Define the set
$$\mathcal S_{\varepsilon} = \{x\in \R^d: x_i \geq 1-(\beta_{\alpha}+\varepsilon), \text{for some $i \in [b]$}\}.$$ Then, by taking $\limsup_{n\to \infty}$ to the right-hand side of \cref{eqn:randomport12151}, we have
\begin{align*}
     \limsup_{n\to \infty}\Prob_D\left(\bigcup_{i=1}^b \left\{\widetilde P_{n,i}(\bUpsilon_n) \geq  1-({\beta}_{\alpha}+\varepsilon)\right\}\right)& =\limsup_{n\to \infty}\Prob\left(\widetilde P_n(\bUpsilon_n)\in \mathcal S_{\varepsilon}\right)\leq \Prob\left(P(\bUpsilon)\in \mathcal S_{\varepsilon}\right) 
\end{align*}
The last step follows from Portmanteau's theorem ($S_{\varepsilon}$ is a closed set) and the fact that $\widetilde P_n(\bUpsilon_n)\overset{\mathcal D}{\to} P(\bUpsilon)$ as $n$ grows to infinity due to \Cref{lemma:PUtoPU}. 

Observe that that
$\Prob\left(P(\bUpsilon)\in \mathcal S_{\varepsilon}\right) =  \Prob\left(P(\bUpsilon)\in \mathcal S_{0}\right)+g(\varepsilon)$. Moreover, $g(\varepsilon) = \Prob\left(P(\bUpsilon)\in \mathcal S_{\varepsilon}\right)- \Prob\left(P(\bUpsilon)\in \mathcal S_{0}\right)$ tends to $0$ when $\varepsilon\to 0$, indeed for small $\varepsilon$, by the union bound we have
\begin{align}
    g(\varepsilon) \leq \sum_{i=1}^b \Prob(1-(\beta_{\alpha}+\varepsilon)\leq P_i(\bUpsilon) \leq 1-\beta_{\alpha})  = b\varepsilon
\end{align}
because $P_i(\bUpsilon)$ has Uniform(0,1) distribution. We conclude that for all $\varepsilon>0$ small enough we have
\begin{align*}
    \limsup_{n\to \infty}\limsup_{M\to \infty}\Prob\left(\bigcup_{i=1}^b \left\{\Upsilon_n(\CMa_i) \geq \qwei{\widehat{\beta}_{\alpha}}\right\}\right)\leq \Prob\left(P(\bUpsilon)\in \mathcal S_{0}\right) + b\varepsilon
\end{align*}
and so,
\begin{align*}
    \limsup_{n\to \infty}\limsup_{M\to \infty}\Prob\left(\bigcup_{i=1}^b \left\{\Upsilon_n(\CMa_i) \geq \qwei{\widehat{\beta}_{\alpha}}\right\}\right)\leq \Prob\left(P(\bUpsilon)\in \mathcal S_{0}\right),
\end{align*}
since $\varepsilon>0$ can be chosen arbitrarily small. We finish the proof by recalling that $\Prob\left(P(\bUpsilon)\in \mathcal S_{0}\right) = \alpha$, from which \cref{thm7:eq1} is proven.

We continue by proving the first equality of \cref{thm7:eq2}. Suppose that the local hypothesis $H_{0i}$ is false. Then, by using the analysis for a single hypothesis (see \Cref{thm:alternative}) we can deduce $\Upsilon_n(\CMa_i)\to \infty$ but $ \widetilde{\Upsilon}_n^W(\CMa_i) = O_p(1)$, the latter can be verified by using \cref{eqn:closeformwild1} to write 
\begin{align}
    \widetilde{\Upsilon}_n^W(\CMa_i) = \frac{1}{n}\int \int (\Q(t)d\widetilde{\N}(t))^{\transpose}\K(t,s)(\Q(s)d\widetilde{\N}(s)),
\end{align}
(note that $\Q$ depends implicitly on the matrix $\CMa_i$) which is non-negative and its expectation is given by 
$$\E\left(\widetilde{\Upsilon}_n^W(\CMa_i)\right) = \frac{1}{n}\sum_{i=1}^n\int \int \Q(t)d\N^i(t)\K(t,s)\Q(s)d\N^i(s)< \infty$$ 
where the inequality holds since the kernel $K$ is bounded and the projection matrix $\Q$ is bounded component-wise as well, so by Markov inequality $\widetilde{\Upsilon}_n^W(\CMa_i) = O_p(1)$.

Hence for any fixed number $M$ of independent copies of $\Upsilon_n^W(\CMa_i)$, the empirical quantile $\qwei{\beta}$ is finite for any $\beta\in (0,1)$, and so 
\begin{align*}
    \limsup_{n\to \infty} \Prob(\Upsilon_n(\CMa_i)\leq  \qwei{\widehat \beta_{\alpha}}) = 0.
\end{align*}
The previous result means that, asymptotically, we reject such hypothesis with probability tending to $1$. Since we are testing a finite number of hypothesis, the first part of \cref{thm7:eq2} follows from the union bound.

For the second equality of \cref{thm7:eq2}, consider the event $B = \bigcup_{i=1}^{b'} \{\Upsilon_n(\CMa_i)> \qwei{\widehat \beta_{\alpha}}\}$ which is the event that at least one true hypothesis is rejected. We proceed to verify the following holds: $\limsup_{n\to \infty}\limsup_{M\to \infty}\Prob(B)\leq \alpha$.

Define $\widehat \beta_{\alpha}'$ as
\begin{align}
    \widehat \beta_{\alpha}' = \sup\left\{\beta\in [0,1]: \frac{1}{M}\sum_{\ell=1}^{M} \ind\{{\widetilde{\Upsilon}}^{\ell}_n(\CMa_i)>\qwei{\beta} \text{ for at least one $i \in \{1,\ldots,b'\}$}\}\leq \alpha\right\} \label{eqn:conditionBeta}
\end{align}
Note that in the definition of $\widehat \beta_{\alpha}'$, the indicator function only considers the first $b'$ hypothesis (instead of all of them as in $\beta_{\alpha}$). For this reason we have $\widehat \beta_{\alpha}'\geq \widehat \beta_{\alpha}$

Define $B' =  \bigcup_{i=1}^{b'} \{\bUpsilon_n(C_i)>\tilde q_{1-\hat \beta_{\alpha}'}^M(C_i)\}$. Since all hypothesis from $1$ to $b'$ are true, then by \cref{thm7:eq1} of this theorem, we have
\begin{align*}
    \limsup_{n\to \infty}\limsup_{M\to \infty}\Prob(B')\leq \alpha.
\end{align*}

Finally, note that $\qwei{\widehat \beta_{\alpha}'}\leq \qwei{\widehat \beta_{\alpha}}$, and so $B\subseteq B'$, and thus
\begin{align*}
    \limsup_{n\to \infty}\limsup_{M\to \infty}\Prob(B)\leq \alpha.
\end{align*}
\end{proof}

\subsubsection{Proof of \Cref{thm:joinconvergenceDegeVStats,thm:multipleV-statsConvergence,lemma:PUtoPU,lemma:hatbetaconvergesup}}

We proceed to prove all the lemmas that featured in the proof of \Cref{thm:multiMatrixAlgorithmWorks}.

\begin{proof}[\textbf{Proof of \Cref{thm:joinconvergenceDegeVStats}}]
Let $Z = \frac{1}{n}\sum_{i=1}^n \sum_{j=1}^n \bGamma(X_i,X_j)$. By the hypotheses on $\Gamma$ we have that $\E(|\bGamma_k(X_i,X_i)|)<\infty$ and $\E(\bGamma_{k}(X_i,X_j)^2)<\infty$ and so by the standard $V$-statistic convergence theorem we have that
\begin{align*}
Z_n^k= \frac{1}{n}\sum_{i=1}^n \sum_{j=1}^n \bGamma_k(X_i,X_j)
\end{align*}
converges in distribution. Therefore, $Z_n^k$ is tight for each $k \in \{1,\ldots, d\}$, and thus $Z_n$ is a thigh sequence of random variables in $\R^d$.

By Prokhorov's theorem it exists a subsequence ${n_j}$ and a random variable $Z$ (taking values in $\R^d$) such that $Z_{n_j} \overset{\mathcal D}{\to} Z$.

We shall prove that the whole sequence $Z_n$ converges in distribution to $Z$, for that we use the Cramer-Wold theorem that states that if for every $a \in \mathcal R^d$ we have that 
$a^{\transpose} Z_n \overset{\mathcal D}{\to} a^{\transpose} Z$
if and only if $Z_n \overset{\mathcal D}{\to} Z$.

Let $a \in \R^d$, then 
\begin{align*}
a^{\transpose} Z_n = \frac{1}{n}\sum_{i=1}^n \sum_{j=1}^n \left(\sum_{k=1}^d a_k \Gamma_k(X_i,X_j)\right)
\end{align*}
Note that this is another degenerate $V$-statistic so it converges in distribution. Such limit is the same as the limit of the subsequence $a^{\transpose} Z_{n_j}$ which we know is $a^{\transpose } Z$ (by the continuous mapping theorem), concluding the result.

For the second part we have that for any $a\in \R^d$, the random variable
$$a^{\transpose} Z_n^W = \frac{1}{n}\sum_{i=1}^n \sum_{j=1}^n W_iW_j\left(\sum_{k=1}^d a_k \Gamma_k(X_i,X_j)\right)$$
is such that, given almost every sequence $X_1,\ldots,$, it converges to $a^{\transpose}Z$ (\cite[Theorem 3.1]{Dehling1994RandomQF}). Thus, by the Cramer-Wald theorem, conditioned on the sequence  $X_1,\ldots,$ we have $Z_n^W\to Z$.
\end{proof}

\begin{proof}[\textbf{Proof of \Cref{thm:multipleV-statsConvergence}}]
Recall that $\bUpsilon_n = (\Upsilon_n(\CMa_1),\ldots, \Upsilon_n(\CMa_b))$. By applying \Cref{prop:approxTestStat} and \cref{eqn:UpsilonApproxVstats1} to each $\Upsilon_n(\CMa_{\ell})$

\begin{align*}
\Upsilon_n(\CMa_j)= \frac{1}{n}\sum_{i=1}^n \sum_{j=1}^nJ_k((Z_i,C_i,x_i),(Z_j,C_j,x_j))+o_p(1)    
\end{align*}
 for each ${\ell} \in \{1,\ldots, b\}$. 

It was proved in \Cref{thm:asymptoticUpsilonRefi} that the $V$-statistics kernel $J_k$ satisfy the the conditions of \Cref{thm:joinconvergenceDegeVStats} (just set $\Gamma_k = J_k$ in the theorem), so we deduce that $\bUpsilon_n$ converges in distribution to a random vector $\bUpsilon$ that takes values in $\R^d$, proving the first statement of the lemma.

We now have to show that $\widetilde{\bUpsilon}_n$ converges in distribution to $\bUpsilon$. This is also a consequence of \Cref{thm:joinconvergenceDegeVStats}.  Indeed, by \cref{eqn:Wild} (of \Cref{thm:Xi}) and \Cref{prop:WBapprox} we have that 
\begin{align*}
\widetilde{\Upsilon}_n(\CMa_{\ell}) = n \sum_{j=1}^n W_i W_jJ_k'((T_i,\Delta_i,x_i), (T_j,\Delta_j,x_j))+o_p(1),
\end{align*}
so, conditionally on almost every sequence $D_1,D_2,\ldots$, we have that $\widetilde \bUpsilon_n$ converges to the  limit distribution $\bUpsilon$ (recall that $J'$ and $J$ have the same eigenvalues as shown in the proof of \Cref{thm:Xi}).
\end{proof}

\begin{proof}[\textbf{Proof of \Cref{lemma:PUtoPU}}]
We just prove the first limit as the second one is done exactly the same way. By writing $\widetilde P_n(\bUpsilon_n) =  P( \bUpsilon_n) + \left( \widetilde P_n( \bUpsilon_n)-P( \bUpsilon_n)\right)$ we just need to prove that $P( \bUpsilon_n)\overset{\mathcal D}{\to} P(\bUpsilon)$ and 
that $\left\|\left( \widetilde P_n( \bUpsilon)-P( \bUpsilon_n)\right)\right\|_1 \to 0$ almost surely as $n$ grows to infinity, and then the result follows from Slutsky's theorem.

The first results follows from the fact that each coordinate of $\bUpsilon$, say $\Upsilon(\CMa_i)$, has continuous distribution, thus $x\to P(x)$ is a continuous transformation in $\R^d$. Then since $ \bUpsilon_n \overset{\mathcal D}{\to} \bUpsilon$, the continuous mapping theorem yields  $P( \bUpsilon_n)\overset{\mathcal D}{\to} P(\bUpsilon)$. 

To verify that $\left\| \widetilde P_n( \bUpsilon_n)-P( \bUpsilon_n)\right\|_1 \to 0$ observe that
\begin{align}
  \left\|\left( \widetilde P_n( \bUpsilon_n)-P( \bUpsilon_n)\right)\right\|_1 \leq \sum_{i=1}^b \sup_{t\in \R} |\widetilde P_{n,i}(t)-P_i(t)|,
\end{align}
and since $\widetilde \bUpsilon_n \overset{\mathcal D}{\to} \bUpsilon$ (and in particular each coordinate converges in distribution), we have that $P_{n,i}(t)\to P_i(t)$ uniformly on $t\in \R$.

\end{proof}
\begin{proof}[\textbf{Proof of \Cref{lemma:hatbetaconvergesup} }]
Denote by $\Prob_D(\cdot)=\Prob(\cdot|D_1,\ldots,D_n)$. Then by the Reverse Fatou's lemma we have that
\begin{align}
    \limsup_{n\to \infty}\limsup_{M\to \infty} \Prob\left(\widehat\beta_\alpha>\beta_\alpha+\varepsilon\right) \leq  \limsup_{n\to \infty}\E\left(\limsup_{M\to \infty} \Prob_D\left(\widehat\beta_\alpha>\beta_\alpha+\varepsilon \right) \right)
\end{align}

By the definition of $\widehat{\beta}_\alpha$ given in \cref{eqn:conditionBeta1}, it holds
\begin{align*}
\left\{\widehat\beta_\alpha>\beta_\alpha+\varepsilon\right\}=\left\{\frac{1}{M}\sum_{\ell=1}^M\ind\left\{\bigcup_{i=1}^b\left\{\widetilde{\Upsilon}^{\ell}_n(\CMa_i)\geq \qwei{{\beta}_{\alpha}-\varepsilon}\right\}\right\}\leq \alpha\right\}.
\end{align*}

Then
\begin{align*}
    \limsup_{M\to\infty}\Prob_D\left(\widehat\beta_\alpha>\beta_\alpha+\varepsilon\right)&=    \limsup_{M\to\infty}\Prob_D\left(\frac{1}{M}\sum_{\ell=1}^M\ind\left\{\bigcup_{i=1}^b\left\{\widetilde{\Upsilon}^{\ell}_n(\CMa_i)\geq \qwei{{\beta}_{\alpha}-\varepsilon}\right\}\right\}\leq \alpha\right)\\
    &\leq\limsup_{M\to\infty}\Prob_D\left(\frac{1}{M}\sum_{\ell=1}^M\ind\left\{\bigcup_{i=1}^b\left\{\widetilde P_{n,i}^M(\widetilde{\bUpsilon}^{\ell}_n)\geq 1-{\beta}_{\alpha}-\varepsilon\right\}\right\}\leq \alpha\right)
\end{align*}
By the definition of the quantile function, observe that
\begin{align*}
   \left\{\widetilde P_{n,i}^M(\widetilde{\bUpsilon}^{\ell}_n)\geq 1-{\beta}_{\alpha}-\varepsilon\right\} &=\left\{\widetilde P_{n,i}^M(\widetilde{\bUpsilon}^{\ell}_n)-\widetilde P_{n,i}(\widetilde{\bUpsilon}^{\ell}_n)+\widetilde P_{n,i}(\widetilde{\bUpsilon}^{\ell}_n)\geq 1-{\beta}_{\alpha}-\varepsilon\right\}\\
   &\subseteq \left\{\left|\widetilde P_{n,i}^M(\widetilde{\bUpsilon}^{\ell}_n)-\widetilde P_{n,i}(\widetilde{\bUpsilon}^{\ell}_n)\right|+\widetilde P_{n,i}(\widetilde{\bUpsilon}^{\ell}_n)\geq 1-{\beta}_{\alpha}-\varepsilon\right\}.
\end{align*}
Now, consider the event
\begin{align*}
    A=\bigcap_{i=1}^b\left\{\left|\widetilde P_{n,i}^M(\widetilde{\bUpsilon}^{\ell}_n)-\widetilde P_{n,i}(\widetilde{\bUpsilon}^{\ell}_n)\right|\leq  \varepsilon\right\},
\end{align*}
and observe that
\begin{align}
&\Prob_D\left(\frac{1}{M}\sum_{\ell=1}^M\ind\left\{\bigcup_{i=1}^b\left\{\widetilde P_{n,i}^M(\widetilde{\bUpsilon}^{\ell}_n)\geq 1-{\beta}_{\alpha}-\varepsilon\right\}\right\}\leq \alpha\right)\nonumber\\
&\leq\Prob_D\left(\frac{1}{M}\sum_{\ell=1}^M\ind\left\{\bigcup_{i=1}^b\left\{\widetilde P_{n,i}^M(\widetilde{\bUpsilon}^{\ell}_n)\geq 1-{\beta}_{\alpha}-\varepsilon\right\}\right\}\leq \alpha,A\right)+\Prob_D\left(A^c\right)\nonumber\\
&\leq\Prob_D\left(\frac{1}{M}\sum_{\ell=1}^M\ind\left\{\bigcup_{i=1}^b\left\{\left|\widetilde P_{n,i}^M(\widetilde{\bUpsilon}^{\ell}_n)-\widetilde P_{n,i}(\widetilde{\bUpsilon}^{\ell}_n)\right|+\widetilde P_{n,i}(\widetilde{\bUpsilon}^{\ell}_n)\geq 1-{\beta}_{\alpha}-\varepsilon\right\}\right\}\leq \alpha,A\right)+\Prob_D\left(A^c\right)\nonumber\\
&\leq\underbrace{\Prob_D\left(\frac{1}{M}\sum_{\ell=1}^M\ind\left\{\bigcup_{i=1}^b\left\{\widetilde P_{n,i}(\widetilde{\bUpsilon}^{\ell}_n)\geq 1-{\beta}_{\alpha}-2\varepsilon\right\}\right\}\leq \alpha\right)}_{(1)}+\underbrace{\Prob_D\left(A^c\right)}_{(2)}.\label{eqn:random9v8j421}
\end{align}

Note that for $(2)$ we have
\begin{align*}
    \Prob_D(A^c)&=\Prob_D\left(\bigcup_{i=1}^b\left\{\left|\widetilde P_{n,i}^M(\widetilde{\bUpsilon}^{\ell}_n)-\widetilde P_{n,i}(\widetilde{\bUpsilon}^{\ell}_n)\right|>  \varepsilon\right\}\right)\leq \sum_{i=1}^b\Prob_D\left(\left|\widetilde P_{n,i}^M(\widetilde{\bUpsilon}^{\ell}_n)-\widetilde P_{n,i}(\widetilde{\bUpsilon}^{\ell}_n)\right|>\varepsilon\right)\\
    \text{(by Markov's inequality)}&\leq \frac{1}{\varepsilon}\sum_{i=1}^b\E_D\left(\left|\widetilde P_{n,i}^M(\widetilde{\bUpsilon}^{\ell}_n)-\widetilde P_{n,i}(\widetilde{\bUpsilon}^{\ell}_n)\right|\right)\leq \frac{1}{\varepsilon}\sum_{i=1}^b\E_D\left(\sup_{x\in\R}\left|\widetilde P_{n,i}^M(x)-\widetilde P_{n,i}(x)\right|\right).
\end{align*}
Then, by the Reverse Fatou's lemma, and the Glivenko–Cantelli theorem, we have
\begin{align*}
\limsup_{M\to\infty}\Prob_D(A^c)&\leq \frac{1}{\varepsilon}\sum_{i=1}^b\limsup_{M\to\infty}\E_D\left(\sup_{x\in\R}\left|\widetilde P_{n,i}^M(x)-\widetilde P_{n,i}(x)\right|\right)\\
&\leq \frac{1}{\varepsilon}\sum_{i=1}^b\E_D\left(\limsup_{M\to\infty}\sup_{x\in\R}\left|\widetilde P_{n,i}^M(x)-\widetilde P_{n,i}(x)\right|\right)\to0,
\end{align*}

We continue with the term $(1)$. Start by defining the set
\begin{align*}
    \mathcal{S}_{\delta}=\{\mathbf{x}\in\R^b:x_i\geq 1-\beta_\alpha-\delta,\text{ for some }i\in[b]\}. 
\end{align*}
Then we have 
\begin{align}
\limsup_{M\to\infty} \ (1)&=\limsup_{M\to\infty}\Prob_D\left(\frac{1}{M}\sum_{\ell=1}^M\ind\left\{\bigcup_{i=1}^b\left\{\widetilde P_{n,i}(\widetilde{\bUpsilon}^{\ell}_n)\geq 1-{\beta}_{\alpha}-2\varepsilon\right\}\right\}\leq \alpha\right)\nonumber\\
&=\limsup_{M\to\infty}\Prob_D\left(\frac{1}{M}\sum_{\ell=1}^M \ind\left\{\widetilde P_n(\widetilde\bUpsilon_n^{\ell})\in \mathcal{S}_{2\varepsilon}\right\}\leq \alpha\right)\leq \Prob_D\left(\Prob_{D}\left(\widetilde P_n(\widetilde \bUpsilon_n) \in \mathcal S_{2\varepsilon}\right)\leq \alpha\right)\nonumber
\end{align}
The last step holds by Portmanteau theorem since the set $\{x\in \R: x\leq \alpha\}$ is closed (and clearly the empirical average converges by the law of large numbers). To finish our proof we just need to prove that
\begin{align}\label{eqn:randomgoal120}
    \limsup_{n\to \infty}  \Prob\left(\Prob_{D}\left(\widetilde P_n(\widetilde \bUpsilon_n) \in \mathcal S_{2\varepsilon}\right)\leq \alpha\right) = 0.
\end{align}
The previous equation holds immediately if we verify that with probability $1$ it exists $N$ (depending on $(D_i)_{i=1}^{\infty}$) such that for all $n\geq N$, we have
$\Prob_{D}\left(\widetilde P_n(\widetilde \bUpsilon_n) \in \mathcal S_{2\varepsilon}\right)>\alpha$. For such a task recall that \Cref{lemma:PUtoPU} states that for almost every sequence of data points $(D_i)_{i=1}^{\infty}$, $\widetilde P_n(\widetilde\bUpsilon_n)$ converges in distribution to $P(\bUpsilon)$,  that is,
\begin{align*}
    \limsup_{n\to\infty}\Prob_D\left(\widetilde P_n(\widetilde\bUpsilon_n)\in \mathcal{S}_{2\varepsilon}\right)&=\Prob\left(P(\bUpsilon)\in \mathcal{S}_{2\varepsilon}\right)= \Prob\left(P(\bUpsilon)\in \mathcal{S}_{0}\right)+ \Prob\left(P(\bUpsilon)\in \mathcal{S}_{2\varepsilon}\setminus \mathcal S_0\right)
\end{align*}
the last equality holds because $\mathcal S_0 \subseteq \mathcal S_{\delta}$ for any $\delta >0$ by the definition of $\mathcal S_{\delta}$.

Now, note that $\Prob\left(P(\bUpsilon)\in \mathcal{S}_{0}\right) = \alpha$ by definition of $\beta_{\alpha}$. Moreover, note that
\begin{align*}
  \mathcal{S}_{2\varepsilon}\setminus \mathcal S_0 = \bigcup_{i=1}^b \left\{ 1-\beta_{\alpha}-2\varepsilon \leq P_i(\bUpsilon) < 1-\beta_{\alpha} \right\}\supseteq \left\{ 1-\beta_{\alpha}-2\varepsilon \leq P_1(\bUpsilon) < 1-\beta_{\alpha} \right\}  
\end{align*}
so, taking probability we deduce that $
\Prob(\mathcal{S}_{2\varepsilon}\setminus \mathcal S_0)\geq 2\varepsilon$
since $P_1(\bUpsilon)$ has uniform$(0,1)$ distribution. 

Therefore, it exists $N$ (depending on all the data points $(D_i)_{i=1}^{\infty}$) such that for all $n\geq N$ it holds
\begin{align*}
    \Prob_D\left(\widetilde P(\widetilde\bUpsilon_n)\in \mathcal{S}_{2\varepsilon}\right) > \alpha + \Prob\left(P(\bUpsilon)\in \mathcal{S}_{2\varepsilon}\setminus \mathcal S_0\right)\geq \alpha+2\varepsilon>\alpha,
\end{align*}
concluding that \cref{eqn:randomgoal120} holds true since $\varepsilon>0$.

\end{proof}

\section{Extra Results for Experiments on  Simulated Data}\label{App:extraexperiments}

\subsection{Censoring}\label{App:censoringExplanation}
For data $A$ and $B$, we generate censoring using an exponential distribution with rate parameter $\gamma$, independent of the group/factors. We consider three types of censoring regimes: low, medium and high, which arise by varying the rate parameter $\gamma$.  \Cref{Table:cen_A}, and \Cref{Table:cen_B} record the censoring percentages, and parameter $\gamma$ chosen for data set A and B, respectively.

\begin{table}[H]
\centering
\begin{tabular}{|l|ccc|}
\hline
&\multicolumn{3}{c|}{Censoring}\\
&Low&Medium&High\\
Hazard&$\gamma=0.1$&$\gamma=0.5$&$\gamma=2$\\
\hline
$\lambda(x)=1$&9\%&33\% &67\%\\
$\lambda(x)=2$&5\%&20\% &50\%\\
\hline
\end{tabular}
\caption{Censoring percentages and rate parameters $\gamma$ for the Data $A$.}\label{Table:cen_A}
\end{table}

\begin{table}[H]
\centering
\begin{tabular}{|l|ccc|}
\hline
&\multicolumn{3}{c|}{Censoring }\\
&Low&Medium&High\\
Hazard&$\gamma=0.1$&$\gamma=0.3$&$\gamma=0.6$\\
\hline
$\lambda(x)=\cos(2x)^2$&16\%&37\% &49\%\\
$\lambda(x)=\sin(2x)^2$&17\%&38\% &61\%\\
$\lambda(x)=1$&9\%&23\% &38\%\\
\hline
\end{tabular}
\caption{Censoring percentages and rate parameters $\gamma$ for the Data $B$.}\label{Table:cen_B}
\end{table}
For data set $C$, the hazard functions that generate censoring for the group $(i,j)$ are given by
\begin{align}\label{eqn:dependentHazardC}
h^c_{1j}(t) = r_c, \quad h^c_{2j}(t) = \frac{1}{2}r_c^{1/2}t^{-1/2}, \quad h^c_{3j}(t) =\frac{3}{2}r_c^{3/2}t^{1/2}
\end{align}
which clearly depends on the factor $\mathcal I$ but not on $\mathcal J$. In the above description the three hazards correspond to: the hazard function of an exponential random variable of rate $r_c$, the hazard of a Weibull with shape $1/2$ and rate $r_c$, and the hazard of a Weibull with shape $3/2$ and rate $r_c$, respectively. The parameter $r_c$ is used to increase/decrease the amount of censoring (the larger $r_c$, the more censoring we get). We consider three parameters, leading to three regimes of censoring: $r_c = 0.1$ for low, $r_c=0.5$ for medium, and $r_c=1$ for high percentage of censored observations. \Cref{Table:cen_D} shows the censoring percentages which are obtained  under the null hypothesis, i.e. for $\theta=0$. 

\begin{table}[H]
\centering

\begin{tabular}{|ll|ccc|}
\hline
\multicolumn{5}{|c|}{Low Censoring $(r_c=0.1)$} \\
\hline
\multicolumn{2}{|c|}{}                                    & \multicolumn{3}{c|}{Factor $\mathcal I$} \\
\multicolumn{2}{|c|}{$\theta=0$}                                         & 1            & 2           & 3           \\ \hline
\multirow{3}{*}{\rotatebox{90}{\footnotesize Factor $\mathcal{J}$}} & 1 &   9\%           &      26\%      &            14\% \\
                                                                    & 2 &     6\%          &    21\% &    5\%       \\
                                                                    & 3 &      5\%                   &    18\%           &            3\% \\ \hline
\end{tabular}
\begin{tabular}{|ll|ccc|}
\hline
\multicolumn{5}{|c|}{Medium Censoring $(r_c=0.5)$} \\
\hline
\multicolumn{2}{|c|}{}                                    & \multicolumn{3}{c|}{Factor $\mathcal I$} \\
\multicolumn{2}{|c|}{$\theta=0$}                                         & 1            & 2           & 3           \\ \hline
\multirow{3}{*}{\rotatebox{90}{\footnotesize Factor $\mathcal{J}$}} & 1 &  $34\%$ & $47\%$ & $56\%$ \\ 
 & 2 &  $26\%$ & $39\%$ & $31\%$ \\
  & 3 &    $21\%$ & $35\%$ & $20\%$ \\ \hline
\end{tabular}
\begin{tabular}{|ll|ccc|}
\hline\multicolumn{5}{|c|}{High Censoring $(r_c=1)$} \\
\hline
\multicolumn{2}{|c|}{
}                                    & \multicolumn{3}{c|}{Factor $\mathcal I$} \\
\multicolumn{2}{|c|}{$\theta=0$}                                         & 1            & 2           & 3           \\ \hline
\multirow{3}{*}{\rotatebox{90}{\footnotesize Factor $\mathcal{J}$}} & 1 &  $51\%$ & $58\%$ & $73\%$ \\
                                                                    & 2 &     $42\%$ & $50\%$ & $51\%$ \\ 
                                                                    & 3 &     $35\%$ & $44\%$ & $38\%$ \\ \hline
\end{tabular}

\caption{Censoring percentages for the data set $C$ under the null hypothesis (i.e., when $\theta = 0$). 
Censoring percentages are the same under the alternative, except for group $(1,2)$ where $\theta>0$ decreases the amount of censoring and $\theta<0$ increases it.} \label{Table:cen_D}
\end{table}

\subsection{Experiments under the Null Hypothesis}\label{appendix:ExperimentsNull}

In this section we show results for experiments under the null hypothesis that were postponed in \Cref{sec:exp_results_H0}. Recall that in the context of data sets $A$ and $B$, we consider the following global null hypothesis $H_0 = \Lambda_1(t) = \Lambda_2(t)$ for all $t\geq 0$, that is, the global null hypothesis states there is no main effect on factor $\mathcal I$. Notice in this case the contrast matrix has only one row so the Multiple Contrast test reduces to a standard test. For data set $C$, the global null hypothesis is that there is no interaction between factors $\mathcal I$ and $\mathcal J$ which holds when the parameter $\theta$ is equal to 0. In this case the contrast matrix has $9$ rows and 9 columns, and so 9 equations are being tested. The Multiple Contrast test considers all 9 equations as local hypotheses, as explained in \Cref{sec:experimentalSetup}. \Cref{table:nullBalanced} and \Cref{table:nullUnbalanced} show the results of our experiments under the null hypothesis with groups of balanced and unbalanced sizes. In our experiments the Type-I error is fixed to value $\alpha = 0.05$.

\begin{landscape}
\begin{table}
\begin{tiny}
\begin{tabular}{|l|l|llllllll|llllllll|llllllll|}
\hline
& & \multicolumn{8}{|c|}{Low Censoring}& \multicolumn{8}{|c|}{Medium Censoring}& \multicolumn{8}{|c|}{Large Censoring}\\
\hline
&SS&K1&K2&K3&K4&K5&P2&P4&M&K1&K2&K3&K4&K5&P2&P4&M&K1&K2&K3&K4&K5&P2&P4&M\\
\hline 
\multirow{9}{*}{\rotatebox{90}{Data A}}&10&0.04&0.035&0.032&0.029&0.033&0.044&0.046&0.042&0.035&0.031&0.033&0.042&0.049&$\boldsymbol{0.055}$&$\boldsymbol{0.051}$&0.042&0.035&0.033&0.041&0.043&0.039&$\boldsymbol{0.051}$&$\boldsymbol{0.053}$&0.045\\
&15&0.039&0.037&0.044&0.045&0.044&0.045&$\boldsymbol{0.052}$&0.044&0.038&0.041&0.042&0.045&0.038&0.047&0.043&0.047&0.033&0.044&0.046&0.049&$\boldsymbol{0.051}$&0.043&$\boldsymbol{0.063}$&$\boldsymbol{0.055}$\\
&20&0.025&0.033&0.049&0.048&0.043&0.046&$\boldsymbol{0.051}$&0.047&0.039&0.035&0.034&0.036&0.04&$\boldsymbol{0.054}$&0.046&0.043&0.029&0.033&0.045&0.049&0.045&0.042&$\boldsymbol{0.051}$&0.042\\
&25&0.033&0.034&0.036&0.041&0.04&0.04&0.046&0.046&0.039&0.039&0.041&0.038&0.047&$\boldsymbol{0.051}$&0.039&0.048&0.044&0.042&0.048&0.047&0.046&0.048&$\boldsymbol{0.052}$&0.047\\
&30&0.035&0.042&0.046&0.042&0.04&0.044&0.047&$\boldsymbol{0.053}$&0.049&$\boldsymbol{0.054}$&0.043&0.044&0.044&$\boldsymbol{0.055}$&$\boldsymbol{0.054}$&$\boldsymbol{0.057}$&$\boldsymbol{0.051}$&0.047&0.047&$\boldsymbol{0.052}$&0.05&0.043&0.044&0.047\\
&35&0.036&0.033&0.048&$\boldsymbol{0.052}$&0.047&0.034&0.045&0.049&0.047&0.042&0.04&0.042&0.037&$\boldsymbol{0.058}$&0.046&0.046&0.043&0.044&0.046&0.048&$\boldsymbol{0.052}$&0.05&0.046&0.049\\
&40&0.039&0.042&0.047&0.046&0.05&0.044&0.05&$\boldsymbol{0.052}$&0.046&0.04&0.041&0.041&0.038&$\boldsymbol{0.057}$&0.045&0.048&0.037&0.036&0.036&0.034&0.032&0.04&0.033&0.041\\
&45&0.042&0.046&0.047&0.039&0.048&0.044&0.05&$\boldsymbol{0.054}$&0.046&0.044&$\boldsymbol{0.053}$&$\boldsymbol{0.057}$&$\boldsymbol{0.051}$&$\boldsymbol{0.052}$&$\boldsymbol{0.057}$&$\boldsymbol{0.052}$&0.039&0.034&0.034&0.036&0.04&0.043&0.043&0.036\\
&50&0.04&0.029&0.039&0.044&0.035&0.04&0.04&0.038&0.05&$\boldsymbol{0.051}$&0.044&0.043&0.04&0.049&0.037&$\boldsymbol{0.056}$&0.046&0.043&0.048&0.044&0.047&0.046&$\boldsymbol{0.052}$&0.048\\
\hline
&SS&K1&K2&K3&K4&K5&P2&P4&M&K1&K2&K3&K4&K5&P2&P4&M&K1&K2&K3&K4&K5&P2&P4&M\\
\hline
\multirow{9}{*}{\rotatebox{90}{Data B}}&10&0.044&0.043&0.04&0.04&0.042&$\boldsymbol{0.055}$&0.043&0.049&0.032&0.036&0.034&0.039&0.044&$\boldsymbol{0.057}$&$\boldsymbol{0.051}$&0.047&0.041&0.048&$\boldsymbol{0.051}$&$\boldsymbol{0.055}$&$\boldsymbol{0.053}$&$\boldsymbol{0.057}$&$\boldsymbol{0.052}$&0.047\\
&15&0.043&0.04&0.043&0.04&0.049&$\boldsymbol{0.053}$&0.048&0.041&0.033&0.031&0.04&0.037&0.043&0.038&0.045&0.039&0.044&0.039&0.043&0.04&0.038&0.043&0.042&0.041\\
&20&0.048&0.047&0.044&0.044&0.048&0.046&0.049&0.05&0.036&0.045&0.045&0.046&$\boldsymbol{0.051}$&0.042&$\boldsymbol{0.051}$&0.048&0.047&0.046&0.043&0.042&0.046&0.042&0.049&$\boldsymbol{0.051}$\\
&25&0.044&0.045&0.047&0.044&0.045&$\boldsymbol{0.065}$&0.05&$\boldsymbol{0.058}$&0.041&0.044&$\boldsymbol{0.054}$&$\boldsymbol{0.051}$&$\boldsymbol{0.052}$&$\boldsymbol{0.051}$&0.044&0.049&0.048&0.05&0.049&$\boldsymbol{0.058}$&$\boldsymbol{0.058}$&$\boldsymbol{0.056}$&0.049&$\boldsymbol{0.06}$\\
&30&0.034&0.038&0.038&0.033&0.036&0.042&0.031&0.038&0.045&$\boldsymbol{0.051}$&$\boldsymbol{0.056}$&0.047&0.046&$\boldsymbol{0.054}$&0.045&$\boldsymbol{0.051}$&0.037&0.03&0.04&0.044&0.042&0.036&0.043&0.033\\
&35&0.041&0.041&0.037&0.037&0.039&$\boldsymbol{0.055}$&0.044&0.039&0.048&0.042&0.045&0.042&0.042&0.042&$\boldsymbol{0.052}$&0.048&$\boldsymbol{0.053}$&$\boldsymbol{0.055}$&0.049&0.046&0.04&$\boldsymbol{0.053}$&0.046&$\boldsymbol{0.055}$\\
&40&0.048&0.041&0.046&0.05&0.049&$\boldsymbol{0.054}$&0.043&0.05&0.037&0.046&0.043&0.043&0.037&0.042&0.049&0.048&0.046&0.049&0.046&0.05&0.05&$\boldsymbol{0.053}$&0.046&0.048\\
&45&0.04&0.035&0.042&0.043&0.05&0.049&0.049&0.044&0.036&0.04&0.045&0.047&0.045&0.037&0.039&$\boldsymbol{0.051}$&0.036&0.034&0.032&0.036&0.035&0.038&0.035&0.036\\
&50&0.046&0.048&0.05&0.047&0.045&$\boldsymbol{0.06}$&0.044&$\boldsymbol{0.052}$&0.041&0.044&0.05&0.044&0.039&0.044&0.039&0.045&$\boldsymbol{0.056}$&$\boldsymbol{0.058}$&0.044&0.048&0.045&$\boldsymbol{0.057}$&0.043&0.05\\
\hline
 &SS&K1&K2&K3&K4&K5&P2&P4&M&K1&K2&K3&K4&K5&P2&P4&M&K1&K2&K3&K4&K5&P2&P4&M\\
\hline
\multirow{9}{*}{\rotatebox{90}{Data C}}&10&0.045&$\boldsymbol{0.054}$&$\boldsymbol{0.056}$&0.049&$\boldsymbol{0.051}$&$\boldsymbol{0.06}$&$\boldsymbol{0.055}$&0.046&0.05&0.046&0.048&0.041&0.036&$\boldsymbol{0.056}$&0.05&0.045&0.035&0.049&0.049&0.041&0.042&$\boldsymbol{0.061}$&0.047&0.038\\
&15&$\boldsymbol{0.051}$&0.049&0.045&0.042&0.046&$\boldsymbol{0.055}$&$\boldsymbol{0.06}$&0.043&0.039&0.045&0.038&0.038&0.033&0.05&$\boldsymbol{0.054}$&0.033&0.038&$\boldsymbol{0.051}$&0.05&$\boldsymbol{0.051}$&$\boldsymbol{0.051}$&$\boldsymbol{0.068}$&$\boldsymbol{0.068}$&0.043\\
&20&0.044&0.041&0.044&$\boldsymbol{0.057}$&$\boldsymbol{0.052}$&$\boldsymbol{0.06}$&$\boldsymbol{0.061}$&0.048&$\boldsymbol{0.059}$&$\boldsymbol{0.053}$&$\boldsymbol{0.051}$&$\boldsymbol{0.056}$&$\boldsymbol{0.056}$&$\boldsymbol{0.057}$&$\boldsymbol{0.055}$&0.04&0.04&0.041&0.042&0.04&0.04&0.039&$\boldsymbol{0.052}$&0.039\\
&25&$\boldsymbol{0.055}$&0.05&$\boldsymbol{0.053}$&0.05&$\boldsymbol{0.051}$&0.05&0.05&0.046&$\boldsymbol{0.051}$&$\boldsymbol{0.056}$&$\boldsymbol{0.055}$&0.05&0.045&$\boldsymbol{0.053}$&$\boldsymbol{0.063}$&0.043&0.046&0.044&0.05&0.041&0.038&$\boldsymbol{0.058}$&$\boldsymbol{0.057}$&0.037\\
&30&0.037&0.042&0.04&0.05&0.05&0.045&$\boldsymbol{0.056}$&0.043&0.041&0.043&0.046&0.045&0.046&0.043&0.049&0.036&0.042&0.046&0.045&0.038&0.034&0.045&0.047&0.048\\
&35&0.038&0.05&$\boldsymbol{0.052}$&$\boldsymbol{0.052}$&$\boldsymbol{0.056}$&$\boldsymbol{0.06}$&$\boldsymbol{0.057}$&$\boldsymbol{0.056}$&0.044&0.043&0.042&0.043&0.049&$\boldsymbol{0.055}$&$\boldsymbol{0.063}$&0.05&0.043&0.047&$\boldsymbol{0.054}$&$\boldsymbol{0.052}$&$\boldsymbol{0.051}$&$\boldsymbol{0.059}$&0.049&0.04\\
&40&$\boldsymbol{0.062}$&$\boldsymbol{0.061}$&$\boldsymbol{0.057}$&0.05&0.05&$\boldsymbol{0.065}$&$\boldsymbol{0.053}$&$\boldsymbol{0.053}$&0.045&$\boldsymbol{0.051}$&0.05&$\boldsymbol{0.052}$&$\boldsymbol{0.052}$&$\boldsymbol{0.051}$&$\boldsymbol{0.064}$&0.041&0.033&0.04&0.041&0.04&0.046&0.038&0.04&0.032\\
&45&$\boldsymbol{0.053}$&0.045&0.045&0.047&0.043&0.043&0.047&0.036&0.047&0.048&0.047&0.046&0.047&$\boldsymbol{0.063}$&$\boldsymbol{0.064}$&0.046&0.047&0.048&0.045&0.044&$\boldsymbol{0.052}$&0.041&0.033&0.043\\
&50&$\boldsymbol{0.059}$&0.047&0.045&0.049&$\boldsymbol{0.051}$&$\boldsymbol{0.06}$&$\boldsymbol{0.053}$&$\boldsymbol{0.054}$&$\boldsymbol{0.061}$&$\boldsymbol{0.06}$&$\boldsymbol{0.055}$&$\boldsymbol{0.059}$&$\boldsymbol{0.059}$&$\boldsymbol{0.06}$&$\boldsymbol{0.054}$&$\boldsymbol{0.056}$&0.041&0.044&0.042&0.05&0.048&$\boldsymbol{0.054}$&0.05&0.034\\
\hline
\end{tabular}
\end{tiny}
\caption{Results for data setting $A$,$B$ and $C$ under the null hypothesis with balanced groups size. The type-I error $\alpha$ is fixed at 0.05. Values in boldface indicate that the power is greater than $\alpha$.}
\label{table:nullBalanced}
\end{table}
\end{landscape}

\begin{landscape}
\begin{table}
\begin{tiny}
\begin{tabular}{|l|l|llllllll|llllllll|llllllll|}
\hline
& & \multicolumn{8}{|c|}{Low Censoring}& \multicolumn{8}{|c|}{Medium Censoring}& \multicolumn{8}{|c|}{Large Censoring}\\
\hline
&MF&K1&K2&K3&K4&K5&P2&P4&M&K1&K2&K3&K4&K5&P2&P4&M&K1&K2&K3&K4&K5&P2&P4&M\\ 
\hline
\multirow{21}{*}{\rotatebox{90}{Data A}}&1&0.033&0.03&0.029&0.033&0.037&0.05&$\boldsymbol{0.054}$&0.039&0.039&0.034&0.038&0.036&0.04&0.047&0.046&0.043&0.034&0.034&0.04&0.036&0.039&0.05&$\boldsymbol{0.063}$&0.041\\
&1.1&0.03&0.026&0.034&0.04&0.046&0.039&0.048&0.035&0.027&0.032&0.037&0.042&0.041&$\boldsymbol{0.052}$&$\boldsymbol{0.059}$&0.038&0.043&0.043&0.038&0.04&0.041&0.045&0.046&0.04\\
&1.2&0.037&0.03&0.02&0.028&0.035&0.042&0.042&0.039&0.032&0.026&0.036&0.034&0.037&0.033&0.039&0.033&0.036&0.028&0.031&0.034&0.031&0.047&0.047&0.039\\
&1.3&0.029&0.024&0.034&0.042&0.039&0.05&0.045&0.035&0.038&0.036&0.033&0.039&0.044&0.049&0.041&0.038&0.039&0.033&0.041&0.043&0.045&$\boldsymbol{0.052}$&$\boldsymbol{0.066}$&0.042\\
&1.4&0.031&0.032&0.035&0.042&0.045&0.05&0.044&0.042&0.035&0.035&0.038&0.037&0.037&$\boldsymbol{0.053}$&$\boldsymbol{0.059}$&0.046&0.031&0.03&0.037&0.041&0.045&0.039&0.05&0.036\\
&1.5&0.037&0.034&0.037&0.037&0.031&0.048&$\boldsymbol{0.066}$&0.039&0.033&0.034&0.034&0.035&0.039&0.041&$\boldsymbol{0.052}$&0.04&0.029&0.033&0.035&0.037&0.034&0.042&$\boldsymbol{0.054}$&0.04\\
&1.6&0.033&0.031&0.035&0.036&0.03&$\boldsymbol{0.052}$&0.04&0.036&0.041&0.038&0.041&0.04&0.043&0.039&0.048&0.042&0.042&0.039&0.036&0.037&0.047&0.04&0.049&0.045\\
&1.7&0.035&0.038&0.041&0.041&0.04&0.046&0.049&0.042&0.027&0.032&0.033&0.03&0.034&0.041&0.045&0.045&0.041&0.038&0.035&0.031&0.037&0.041&$\boldsymbol{0.059}$&0.047\\
&1.8&0.029&0.028&0.027&0.031&0.038&0.034&0.05&0.03&0.043&0.041&0.035&0.033&0.039&0.044&0.038&0.038&0.035&0.039&0.047&0.044&0.041&$\boldsymbol{0.051}$&$\boldsymbol{0.062}$&0.047\\
&1.9&0.041&0.047&0.043&0.045&0.039&$\boldsymbol{0.055}$&$\boldsymbol{0.057}$&$\boldsymbol{0.051}$&0.032&0.036&0.041&0.04&0.042&0.048&$\boldsymbol{0.056}$&0.044&0.031&0.032&0.036&0.042&0.045&0.044&0.046&0.042\\
&2&0.025&0.029&0.04&0.04&0.039&$\boldsymbol{0.051}$&0.042&0.04&0.046&0.036&0.046&0.047&0.05&0.04&0.042&0.046&0.045&0.036&0.036&0.031&0.034&0.04&0.046&0.046\\
&2.1&0.04&0.045&0.045&0.047&0.041&$\boldsymbol{0.051}$&$\boldsymbol{0.051}$&$\boldsymbol{0.055}$&0.038&0.029&0.037&0.042&0.045&0.036&0.04&0.035&0.035&0.036&0.04&0.036&0.038&0.044&0.042&0.041\\
&2.2&0.039&0.036&0.034&0.042&0.038&0.04&0.047&0.042&0.041&0.041&0.042&0.046&$\boldsymbol{0.052}$&0.047&0.039&0.043&0.042&0.039&0.039&0.038&0.036&0.048&0.049&0.042\\
&2.3&0.048&$\boldsymbol{0.054}$&0.037&0.038&0.032&$\boldsymbol{0.051}$&0.047&$\boldsymbol{0.056}$&0.036&0.031&0.041&0.044&0.04&0.041&0.049&0.038&0.04&0.037&0.049&0.05&0.05&0.05&0.047&0.043\\
&2.4&0.041&0.033&0.038&0.042&0.039&0.045&$\boldsymbol{0.056}$&0.046&0.04&0.044&0.05&$\boldsymbol{0.052}$&0.044&$\boldsymbol{0.051}$&$\boldsymbol{0.06}$&0.046&0.048&0.043&0.047&0.045&0.043&$\boldsymbol{0.057}$&$\boldsymbol{0.053}$&$\boldsymbol{0.051}$\\
&2.5&0.047&0.042&0.039&0.043&0.045&0.046&0.043&0.04&0.032&0.028&0.034&0.039&0.033&0.032&0.031&0.031&0.036&0.043&0.048&0.047&0.041&0.047&0.042&0.046\\
&2.6&$\boldsymbol{0.051}$&$\boldsymbol{0.06}$&$\boldsymbol{0.051}$&$\boldsymbol{0.053}$&$\boldsymbol{0.055}$&$\boldsymbol{0.062}$&$\boldsymbol{0.062}$&$\boldsymbol{0.07}$&0.04&0.038&0.034&0.031&0.035&0.034&0.04&0.04&0.04&0.045&0.045&0.046&0.049&0.045&0.049&0.042\\
&2.7&0.043&0.042&0.033&0.033&0.032&0.043&0.04&0.038&0.033&0.042&0.04&0.04&0.041&0.042&0.047&0.044&0.023&0.025&0.03&0.025&0.029&0.038&0.028&0.029\\
&2.8&0.041&0.035&0.04&0.04&0.041&$\boldsymbol{0.051}$&0.05&0.049&0.043&0.044&0.048&0.043&0.042&$\boldsymbol{0.056}$&$\boldsymbol{0.054}$&$\boldsymbol{0.055}$&0.036&0.033&0.033&0.031&0.031&0.038&0.043&0.038\\
&2.9&0.036&0.039&0.035&0.039&0.044&0.043&0.042&0.037&0.037&0.03&0.032&0.039&0.044&0.034&0.046&0.033&0.039&0.028&0.029&0.026&0.032&0.04&$\boldsymbol{0.055}$&0.036\\
&3&$\boldsymbol{0.056}$&0.049&0.048&0.044&0.044&0.047&0.039&$\boldsymbol{0.054}$&$\boldsymbol{0.051}$&0.04&0.041&0.042&0.042&0.038&0.035&0.041&0.029&0.029&0.037&0.035&0.035&0.042&$\boldsymbol{0.051}$&0.033\\
\hline
&SS&K1&K2&K3&K4&K5&P2&P4&M&K1&K2&K3&K4&K5&P2&P4&M&K1&K2&K3&K4&K5&P2&P4&M\\
\hline
\multirow{21}{*}{\rotatebox{90}{Data B}}&1&0.042&0.035&0.033&0.036&0.038&$\boldsymbol{0.051}$&0.048&0.047&0.033&0.028&0.03&0.034&0.037&0.04&0.04&0.031&0.042&0.032&0.036&0.038&0.044&$\boldsymbol{0.053}$&0.046&0.041\\
&1.1&0.032&0.033&0.038&0.04&0.041&0.04&0.036&0.034&0.026&0.021&0.029&0.032&0.033&0.045&0.033&0.029&0.031&0.039&0.044&0.049&0.043&$\boldsymbol{0.052}$&$\boldsymbol{0.059}$&0.048\\
&1.2&0.034&0.034&0.033&0.035&0.038&0.04&0.039&0.044&0.029&0.029&0.023&0.027&0.035&0.045&$\boldsymbol{0.053}$&0.03&0.038&0.033&0.034&0.036&0.045&0.042&0.045&0.039\\
&1.3&0.026&0.027&0.031&0.039&0.04&0.04&$\boldsymbol{0.055}$&0.04&0.042&0.035&0.035&0.036&0.034&$\boldsymbol{0.051}$&$\boldsymbol{0.054}$&0.049&0.046&0.039&0.038&0.033&0.039&0.044&$\boldsymbol{0.055}$&0.045\\
&1.4&0.035&0.037&0.043&0.037&0.042&$\boldsymbol{0.052}$&$\boldsymbol{0.051}$&$\boldsymbol{0.053}$&0.032&0.026&0.031&0.03&0.03&0.033&0.045&0.025&0.038&0.034&0.045&0.037&0.037&0.045&$\boldsymbol{0.052}$&0.041\\
&1.5&0.032&0.031&0.034&0.035&0.04&0.044&0.047&0.038&0.038&0.038&0.042&0.038&0.038&0.044&0.04&0.044&0.033&0.031&0.038&0.038&0.043&0.047&$\boldsymbol{0.055}$&0.035\\
&1.6&0.036&0.033&0.039&0.036&0.036&0.043&$\boldsymbol{0.051}$&0.039&0.035&0.034&0.034&0.034&0.043&$\boldsymbol{0.051}$&0.044&0.042&0.043&0.039&0.038&0.042&0.042&0.046&0.049&0.034\\
&1.7&0.029&0.033&0.04&0.042&0.04&0.049&0.05&0.048&0.035&0.03&0.036&0.032&0.039&0.04&0.041&0.04&0.035&0.031&0.037&0.036&0.034&0.044&0.047&0.04\\
&1.8&0.043&0.046&0.047&0.044&0.046&0.05&0.042&$\boldsymbol{0.051}$&0.045&0.033&0.04&0.036&0.043&$\boldsymbol{0.051}$&0.041&0.043&0.039&0.04&0.037&0.039&0.034&$\boldsymbol{0.061}$&$\boldsymbol{0.056}$&$\boldsymbol{0.051}$\\
&1.9&0.032&0.028&0.034&0.036&0.038&0.033&0.04&0.034&0.036&0.043&0.042&0.043&0.039&0.043&0.047&0.044&0.048&0.042&0.043&0.042&0.036&$\boldsymbol{0.055}$&0.047&$\boldsymbol{0.055}$\\
&2&$\boldsymbol{0.051}$&0.043&0.042&0.043&$\boldsymbol{0.053}$&$\boldsymbol{0.051}$&0.05&$\boldsymbol{0.052}$&0.042&0.049&$\boldsymbol{0.051}$&$\boldsymbol{0.054}$&$\boldsymbol{0.064}$&$\boldsymbol{0.057}$&$\boldsymbol{0.064}$&$\boldsymbol{0.057}$&0.038&0.042&0.043&0.044&0.043&$\boldsymbol{0.057}$&$\boldsymbol{0.057}$&0.042\\
&2.1&0.036&0.03&0.044&0.048&$\boldsymbol{0.052}$&$\boldsymbol{0.063}$&$\boldsymbol{0.06}$&0.05&0.031&0.037&0.038&0.04&0.047&$\boldsymbol{0.051}$&$\boldsymbol{0.057}$&0.047&0.036&0.04&0.044&0.041&0.038&0.042&$\boldsymbol{0.051}$&0.045\\
&2.2&0.032&0.03&0.033&0.038&0.039&0.05&0.05&0.04&0.04&0.04&0.043&0.037&0.04&0.049&$\boldsymbol{0.053}$&0.048&0.043&0.041&0.038&0.032&0.034&0.047&0.038&0.045\\
&2.3&0.043&0.037&0.039&0.038&0.04&0.042&0.045&0.049&0.043&0.038&0.04&0.039&0.041&0.05&$\boldsymbol{0.053}$&0.038&0.049&0.044&$\boldsymbol{0.051}$&0.048&0.048&0.05&$\boldsymbol{0.051}$&0.05\\
&2.4&0.04&0.037&0.036&0.039&0.038&0.037&0.047&0.036&0.035&0.033&0.036&0.041&0.043&0.041&0.049&0.036&0.045&0.043&0.047&0.048&0.048&0.046&$\boldsymbol{0.051}$&0.039\\
&2.5&0.037&0.041&0.036&0.033&0.032&0.036&0.047&0.04&0.042&0.037&0.034&0.034&0.031&0.045&0.044&0.041&0.024&0.026&0.038&0.042&0.042&0.037&0.045&0.04\\
&2.6&0.038&0.028&0.05&$\boldsymbol{0.052}$&0.047&0.042&0.035&$\boldsymbol{0.051}$&0.038&0.037&0.04&0.042&0.043&0.042&0.04&0.043&0.031&0.032&0.036&0.041&0.041&0.044&0.039&0.036\\
&2.7&0.049&$\boldsymbol{0.052}$&$\boldsymbol{0.052}$&$\boldsymbol{0.053}$&$\boldsymbol{0.059}$&$\boldsymbol{0.051}$&$\boldsymbol{0.054}$&$\boldsymbol{0.061}$&0.039&0.036&0.039&0.041&0.048&0.042&$\boldsymbol{0.053}$&0.041&0.044&0.047&$\boldsymbol{0.058}$&$\boldsymbol{0.058}$&$\boldsymbol{0.059}$&$\boldsymbol{0.052}$&$\boldsymbol{0.063}$&$\boldsymbol{0.057}$\\
&2.8&0.041&0.041&0.044&0.046&0.044&$\boldsymbol{0.052}$&$\boldsymbol{0.054}$&$\boldsymbol{0.056}$&0.038&0.042&0.041&0.041&0.044&$\boldsymbol{0.055}$&0.048&0.048&0.038&0.041&0.048&0.048&$\boldsymbol{0.051}$&$\boldsymbol{0.053}$&$\boldsymbol{0.051}$&$\boldsymbol{0.052}$\\
&2.9&0.042&0.046&0.048&0.041&0.048&0.048&0.05&$\boldsymbol{0.051}$&0.046&0.04&0.034&0.034&0.032&0.043&0.045&0.045&0.049&0.043&0.044&0.045&0.049&$\boldsymbol{0.059}$&0.046&0.044\\
&3&0.034&0.033&0.033&0.043&0.044&0.039&0.036&0.04&0.05&0.047&$\boldsymbol{0.055}$&0.05&0.044&0.049&$\boldsymbol{0.052}$&$\boldsymbol{0.055}$&0.036&0.036&0.037&0.038&0.038&$\boldsymbol{0.054}$&0.035&0.046\\
\hline
 &SS&K1&K2&K3&K4&K5&P2&P4&M&K1&K2&K3&K4&K5&P2&P4&M&K1&K2&K3&K4&K5&P2&P4&M\\
\hline
\multirow{21}{*}{\rotatebox{90}{Data C}}&1&0.042&0.035&0.033&0.036&0.038&$\boldsymbol{0.051}$&0.048&0.047&0.033&0.028&0.03&0.034&0.037&0.04&0.04&0.031&0.042&0.032&0.036&0.038&0.044&$\boldsymbol{0.053}$&0.046&0.041\\
&1.1&0.032&0.033&0.038&0.04&0.041&0.04&0.036&0.034&0.026&0.021&0.029&0.032&0.033&0.045&0.033&0.029&0.031&0.039&0.044&0.049&0.043&$\boldsymbol{0.052}$&$\boldsymbol{0.059}$&0.048\\
&1.2&0.034&0.034&0.033&0.035&0.038&0.04&0.039&0.044&0.029&0.029&0.023&0.027&0.035&0.045&$\boldsymbol{0.053}$&0.03&0.038&0.033&0.034&0.036&0.045&0.042&0.045&0.039\\
&1.3&0.026&0.027&0.031&0.039&0.04&0.04&$\boldsymbol{0.055}$&0.04&0.042&0.035&0.035&0.036&0.034&$\boldsymbol{0.051}$&$\boldsymbol{0.054}$&0.049&0.046&0.039&0.038&0.033&0.039&0.044&$\boldsymbol{0.055}$&0.045\\
&1.4&0.035&0.037&0.043&0.037&0.042&$\boldsymbol{0.052}$&$\boldsymbol{0.051}$&$\boldsymbol{0.053}$&0.032&0.026&0.031&0.03&0.03&0.033&0.045&0.025&0.038&0.034&0.045&0.037&0.037&0.045&$\boldsymbol{0.052}$&0.041\\
&1.5&0.032&0.031&0.034&0.035&0.04&0.044&0.047&0.038&0.038&0.038&0.042&0.038&0.038&0.044&0.04&0.044&0.033&0.031&0.038&0.038&0.043&0.047&$\boldsymbol{0.055}$&0.035\\
&1.6&0.036&0.033&0.039&0.036&0.036&0.043&$\boldsymbol{0.051}$&0.039&0.035&0.034&0.034&0.034&0.043&$\boldsymbol{0.051}$&0.044&0.042&0.043&0.039&0.038&0.042&0.042&0.046&0.049&0.034\\
&1.7&0.029&0.033&0.04&0.042&0.04&0.049&0.05&0.048&0.035&0.03&0.036&0.032&0.039&0.04&0.041&0.04&0.035&0.031&0.037&0.036&0.034&0.044&0.047&0.04\\
&1.8&0.043&0.046&0.047&0.044&0.046&0.05&0.042&$\boldsymbol{0.051}$&0.045&0.033&0.04&0.036&0.043&$\boldsymbol{0.051}$&0.041&0.043&0.039&0.04&0.037&0.039&0.034&$\boldsymbol{0.061}$&$\boldsymbol{0.056}$&$\boldsymbol{0.051}$\\
&1.9&0.032&0.028&0.034&0.036&0.038&0.033&0.04&0.034&0.036&0.043&0.042&0.043&0.039&0.043&0.047&0.044&0.048&0.042&0.043&0.042&0.036&$\boldsymbol{0.055}$&0.047&$\boldsymbol{0.055}$\\
&2&$\boldsymbol{0.051}$&0.043&0.042&0.043&$\boldsymbol{0.053}$&$\boldsymbol{0.051}$&0.05&$\boldsymbol{0.052}$&0.042&0.049&$\boldsymbol{0.051}$&$\boldsymbol{0.054}$&$\boldsymbol{0.064}$&$\boldsymbol{0.057}$&$\boldsymbol{0.064}$&$\boldsymbol{0.057}$&0.038&0.042&0.043&0.044&0.043&$\boldsymbol{0.057}$&$\boldsymbol{0.057}$&0.042\\
&2.1&0.036&0.03&0.044&0.048&$\boldsymbol{0.052}$&$\boldsymbol{0.063}$&$\boldsymbol{0.06}$&0.05&0.031&0.037&0.038&0.04&0.047&$\boldsymbol{0.051}$&$\boldsymbol{0.057}$&0.047&0.036&0.04&0.044&0.041&0.038&0.042&$\boldsymbol{0.051}$&0.045\\
&2.2&0.032&0.03&0.033&0.038&0.039&0.05&0.05&0.04&0.04&0.04&0.043&0.037&0.04&0.049&$\boldsymbol{0.053}$&0.048&0.043&0.041&0.038&0.032&0.034&0.047&0.038&0.045\\
&2.3&0.043&0.037&0.039&0.038&0.04&0.042&0.045&0.049&0.043&0.038&0.04&0.039&0.041&0.05&$\boldsymbol{0.053}$&0.038&0.049&0.044&$\boldsymbol{0.051}$&0.048&0.048&0.05&$\boldsymbol{0.051}$&0.05\\
&2.4&0.04&0.037&0.036&0.039&0.038&0.037&0.047&0.036&0.035&0.033&0.036&0.041&0.043&0.041&0.049&0.036&0.045&0.043&0.047&0.048&0.048&0.046&$\boldsymbol{0.051}$&0.039\\
&2.5&0.037&0.041&0.036&0.033&0.032&0.036&0.047&0.04&0.042&0.037&0.034&0.034&0.031&0.045&0.044&0.041&0.024&0.026&0.038&0.042&0.042&0.037&0.045&0.04\\
&2.6&0.038&0.028&0.05&$\boldsymbol{0.052}$&0.047&0.042&0.035&$\boldsymbol{0.051}$&0.038&0.037&0.04&0.042&0.043&0.042&0.04&0.043&0.031&0.032&0.036&0.041&0.041&0.044&0.039&0.036\\
&2.7&0.049&$\boldsymbol{0.052}$&$\boldsymbol{0.052}$&$\boldsymbol{0.053}$&$\boldsymbol{0.059}$&$\boldsymbol{0.051}$&$\boldsymbol{0.054}$&$\boldsymbol{0.061}$&0.039&0.036&0.039&0.041&0.048&0.042&$\boldsymbol{0.053}$&0.041&0.044&0.047&$\boldsymbol{0.058}$&$\boldsymbol{0.058}$&$\boldsymbol{0.059}$&$\boldsymbol{0.052}$&$\boldsymbol{0.063}$&$\boldsymbol{0.057}$\\
&2.8&0.041&0.041&0.044&0.046&0.044&$\boldsymbol{0.052}$&$\boldsymbol{0.054}$&$\boldsymbol{0.056}$&0.038&0.042&0.041&0.041&0.044&$\boldsymbol{0.055}$&0.048&0.048&0.038&0.041&0.048&0.048&$\boldsymbol{0.051}$&$\boldsymbol{0.053}$&$\boldsymbol{0.051}$&$\boldsymbol{0.052}$\\
&2.9&0.042&0.046&0.048&0.041&0.048&0.048&0.05&$\boldsymbol{0.051}$&0.046&0.04&0.034&0.034&0.032&0.043&0.045&0.045&0.049&0.043&0.044&0.045&0.049&$\boldsymbol{0.059}$&0.046&0.044\\
&3&0.034&0.033&0.033&0.043&0.044&0.039&0.036&0.04&0.05&0.047&$\boldsymbol{0.055}$&0.05&0.044&0.049&$\boldsymbol{0.052}$&$\boldsymbol{0.055}$&0.036&0.036&0.037&0.038&0.038&$\boldsymbol{0.054}$&0.035&0.046\\
\hline
\end{tabular}

\end{tiny}
\caption{Results for data setting $A$,$B$ and $C$ under the null hypothesis with unbalanced groups size. The type-I error $\alpha$ is fixed at 0.05. Values in boldface indicate that the power is greater than $\alpha$.}
\label{table:nullUnbalanced}
\end{table}
\end{landscape}

\subsection{Extra Experiments under the Alternative Hypothesis}\label{Appendix:ExperimentsAlternative}

We show the results of our experiments under the alternative hypothesis in the balanced setting that were deferred from \Cref{sec:experiments_alternative}.

\subsubsection{Data $A$ and $B$}\label{Appendix:ExperimentsAlternativeAB} \Cref{Fig:DataAE1-bal,Fig:DataBEI-bal} show the postponed results for the balanced setting.

\begin{figure}[H]
\centering
\includegraphics[scale=\size]{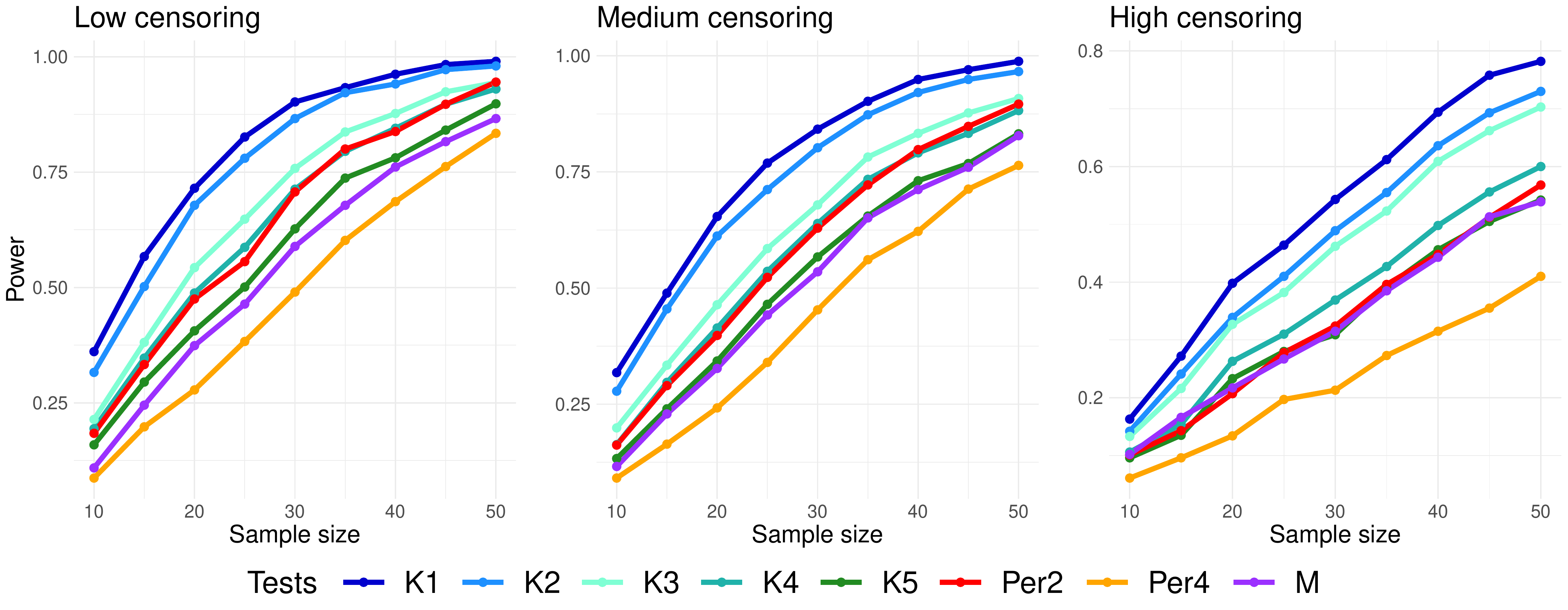}\caption{Test power versus sample size for the hypothesis there is no  effect of the factor $\mathcal{I}$ for Data $A$ in the balanced case. }\label{Fig:DataAE1-bal}
\end{figure}

\begin{figure}[H]
\centering
\includegraphics[scale=\size]{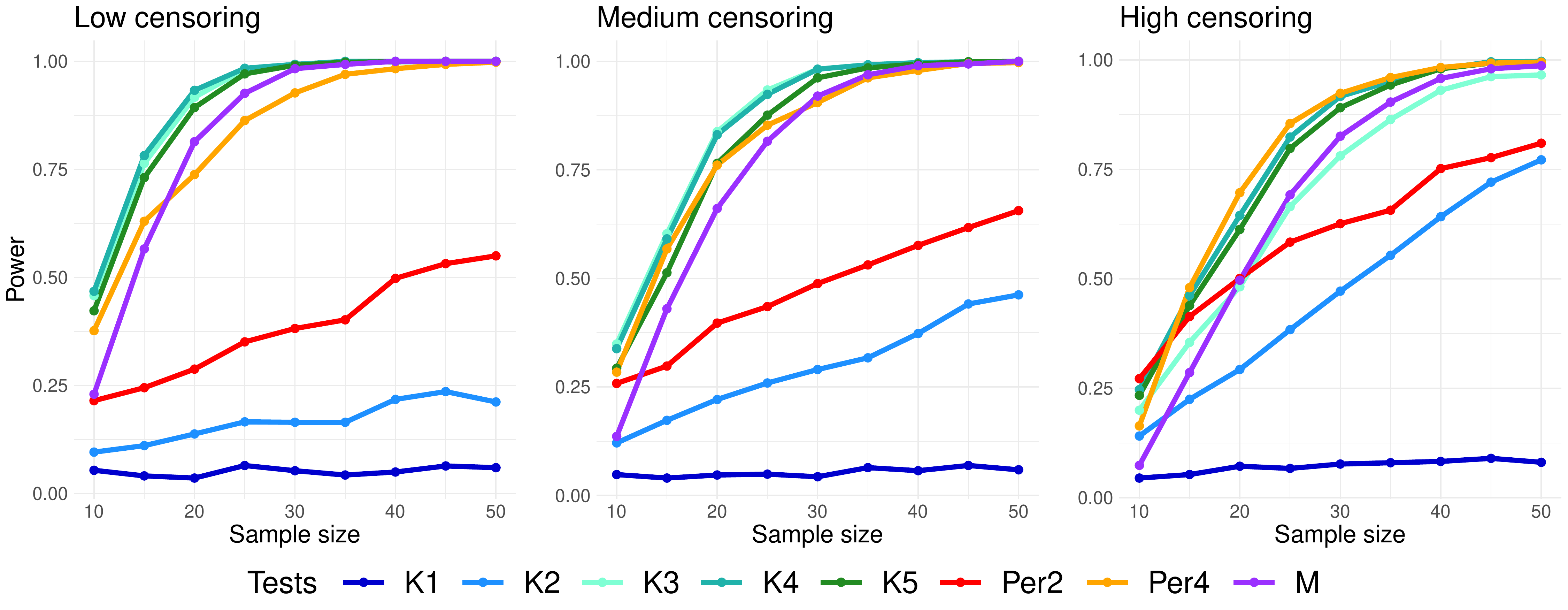}\caption{Test power versus sample size for the hypothesis there is no  effect of the factor $\mathcal{I}$ for Data $B$ in the balanced case.}\label{Fig:DataBEI-bal}
\end{figure}

In \Cref{Fig:DataAE1-bal,Fig:DataBEI-bal} we observe a similar behaviour than for the unbalanced setting. Indeed, note that for the data set $A$ the best results are obtained for kernels with a large length-scale parameter, whereas for the data set $B$, the best results are attained by kernels with a small length-scale parameter.

We continue by presenting the results for the Multiple Contrast testing procedure. \Cref{BarplotDataA_Bal,BarplotDataB_Bal} show the results in the balanced data setting that was postponed from \Cref{{sec:experiment_marginal}}. In these results, we observe that the Multiple Contrast test behaves similarly to what we observe in the unbalanced setting, and in particular, we can see how with more data points the test can recognise that both hypotheses $H_{01}$ and $H_{02}$ are false. For the data set $A$ the power of the test is explained  among the alternatives i)only $H_{01}$ is false, ii) only $H_{02}$ is false, and iii) only $H_{01}$ and $H_{02}$ are false at the same time. We notice that with a small sample size the first two options explain most of the test-power, but as the sample size increases the test starts realising that both hypotheses are false at the same time.  For the data set  $B$ we see that the test quickly realises that both hypothesis are false at the same time.

\begin{figure}[h]
\centering
\includegraphics[scale=\size]{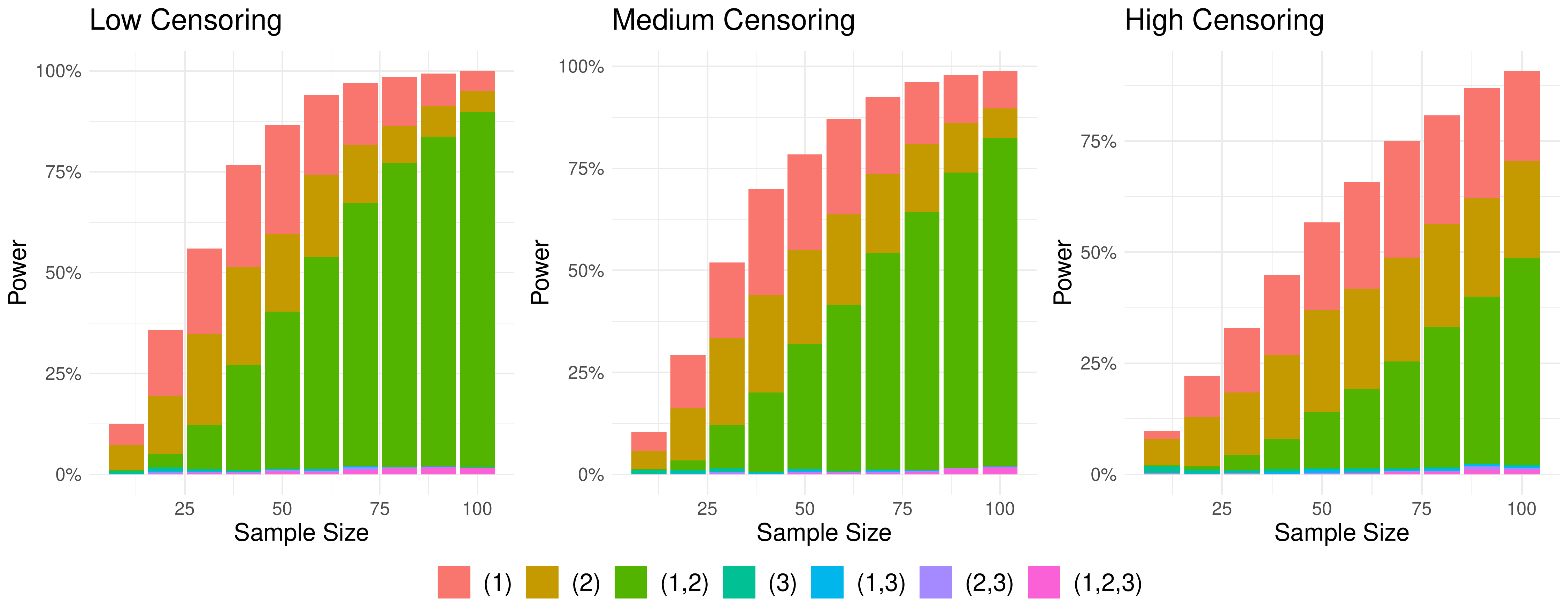}\caption{Distribution of the power attained by the Multiple Contrast test when testing the hypothesis that there is no  effect of the factor $\mathcal{I}$ for Data $A$ in the balanced case.}\label{BarplotDataA_Bal}
\end{figure}

\begin{figure}[h]
\centering
\includegraphics[scale=\size]{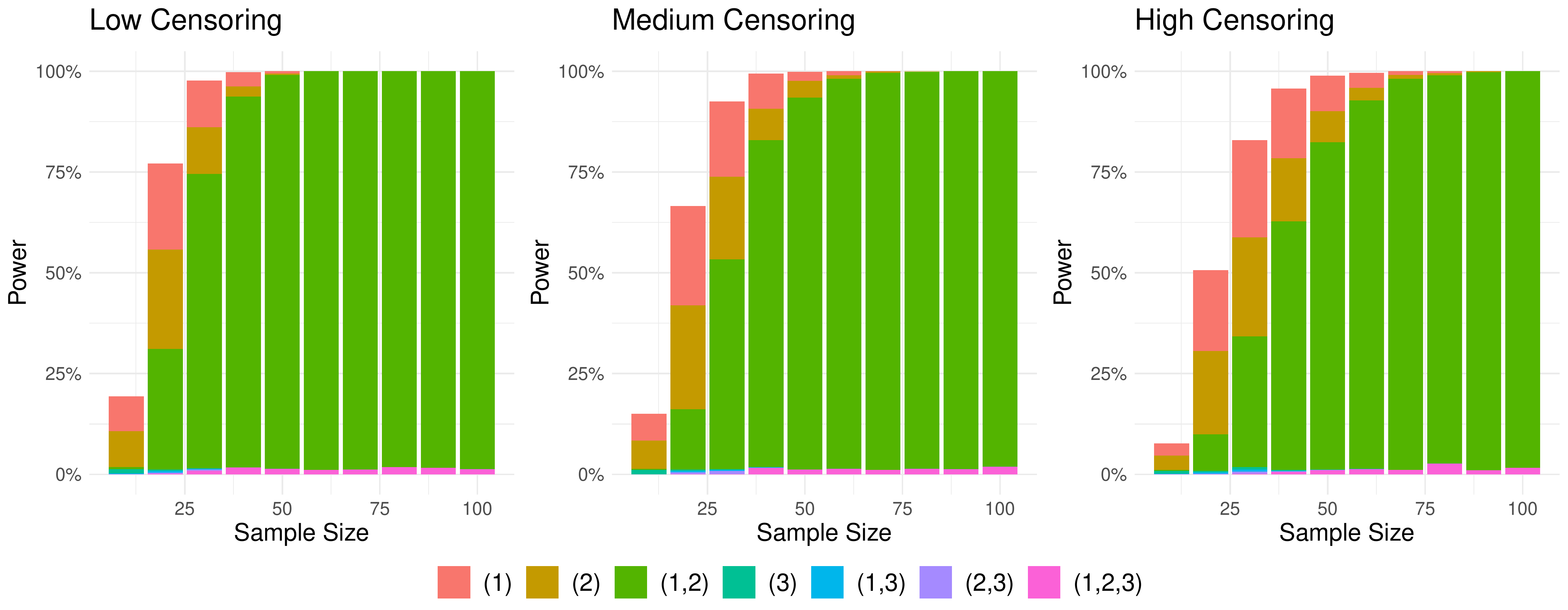}\caption{Distribution of the power attained by the Multiple Contrast test when testing the hypothesis that there is no  effect of the factor $\mathcal{I}$ for Data $B$ in the balanced case.}\label{BarplotDataB_Bal}
\end{figure}

\subsubsection{Data $C$}\label{Appendix:ExperimentsAlternativeC}

We show the results for data set $C$. In this case we consider two different alternatives, the first one for $\theta = 1$ and the second one for $\theta=2$. We begin  describing the results for $\theta=2$ in the balanced setting as the respective results for the unbalanced setting were presented in \Cref{sec:experiments_alternative}.

\begin{figure}[H]
\centering
\includegraphics[scale=\size]{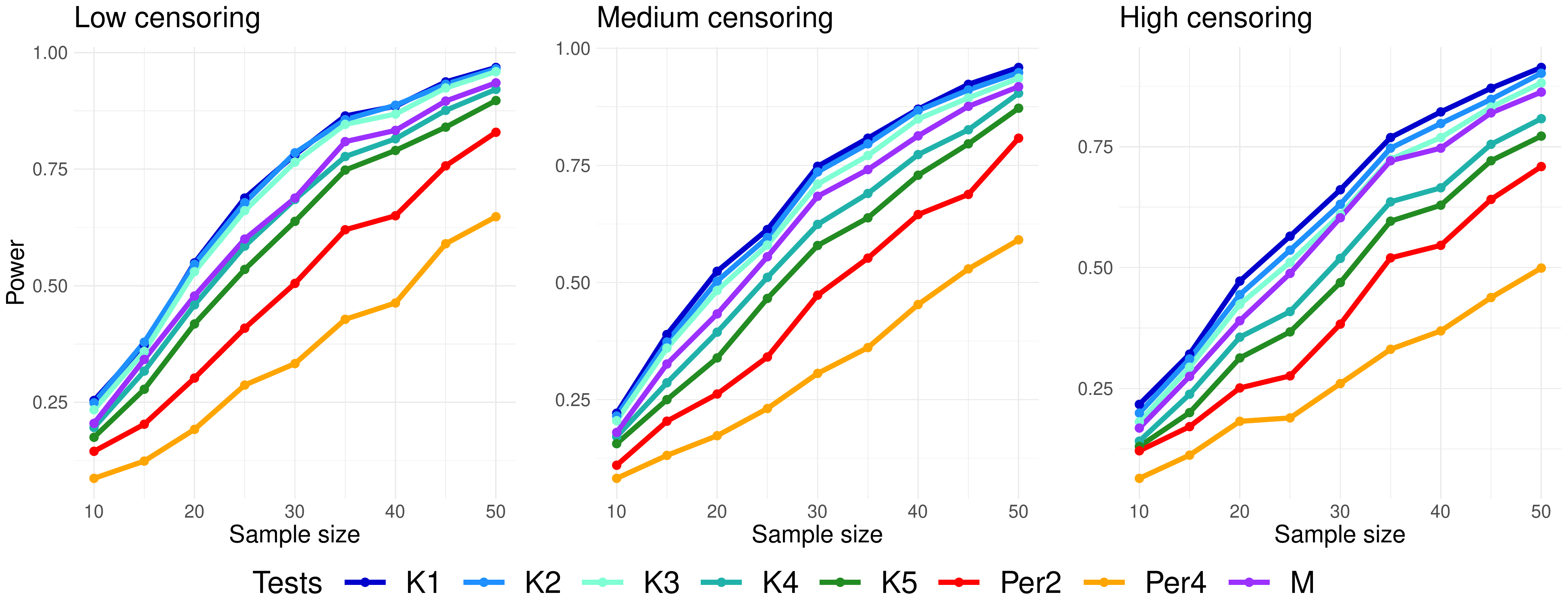}\caption{Test power versus sample size for the hypothesis there is no  interaction term in Data $C$ when $\theta = 2$ (equivalent to test $\theta = 0$) in the balanced case.}\label{Fig:DataCt2Bal}
\end{figure}

In \Cref{Fig:DataCt2Bal} we observe the rejection rate for each test as the sample size increases from 10 to 50 data points for each group. Notice the behaviour of all test in this scenario is very similar to the one we observed in the unbalanced data scenario -indeed, the tests keep the same order in terms of rejection rates- but overall, all tests increase their rejection rate for a fixed sample size. This is somehow expected as the balanced data scenario is much simpler than the unbalanced case. A detailed evaluation of the Multiple Contrast test is shown in \Cref{BarplotDataC_t2_Bal}. To understand this figure, recall that the contrast matrix in this setting is composed of 9 local hypotheses, which are all false under the alternative hypothesis, which is exactly our case as $\theta=2$.  Thus \Cref{BarplotDataC_t2_Bal} records the number of local hypotheses that are being rejected for each sample size. Similarly to what occurs in the unbalanced data scenario, we observe that the test starts rejecting more hypotheses as the sample size grows, note however that in this case the growth is faster as the problem is easier in the balanced data scenario. Lastly, observe that, overall, this seems to be an expensive data problem as in neither scenario (low, medium, or high censoring), the test is able to confidently deduce that all 9 hypotheses must be rejected. 

\begin{figure}[h]
\centering
\includegraphics[scale=\size]{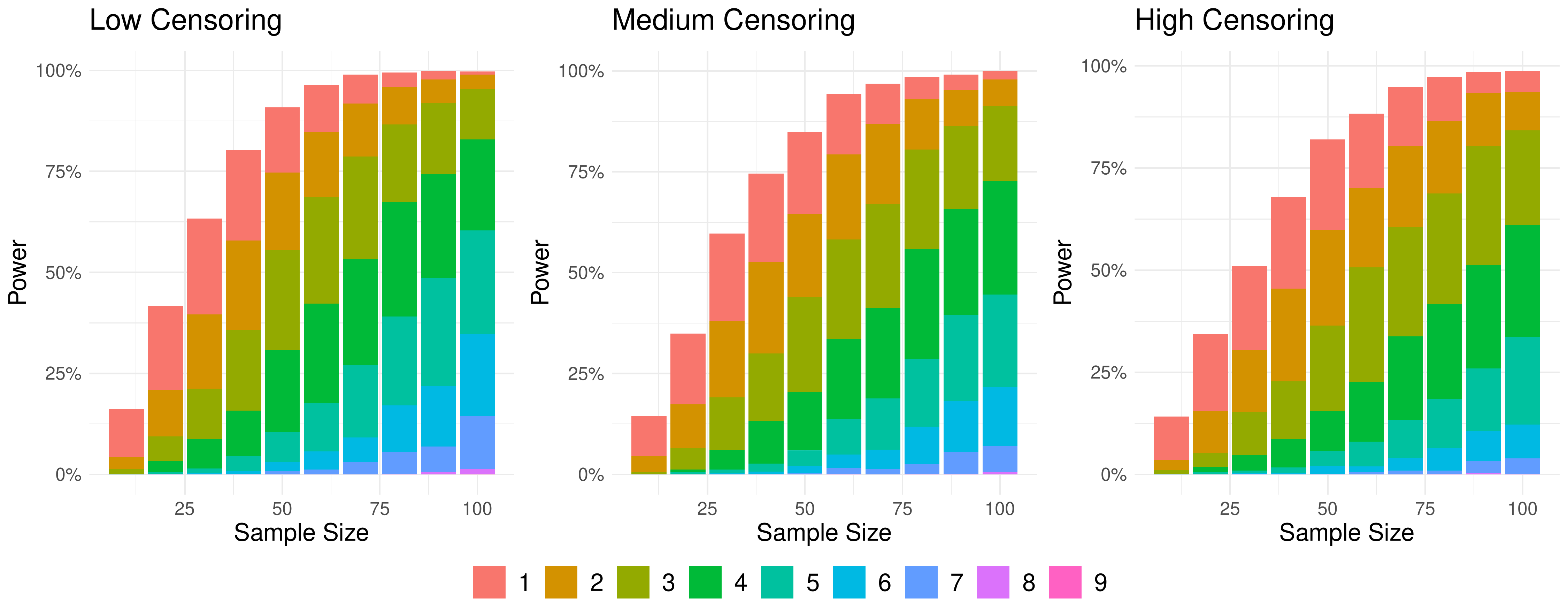}\caption{Distribution of the power attained by the Multiple Contrast test when testing the hypothesis that there is no  interaction term in Data  $C$ when $\theta = 2$, in the unbalanced case. }\label{BarplotDataC_t2_Bal}
\end{figure}

We proceed to show results for $\theta=1$  which were postponed in \Cref{sec:experiments_alternative}. Note that in this scenario, rejecting the null hypothesis should be harder than when compared to the case in which $\theta=2$, as $\theta=1$ is closer to the null hypothesis (recall the null hypothesis is recovered when $\theta=0$). This is indeed what we observe in our simulations -presented in \Cref{Fig:DataCt1Bal} and \Cref{Fig:DataCt1Unb} for the balanced and unbalanced data scenarios, respectively- where we observe a clear drop in the rejection rate (power of the test) in all censoring scenarios (low, medium or high). Notwithstanding this, we observe that the overall order of the tests in terms of their rejection rates remains the same, and in particular, the best performance is attained by the the kernel test with largest length-scale parameter. \Cref{BarplotDataC_t1} and \Cref{BarplotDataC_t1_Unb} show the behaviour of the Multiple Contrast test for the balanced and unbalance data-settings, respectively. Note that here we also can observe a drop in the power of the test as previously described, and in particular, we can observe that the tests overall are rejecting less hypotheses when compared to the case in which $\theta=2$

\begin{figure}[H]
\centering
\includegraphics[scale=\size]{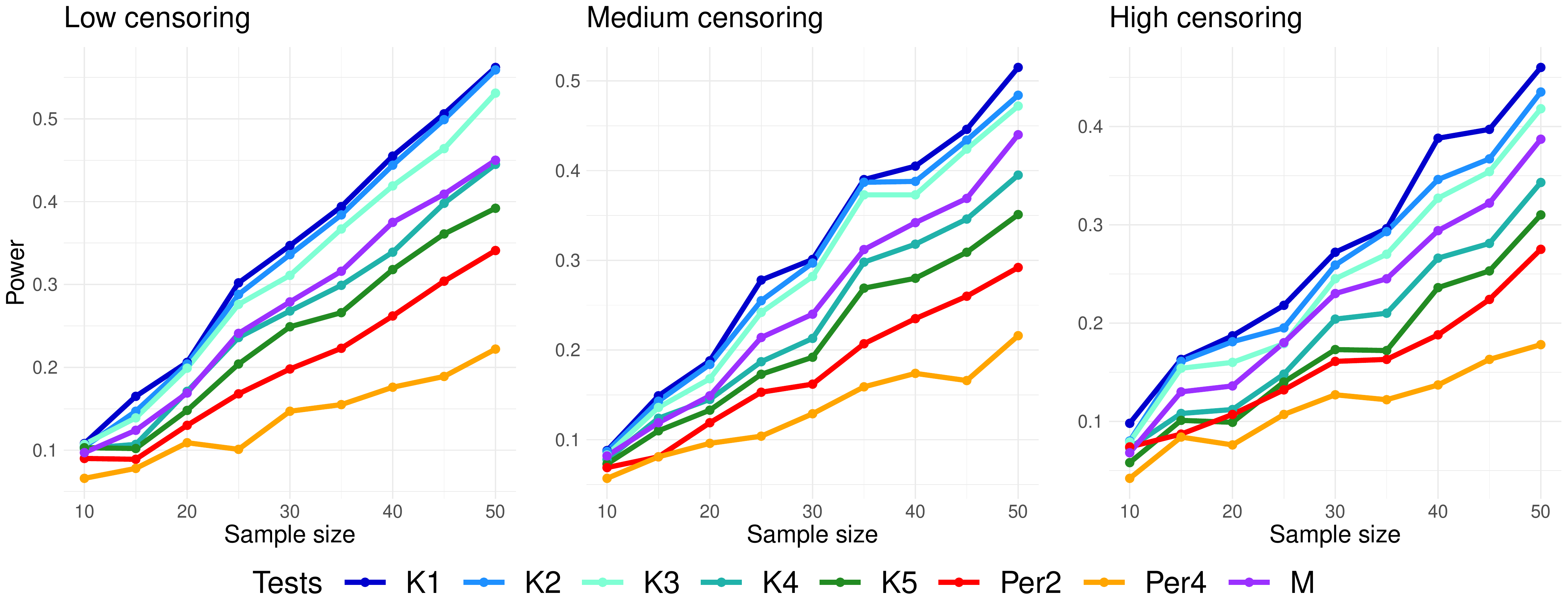}\caption{Test power versus sample size for the hypothesis there is no  interaction term in Data $C$ when $\theta = 1$ (equivalent to test $\theta = 0$) in the balanced case.}\label{Fig:DataCt1Bal}
\end{figure}

\begin{figure}[H]
\centering
\includegraphics[scale=\size]{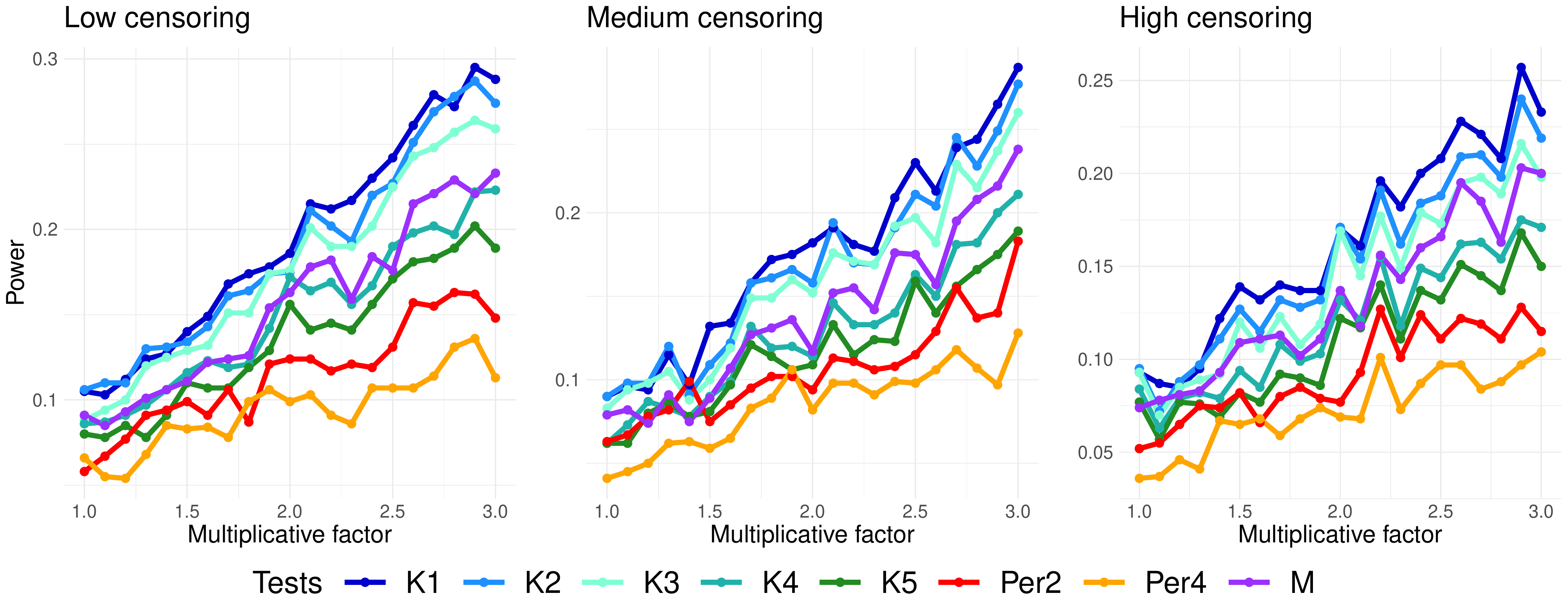}\caption{Test power versus sample size for the hypothesis there is no  interaction term in Data $C$ when $\theta = 1$ (equivalent to test $\theta = 0$) in the unbalanced case.}\label{Fig:DataCt1Unb}
\end{figure}

\begin{figure}[h]
\centering
\includegraphics[scale=\size]{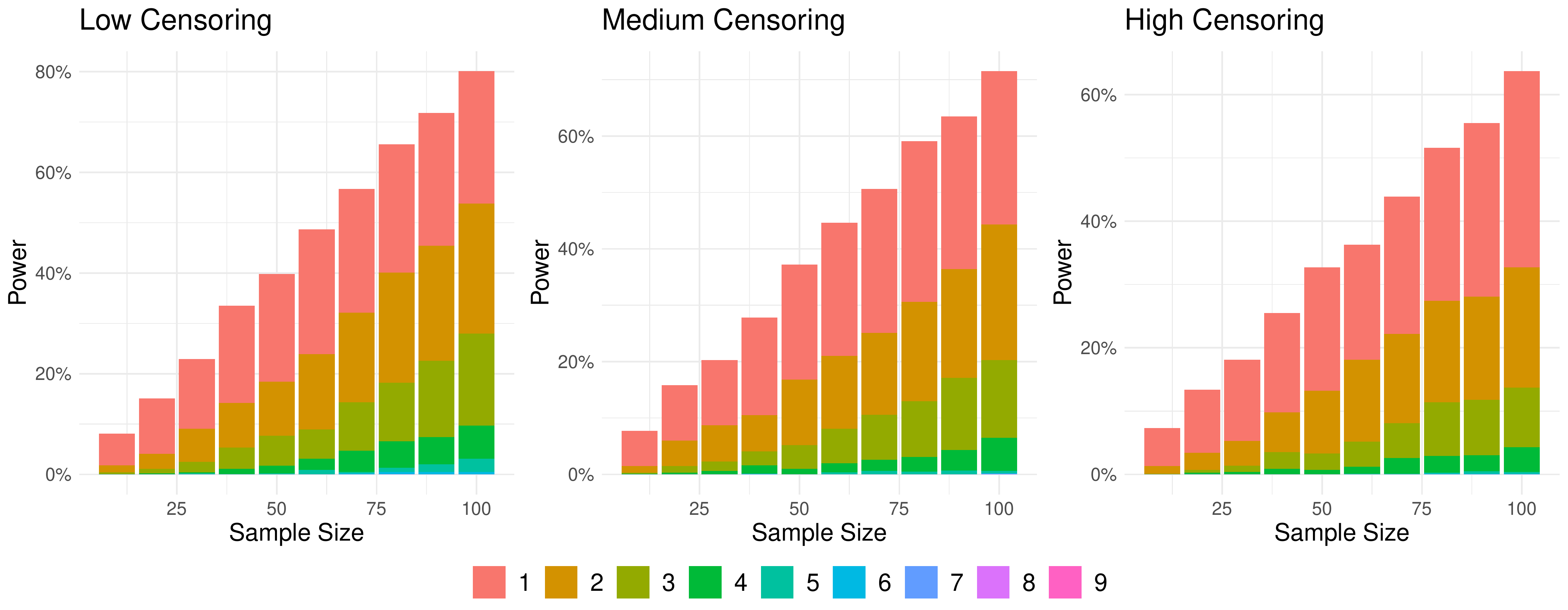}\caption{Distribution of the power attained by the Multiple Contrast test when testing the hypothesis that there is no  interaction term in Data  $C$ when $\theta = 1$, in the balanced case. }\label{BarplotDataC_t1}
\end{figure}

\begin{figure}[h]
\centering
\includegraphics[scale=\size]{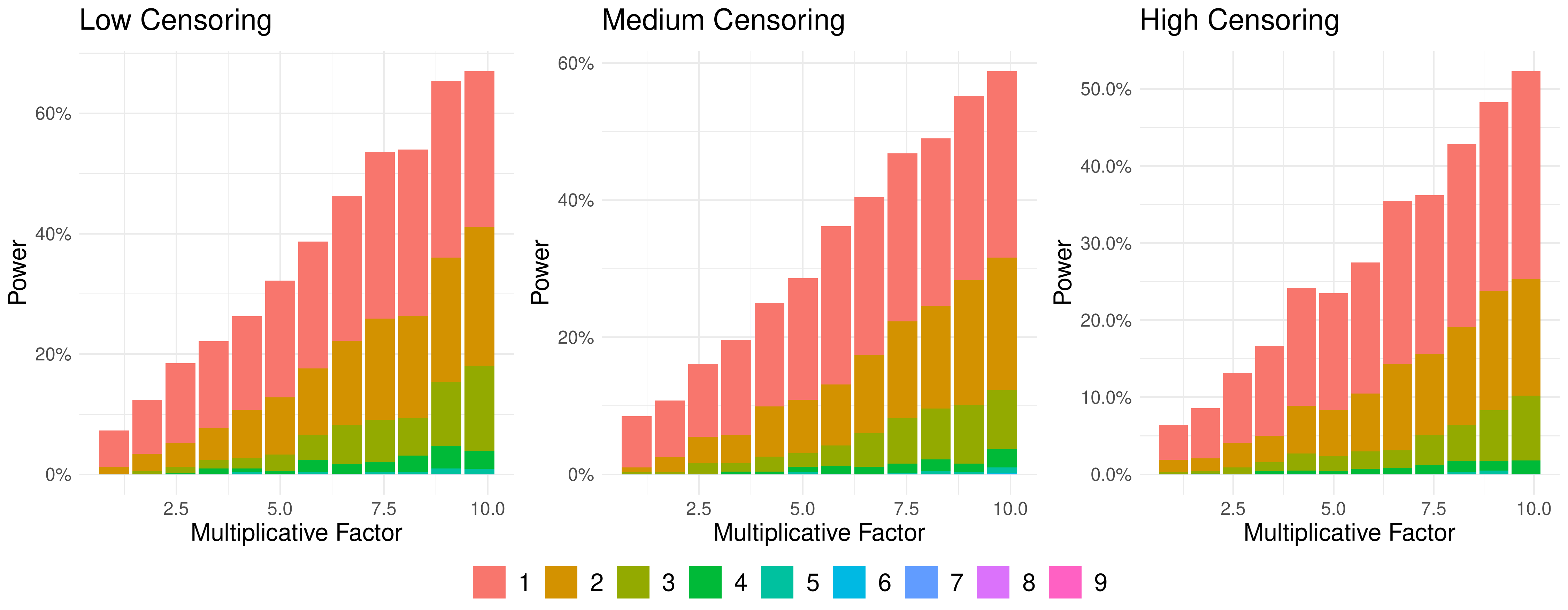}\caption{Distribution of the power attained by the Multiple Contrast test when testing the hypothesis that there is no  interaction term in Data  $C$ when $\theta = 1$, in the unbalanced case.}\label{BarplotDataC_t1_Unb}
\end{figure}

We finish this section by presenting \Cref{Fig:DataCtvarBal} which studies the behaviour of all the tests when small deviations from the null hypothesis occur, that is, then $\theta$ approaches 0. In particular, \Cref{Fig:DataCtvarBal} presents results for the balanced data setting which were postponed in \Cref{sec:experiments_alternative}.

\begin{figure}[H]
\centering
\includegraphics[scale=\size]{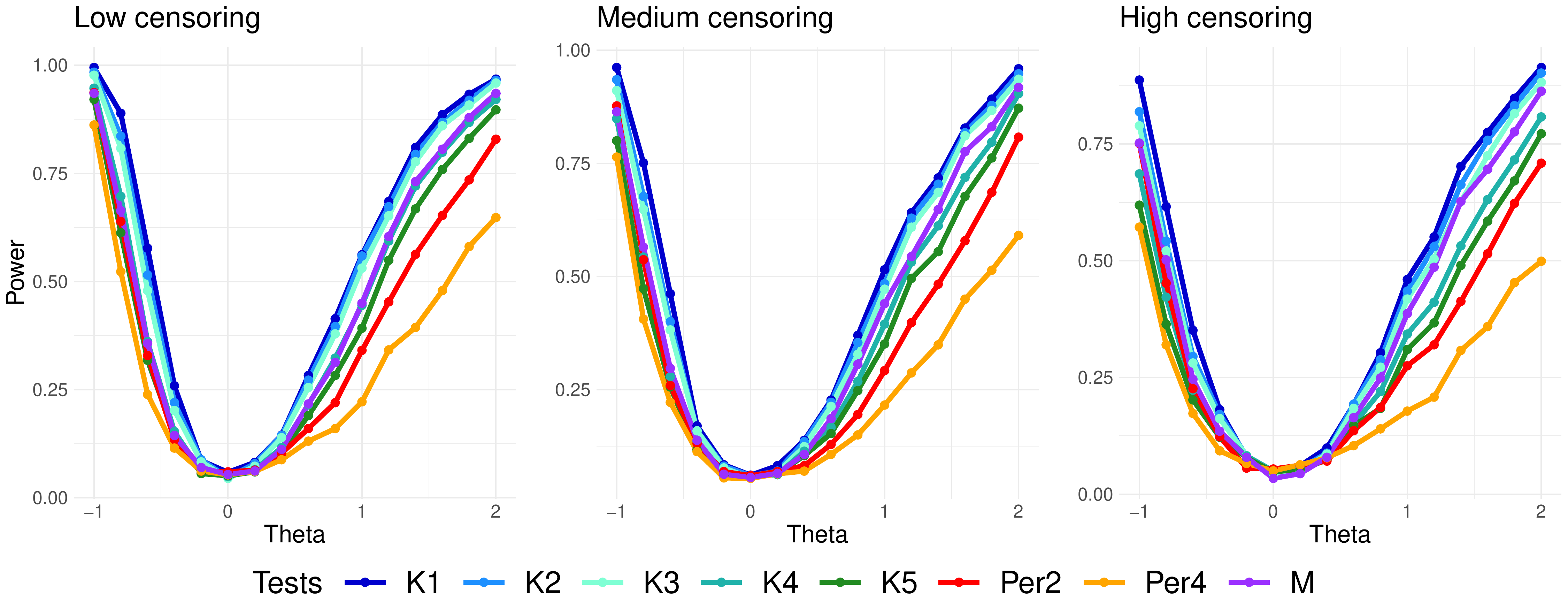}\caption{Tests power vs $\theta$ for the hypothesis that there is no interaction term in Data $C$ in the balanced case.}\label{Fig:DataCtvarBal}
\end{figure}

From the previous figures we observe that all tests achieve the correct Type-I error when $\theta=0$, and that all kernel tests perform better than the CASANOVA procedure, even the Multiple Contrast test which we know is more data expensive. When comparing kernel tests, we observe that the best results are achived by larger length-scale parameters (K1 to K3), which seems reasonable since the hazard functions that were used to generate the data are rather smooths and does not have huge fluctuations.

\newpage

\bibliographystyle{plainnat}
\bibliography{ref}

\end{document}